\documentclass{lmcs}

\keywords{Quantum Computing, Semantics of Programming Languages, Coherent Control}

\usepackage{hyperref}
\usepackage{amsmath}
\usepackage{amssymb}
\usepackage{stmaryrd}
\usepackage{mathtools}
\usepackage{mathrsfs}
\usepackage{ebproof}
\usepackage{xfrac}
\usepackage{subcaption}
\usepackage{tikzit}

\tikzstyle{texthere}=[inline text, font={\footnotesize}]
\tikzstyle{small box}=[rectangle, inline text, fill=white, draw, minimum height=7.5mm, yshift=-0.5mm, minimum width=7.5mm, font={\small}]

\tikzstyle{arrow}=[->]

\input{diagrams.tikzdefs}

\DeclarePairedDelimiter\bra{\langle}{\rvert}
\DeclarePairedDelimiter\ket{\lvert}{\rangle}
\DeclarePairedDelimiterX\braket[2]{\langle}{\rangle}{#1 \delimsize\vert #2}
\newcommand{\Hilb}{\mathcal{H}}
\newcommand{\tb}[1]{\text{\textnormal{\textbf{\texttt{#1}}}}}
\newcommand{\sem}[1]{\llbracket #1 \rrbracket}
\newcommand{\p}{\mathtt{p}}
\newcommand{\q}{\mathtt{q}}
\newcommand{\qr}{\mathtt{r}}
\newcommand{\qc}{\mathtt{c}}
\newcommand{\qt}{\mathtt{t}}
\newcommand{\pfin}[1]{\mathcal{P}_{\textup{fin}}(#1)}
\newcommand{\smx}[1]{\left(\begin{smallmatrix}#1\end{smallmatrix}\right)\!}
\newcommand{\op}[1]{{\left\vert\kern-0.25ex\left\vert\kern-0.25ex\left\vert #1
    \right\vert\kern-0.25ex\right\vert\kern-0.25ex\right\vert}}

\begin{document}

\title[Quantum Control and General Recursion Beyond the Unitary Case]{Quantum Control and General Recursion Beyond the Unitary Case}
\titlecomment{{\lsuper*}
This is an extended version of a LICS 2026 paper~\cite{bpp}, with full proofs and additional examples.
}
\thanks{This work is supported by the Plan France 2030 through the PEPR integrated project EPiQ (ANR-22-PETQ-0007) and the HQI platform (ANR-22-PNCQ-0002); by the European Union through the MSCA Staff Exchanges project QCOMICAL (Grant Agreement ID: 101182520); and by the Maison du Quantique MaQuEst.}

\specialissue{}

\author[K.~Barsse]{Kathleen Barsse\lmcsorcid{0009-0007-6398-5395}}
\author[R.~Péchoux]{Romain Péchoux\lmcsorcid{0000-0003-0601-5425}}
\author[S.~Perdrix]{Simon Perdrix\lmcsorcid{0000-0002-1808-2409}}

\address{Université de Lorraine, CNRS, Inria, LORIA, F-54000 Nancy, France}
\email{kathleen.barsse@inria.fr, romain.pechoux@loria.fr, simon.perdrix@loria.fr}


\begin{abstract}
  Coherent control, \emph{aka} quantum control, is a central concept in quantum computing that is attracting increasing attention
from both the quantum foundations and quantum software communities.
Defining coherent control in the presence of recursion and measurement has long been known to be a major challenge. In particular, no-go results have been established for standard semantical domains like completely positive maps. We address this problem by introducing the first quantum programming language with recursion that allows for the coherent control of arbitrary quantum operations. We equip this language with both an operational  and a denotational semantics that we prove to be adequate. To design these semantics, we  show that combining coherent control, recursion, and measurement crucially requires describing the evolution of  subprograms in the absence of input.
To address this, the operational semantics takes into account a default evolution branch, while the denotational semantics uses the concept of coherent quantum operation, based on vacuum extensions.
We strengthen the validity of our approach by developing an observational equivalence: two programs are equivalent if their probability of termination is the same in any context.
The denotational semantics is shown to be fully abstract with respect to this observational equivalence.
\end{abstract}

\maketitle

\section{Introduction}

\subsection{Context and Motivations}

In quantum computing, the execution flow of operations acting on quantum states is generally carried out in a fixed or classically controlled manner, based on information such as the classical outcomes of quantum measurements.
This is the ``quantum data, classical control'' paradigm, which is known to enable a speedup in comparison to classical algorithms. One step further, one can also consider the superposition of quantum processes themselves, depending on a quantum state. This notion is referred to as \emph{coherent control} or \emph{quantum control}. A fundamental example of coherent control is the \emph{quantum case}.  This process can be thought of as a quantum version of the classical ``\tb{if}'' statement where the two branches can run in a superposition.
In the case of unitary evolutions,
the quantum case acts on a pair of control and target qubits as follows: if the control qubit is in state $\ket{0}$, a unitary $U$ is applied to the target qubit, if the control qubit is in state $\ket{1}$, a unitary $V$ is applied to the target qubit, otherwise $U$ and $V$ are applied in a coherent superposition depending on the state of the control qubit. Namely,  the quantum case  is defined on unitary maps as $(U,V)\mapsto \ket{0}\bra{0}\otimes U +\ket{1}\bra{1}\otimes V$.

In comparison to quantum processes whose execution flow is fixed or classically controlled, models with coherent control are known to provide some advantages in computational complexity~\cite{unitary2012,unitary2014,araujo,kristjansson_exponential} and communication tasks~\cite{ebler,grenoble_coherent_control,kristjansson}.
For this reason, the possibility of including coherent control in quantum programming languages is attracting a great deal of interest (see~\cite{valiron} for an overview).

The main challenge is to combine  \emph{quantum control}  with the key features of a quantum programming language: \emph{general recursion}, which refers to the notion of partial recursive functions in computability theory; and  the capability to handle not only unitary evolutions but also \emph{measurements}. Notice that there are already solid contributions that combine two of these three features, which we briefly review below.

Without quantum control, it is well-known that completely positive maps, or quantum channels, can be used to represent both unitary evolutions and measurements.
Moreover, Selinger \cite{QPL} showed that, when equipped with the Löwner partial ordering, completely positive trace non-increasing maps form a $\omega$-CPO, thereby enabling the use of domain-theoretic fixed‑point techniques for general recursion.
Nevertheless, incorporating quantum control into this framework is far from straightforward,
as the quantum case operation is incompatible with the Löwner ordering \cite{badescu}.

To circumvent these no-go results, a compelling line of research considers pure evolutions only, which entails dropping measurement capabilities and adopting quantum-controlled recursion. Indeed, whereas recursion usually relies on measurements to determine whether recursive calls should be made, one can consider quantum controlled recursion in which recursive calls are performed in superposition.
However, suitable definitions of quantum recursion, such as a quantum while loop, involve not only considering systems of infinite dimension, but also defining the outcomes of programs as limits of infinite computations~\cite{ying_recursion,ying_book,qloop2025}. This makes its practical significance questionable, as for instance the notion of termination is unclear in this context.

Finally, in the absence of general recursion, combining quantum control with measurement raises the issue of controlling general quantum operations (beyond the unitary case), which is non trivial.
Indeed, although, as we have just seen,  the quantum case  is well defined on unitary maps, its extension to completely positive maps is ill-defined~\cite{badescu}.
In order to support control of non unitary operations, it is necessary to choose a different semantic domain.
This is achieved in~\cite{ying_control_flow,ying_alternation,badescu}, with semantics relying on Kraus decompositions.
In recent years, the question of coherently controlling arbitrary quantum operations (in the absence of general recursion) has been intensively studied in foundations of quantum mechanics, where it is pointed out that to define coherent control of quantum operations, one must specify additional information on the implementation of each controlled channel~\cite{kristjansson,grenoble_coherent_control}. Various frameworks have been introduced to this end, including \emph{vacuum-extended operations}~\cite{kristjansson,kristjansson_1}, \emph{routed channels}~\cite{V21} and Kraus decompositions with input environment states~\cite{grenoble_coherent_control}.

In summary,  the development of a full-fledged programming language that combines quantum control with general recursion and measurement has remained an open problem for over a decade, as existing frameworks supporting general recursion lack either quantum control or measurement. Solving it is nevertheless essential to the development of high-level quantum programming languages in their full generality.

\subsection{Contributions}

In this paper, we solve the above longstanding problem by providing the first semantically well-behaved programming language (Figure~\ref{fig:syntax}) allowing for measurements, recursion, and quantum control beyond the unitary case.
Quantum control is formulated as a quantum case primitive called ``\tb{qcase}'', which we present as a quantum version of the classical conditional statement.
The novelty of our programming language is the possibility for both coherent control of arbitrary quantum processes (i.e., not restricted to unitary maps) and recursion via a classically controlled while loop.

In order to assign a formal meaning to programs, we construct an operational semantics, consisting of a probabilistic big-step transition system, and a denotational semantics. Despite the apparent issues surrounding the definition of the quantum case~\cite{badescu}, this operation is physically meaningful.
Intuitively, it requires some additional information that describes the behavior of the program in the absence of inputs (see Figure~\ref{fig:qc}).

In the operational semantics, this additional information takes the form of a default transition for each probabilistic statement
(Figure~\ref{fig:semantics_operational}). In the denotational semantics, the additional information is expressed using the framework of \emph{coherent quantum operations} (Definition~\ref{def:cqo} and Figure~\ref{fig:semantics_denotational}), which take inspiration from vacuum-extended operations~\cite{kristjansson,kristjansson_1}. Coherent quantum operations are shown to form a pointed DCPO (Proposition~\ref{prop:cqoDCPO}), allowing for domain-theoretic fixed‑point techniques for general recursion in the presence of quantum control.
Hence, this enables us to define the semantics of the quantum case in a physically meaningful way, without limiting which operations are considered controllable,  while ensuring the usual expected properties such as compositionality.

We prove the following main properties of the language:
\begin{itemize}
\item Universality (Theorem~\ref{thm:universality}):
every coherent quantum operation can be implemented by a program in the language. In particular, every quantum operation can be realized via a coherent quantum operation.
Moreover, approximate universality is proved in the case where the set of built-in unitary gates is restricted to the Hadamard gate $H$ and the $T$-gate (Proposition~\ref{prop:approx_universality}).
\item Adequacy (Theorem~\ref{thm:adequacy}): the independently-defined operational and denotational semantics are adequate, meaning that they describe the same behavior of programs. Consequently, this links the default transition approach of the operational semantics to the coherent quantum operation approach of the denotational semantics.
\item Full abstraction (Theorem~\ref{thm:fullabstraction}): we introduce a notion of observational equivalence (Definition~\ref{def:observational_equiv}), which states that two programs are observationally equivalent if their probability of termination is the same whatever the context. This formalizes the idea that the observable behavior of two programs is indistinguishable. We show that the denotational semantics  is fully abstract for observational equivalence. Hence, the denotational semantics exactly captures the observable behavior of programs.
\end{itemize}
All these concepts are illustrated by numerous examples scattered throughout the paper, giving an idea of how expressive the language can be.

\subsection{Overview}

We now give a brief overview in order to convey the intuition and explain our design choices. We introduce a quantum programming language with single-qubit unitaries, quantum measurement, qubit discarding, and while loops, as well as a primitive for quantum control, written as $\tb{qcase}$ (quantum case).
Its syntax can be viewed as that of QPL~\cite{QPL} extended by a quantum case. Given statements $S_0$ and $S_1$,
\[
 \tb{qcase}\ \q\ (0\rightarrow S_0,1\rightarrow S_1)
\]
is the statement that executes $S_0$ and $S_1$ in a coherent superposition, depending on the state of the control qubit $\q$. Importantly,  qubit $\q$ is not measured during this process. This ensures that the resulting process is a superposition of the two possible statements rather than a probabilistic mixture. The standard measurement operation is instead written with the $\tb{meas}$ primitive:
\[
 \tb{meas}\ \q\ (0\rightarrow S_0,1\rightarrow S_1).
\]
The control qubit $\q$ is measured in the computational basis: if the outcome is $0$, then $S_0$ is executed, and if the outcome is $1$, then $S_1$ is executed.
For both $\tb{qcase}$ and $\tb{meas}$, neither of the statements $S_0$ or $S_1$ is required to be unitary. Regarding recursion, we include a $\tb{while}$ primitive whose syntax is the following:
\[
 \tb{while}\ \q\ \tb{do}\ S.
\]
$\tb{while}$ is a measurement-based loop: with each iteration, the control qubit q is measured in the computational basis, and the loop is exited once output 0 is obtained.
Integrating loops into a quantum programming language brings up considerations that are not encountered on lower-level models such as quantum circuits.
For example, the principle of deferred measurement~\cite{nielsen}, which states that all measurements can be postponed to the end of the computation, no longer holds in the presence of classically controlled while loops.

The operational semantics defines the evolution of an input state $\ket{\psi}$ induced by a program. To explain the semantics of $\tb{qcase}$, consider the statement $\tb{qcase}\ \q \ (0\rightarrow S_0,1\rightarrow S_1)$. We begin by examining the base cases. If the input has the form $\ket{0}_\q \otimes \ket{\psi}$, since the control qubit is in state $\ket{0}$, the target input $\ket{\psi}$ is passed to $S_0$, and $S_1$ does not receive any input, as illustrated by Figure~\ref{sfig0}.
\begin{figure}[!h]
\begin{subfigure}[t]{0.35\textwidth}
\[\begin{tikzpicture}
	\begin{pgfonlayer}{nodelayer}
		\node [style=small box, fill=gray!20] (0) at (0, 1) {$S_0$};
		\node [style=small box, fill=gray!20] (1) at (0, -1) {$S_1$};
		\node [style=none] (2) at (-0.75, 1) {};
		\node [style=none] (3) at (-1.75, 1) {};
		\node [style=none] (4) at (-1.75, -1) {};
		\node [style=none] (5) at (-0.75, -1) {};
		\node [style=none] (6) at (0.75, -1) {};
		\node [style=none] (7) at (1.75, -1) {};
		\node [style=none] (8) at (0.75, 1) {};
		\node [style=none] (9) at (1.75, 1) {};
		\node [style=none] (10) at (-0.75, 1) {};
		\node [style=none] (11) at (-3.5, 1.5) {};
		\node [style=none] (12) at (-3.5, -1.5) {};
		\node [style=texthere] (13) at (-2.5, 1) {$|\psi\rangle$};
		\node [style=texthere] (14) at (-2.75, -1) {$\substack{\text{empty}\\ \text{input}}$};
		\node [style=none] (15) at (3.5, 1.5) {};
		\node [style=none] (16) at (3.5, -1.5) {};
	\end{pgfonlayer}
	\begin{pgfonlayer}{edgelayer}
		\draw [style=arrow] (3.center) to (10.center);
		\draw [style=arrow] (4.center) to (5.center);
		\draw [style=arrow] (6.center) to (7.center);
		\draw [style=arrow] (8.center) to (9.center);
	\end{pgfonlayer}
\end{tikzpicture}\]
 \caption{No input on $S_1$}
        \label{sfig0}
\end{subfigure}
\begin{subfigure}[t]{0.35\textwidth}
\[\begin{tikzpicture}
	\begin{pgfonlayer}{nodelayer}
		\node [style=small box, fill=gray!20] (0) at (0, 1) {$S_0$};
		\node [style=small box, fill=gray!20] (1) at (0, -1) {$S_1$};
		\node [style=none] (2) at (-0.75, 1) {};
		\node [style=none] (3) at (-1.75, 1) {};
		\node [style=none] (4) at (-1.75, -1) {};
		\node [style=none] (5) at (-0.75, -1) {};
		\node [style=none] (6) at (0.75, -1) {};
		\node [style=none] (7) at (1.75, -1) {};
		\node [style=none] (8) at (0.75, 1) {};
		\node [style=none] (9) at (1.75, 1) {};
		\node [style=none] (10) at (-0.75, 1) {};
		\node [style=none] (11) at (-3.5, 1.5) {};
		\node [style=none] (12) at (-3.5, -1.5) {};
		\node [style=texthere] (13) at (-2.75, 1) {$\substack{\text{empty}\\ \text{input}}$};
		\node [style=texthere] (14) at (-2.5, -1) {$|\psi\rangle$};
		\node [style=none] (15) at (3.5, 1.5) {};
		\node [style=none] (16) at (3.5, -1.5) {};
	\end{pgfonlayer}
	\begin{pgfonlayer}{edgelayer}
		\draw [style=arrow] (3.center) to (10.center);
		\draw [style=arrow] (4.center) to (5.center);
		\draw [style=arrow] (6.center) to (7.center);
		\draw [style=arrow] (8.center) to (9.center);
	\end{pgfonlayer}
\end{tikzpicture}\]
 \caption{No input on $S_0$}
        \label{sfig1}
\end{subfigure}
\caption{Quantum Case on Base Cases}\label{fig:qc}
\end{figure}
And conversely, if the input has the form $\ket{1}_\q \otimes \ket{\psi}$, the target input $\ket{\psi}$ is passed to $S_1$, and $S_0$ does not receive any input (see Figure~\ref{sfig1}).
The action on superpositions is obtained by linearity. Therefore, it is necessary to define the action of programs both on regular inputs and on the empty input. Given that the transition system is probabilistic, it suffices to assign to each program a default transition: when the input is empty, this is the transition that occurs. Accordingly, transitions will have the form $$[S,\ket{\psi}]_\Gamma \stackrel{\nu}{\to} \ket{\psi'},$$ for some  program statement $S$, input state $\ket{\psi}$, and output state $\ket{\psi'}$. The bit $\nu \in \{0,1\}$ indicates whether this transition is the default one ($\nu=1$ for default transitions). Such a transition occurs with probability $\frac{\|\ket{\psi'}\|^2}{\|\ket{\psi}\|^2}$, if $\ket{\psi}\neq 0$ (probability $0$ otherwise).

The denotational semantics is described using coherent quantum operations. A coherent quantum operation is a pair $(\mathcal{C},F)$, where $\mathcal{C}$ is a quantum operation (i.e., a completely positive trace non-increasing linear map), and $F$ is a linear map called a \emph{transformation matrix}.
Although it is not possible to define the quantum case operation on quantum operations only, the adjunction of a transformation matrix (interpreted as necessary information about the implementation of the quantum channel in \cite{grenoble_coherent_control,kristjansson}) carries sufficient additional information for well-defined quantum control.
Therefore, the denotational semantics assigns to each program a pair $(\mathcal{C},F)$.

Adequacy (Theorem~\ref{thm:adequacy}) states that the operational and denotational semantics describe the same behavior of programs. Consequently, this result establishes a link between the addition of default transitions in the operational semantics, and the addition of transformation matrices in the denotational semantics. Full abstraction result (Theorem~\ref{thm:fullabstraction}) takes this correspondence further by showing that the denotational semantics matches the notion of observational equivalence.
As an example of observational equivalence, we will show that for all statements $S_0,S_0',S_1,S_1'$, the programs \[\tb{qcase }\q\ (0\rightarrow S_0,\ 1\rightarrow S_1); \ \tb{qcase }\q\ (0\rightarrow S_0',\ 1\rightarrow S_1')\] and \[\tb{qcase }\q\ (0\rightarrow S_0;S_{0}',\ 1\rightarrow S_1;S_1')\] are observationally equivalent, which is essentially an interchange law for the quantum case (Example~\ref{ex:interchange}). Conversely, we will show that $\tb{qcase} \ \q \ (0\rightarrow S,1\rightarrow S)$ is generally not observationally equivalent to $S$ (Example~\ref{ex:qcase_of_coin}).

\subsection{Related Work}

There exists a variety of programming languages that support coherent control.
It can take the form of quantum $\tb{if}$ statements~\cite{altenkirch,FBQP,ying_laws,ying_verification,silq,HLMR25}, generalizations of pattern matching~\cite{sabry,qunity,dave2025combining,HLMR25}, superpositions of terms in approaches based on the $\lambda$-calculus~\cite{lineal,realizability,lambda_s1}, superpositions of the wiring between quantum gates~\cite{addressable_gates}, or superpositions in the state of the program counter~\cite{quantum_control_machine,zhang}.
In all the above examples, quantum control exclusively concerns unitary operations.
Languages with control of non unitary operations were developed in~\cite{ying_control_flow,ying_alternation,badescu}, where the semantics is based on Kraus decompositions. However, these languages do not include primitives for recursion, be it classical or quantum.

Though not directly studied here, the question of defining quantum controlled recursion is closely related to our work. The first definition of quantum controlled recursion was proposed in~\cite{ying_recursion}, using Fock spaces to represent the state of an arbitrarily large number of control qubits. However, the interpretation of programs is not Scott-continuous. More recently, in~\cite{ying_verification}, the authors define recursive procedures parameterized by classical expressions, in the presence of a quantum conditional. In~\cite{qloop2025}, the authors define a semantics of quantum loops as linear maps over infinite-dimensional Hilbert spaces. None of the three references above feature measurement; to our knowledge, the state of the art does not provide any framework for quantum recursion in the presence of measure.

From another perspective, the coherent control of arbitrary operations has been studied in the physics literature~\cite{universal_control,dong2021controlled,grenoble_coherent_control,kristjansson,kristjansson_1,V21}.
These models do not include recursion, or only finite recursion \cite{PBS}.
It is useful to distinguish two types of quantum control, as was done in~\cite{kristjansson}. The first kind is the control of choice of quantum channels, where several alternative channels are applied in a superposition.
This is the kind of control that is studied in this paper.
The second is the control of causal orders, where instead of controlling which operation to apply, we control the order in which a set of operations are applied. A well-known example of control of orders is the \emph{quantum switch}~\cite{chiribella}, where two operations are performed in a superposition of orders. The former kind of control was shown to demonstrate an advantage in a communication task involving superpositions of depolarizing channels~\cite{grenoble_coherent_control}.

This paper is structured as follows. In Section~\ref{section:syntax}, we present the syntax and well-formedness rules of our language. In Sections~\ref{section:operational_semantics} and~\ref{section:denotational_semantics}, we describe the operational and denotational semantics, respectively. Then we present our main results: universality of the language for coherent quantum operations in Section~\ref{section:universality_all}, adequacy between the operational and denotational semantics in Section~\ref{section:adequacy}, and full abstraction in Section~\ref{section:observational_equivalence}. In Section~\ref{section:interpretation_design}, we discuss our design choices and their relation to the physical interpretation of the language. We give our summary and conclusion in Section~\ref{section:ccl}.

\section{A Programming Language for Quantum Control}
\label{section:syntax}

\subsection{Syntax}

We present a quantum language mixing classical and quantum control. Let $(\mathcal{X}, \prec)$ be a totally strictly ordered and countably infinite set of variables $\p,\q, \ldots$
The \emph{statements} of the language are described by the grammar in Figure~\ref{fig:syntax}.
We provide some brief explanations on statements:
$\tb{skip}$ is a no-op statement, $\tb{new qbit }\q$ initializes the qubit variable $\q$ in the state $\ket{0}$, and $\tb{discard}$ traces out its argument. The statement $\q {*}{=}U$ applies the single-qubit unitary gate $U$ to $\q$. As we will see in Section~\ref{section:universality_approx}, the considered gates can be restricted to $H$ and $T$.
$S_0;S_1$ applies $S_0$ then $S_1$ in a sequence. The statement $\tb{meas}\ \q\ (0\rightarrow S_0,1\rightarrow S_1)$ performs a standard computational basis measurement and executes $S_0$ or $S_1$ depending on the outcome. The statement $\tb{while\,}\q\ \tb{do\,}S$ is a classically-controlled loop: with each iteration, the control qubit $\q$ is measured in the computational basis, and the loop is exited once outcome $0$ is obtained. The statement $\tb{qcase}\ \q \ (0\rightarrow S_0,1\rightarrow S_1)$ is a quantum conditional without measurement of the control qubit $\q$.
Var$(S)$ is inductively defined as the set of variables appearing in the statement $S$.

\begin{figure}
 \hrulefill
 \begin{align*}
 S,S_0,S_1 ::=\ & \ \tb{skip}\ |\  \tb{new qbit }\q\ |\ \tb{discard }  \q \ |\ \q{*}{=}U\ |\ S_0;S_1
 \\
 |\ & \ \tb{meas}\ \q\ (0\rightarrow S_0,1\rightarrow S_1)   \ | \  \tb{while} \ \q\ \tb{do}\ S\ |\ \tb{qcase}\ \q \ (0\rightarrow S_0,1\rightarrow S_1)
\end{align*}
\hrulefill
\caption{Syntax of Statements }
\label{fig:syntax}
\end{figure}

An \emph{environment} $\Gamma\in \pfin{\mathcal X}$ is a finite set of variables, written as $\Gamma \triangleq \q_1,\ldots,\q_n$, where each variable is implicitly assigned the type qubit. Let $\emptyset$ denote the empty environment. $|\Gamma|$ is the cardinal of $\Gamma$.
Let
$\cdot,\cdot:\pfin{\mathcal X}\times \pfin{\mathcal X}\to \pfin{\mathcal X}$ be the partial function such that $\Gamma,\Delta$ is equal to $\Gamma\cup \Delta$ if $\Gamma \cap \Delta = \emptyset$, and is undefined otherwise.

A \emph{program} is a pair $(\Gamma;S) $ consisting of a statement $S$ and an input environment $\Gamma$.

\subsection{Well-Formedness}\label{ss:wf}

In order to guarantee that programs can be given a physical meaning, we introduce a well-formedness condition. A \emph{judgment} is an expression of the form
\[
 \Gamma \vdash S\triangleright \Delta,
\]
where $S$ is a statement and $\Gamma$ as well as $\Delta$ are environments. Well-formedness is defined by the set of inference rules in Figure~\ref{fig:wellformedness}.
A program $(\Gamma;S)$ is said to be \emph{well formed} if  there exists an environment $\Delta$ such that $\Gamma \vdash S \triangleright \Delta$ is derivable; and a statement $S$ is said to be \emph{well formed} if there exist $\Gamma$ such that $(\Gamma;S)$ is a well-formed program. Throughout the paper, we will only consider well-formed programs.

\begin{figure}
\hrulefill
\[
\begin{array}{c}
\begin{prooftree}
 \hypo{\vphantom{S_0}}
 \infer1{\Gamma \vdash \tb{skip} \triangleright \Gamma}
\end{prooftree}
\qquad
\begin{prooftree}
 \hypo{\vphantom{S_0}}
 \infer1{\Gamma \vdash \tb{new qbit }\q\  \triangleright \q , \Gamma }
\end{prooftree}
\qquad
\begin{prooftree}
 \hypo{\vphantom{S_0}}
 \infer1{\q,\Gamma \vdash \tb{discard }\q\  \triangleright \Gamma}
\end{prooftree}
\qquad
\begin{prooftree}
 \hypo{\vphantom{S_0}}
 \infer1{\q,\Gamma \vdash \q{*}{=}U \triangleright \q,\Gamma}
\end{prooftree}
\\[4ex]
\begin{prooftree}
 \hypo{\Gamma \vdash S_0 \triangleright \Gamma'}
 \hypo{\Gamma'\vdash S_1 \triangleright \Gamma''}
 \infer2[
 ]{\Gamma \vdash S_0;S_1 \triangleright\Gamma''}
\end{prooftree}
\qquad
\begin{prooftree}
 \hypo{\q,\Gamma \vdash S_0  \triangleright \Gamma'}
 \hypo{\q,\Gamma\vdash S_1 \triangleright\Gamma'}
 \infer2{\q,\Gamma \vdash \tb{meas }\q\ (0\rightarrow S_0,1\rightarrow S_1) \triangleright \Gamma'}
\end{prooftree}
\\[4ex]
 \begin{prooftree}
 \hypo{\q,\Gamma\vdash S \ \triangleright \q,\Gamma}
  \infer1{\q,\Gamma\vdash \tb{while }\q \tb{ do }S \triangleright \q,\Gamma}
 \end{prooftree}
 \qquad
\begin{prooftree}
 \hypo{\Gamma \vdash S_0 \triangleright \Gamma'}
 \hypo{\Gamma \vdash S_1 \triangleright \Gamma'}
 \hypo{\q\notin \text{Var}(S_0)\cup\text{Var}(S_1)}
\infer3{\q,\Gamma \vdash \tb{qcase}\ \q\ (0\rightarrow S_0,1\rightarrow S_1) \triangleright \q,\Gamma'}
\end{prooftree}
\\[4ex]
\end{array}
\]
\hrulefill
\caption{Well-Formedness Rules}
\label{fig:wellformedness}
\end{figure}

In particular, well-formedness forbids the branches of a $\tb{qcase}$ statement from accessing the control qubit. Moreover, it excludes statements such as  $\tb{meas }\p \ (0\rightarrow \tb{new qbit }\q,1\rightarrow \tb{skip})$ where the branching statements lead to inconsistent sets of variables in the output.

For a well-formed program $(\Gamma;S)$, there exists exactly one admissible environment $\Delta$ for which $\Gamma\vdash S \triangleright \Delta$ is derivable. We write this unique output environment as $\Gamma^S$. However, the input environment $\Gamma$ is not unique.
For instance, if $\Gamma \vdash S \triangleright\Delta$ is derivable and $\q\notin \Gamma\cup \text{Var}(S)$, then $\q,\Gamma \vdash S \triangleright \q,\Delta$ is derivable.

To ease the definition of the semantics, we do not consider alpha conversion, and as a consequence the statement $\q{*}{=}U;$ $ \tb{new qbit }\q;$ $ \tb{discard } \q$ is not well formed. We thus define the set of \emph{bound variables} of a statement,
which intuitively correspond to the set of variables that are initialized then deleted inside the program, appearing in neither its input nor its output environment.
Formally, the set of bound variables of a well-formed statement $S$ is defined as BV$(S)\triangleq \text{Var}(S)\setminus (\Gamma\cup\Gamma^S)$, where $\Gamma$ is an environment such that $(\Gamma;S)$ is a well-formed program.
The set BV$(S)$ is well defined because it does not depend on the choice of $\Gamma$.

\section{Operational Semantics}
\label{section:operational_semantics}

We now present a big-step operational semantics for well-formed programs of the language. The operational semantics describes the evolution of an input state induced by a program statement in terms of a probabilistic transition system.

\subsection{Quantum States and Configurations}

We recall some standard definitions from quantum computing and introduce the necessary notations (See also~\cite{nielsen} for additional information).
A quantum system is represented by a Hilbert space $\Hilb$, called its \emph{state space}. In this paper, Hilbert spaces will always be assumed to be finite-dimensional, and we assume that standard notations such as the inner product and the tensor product $\otimes$ are familiar to the reader. The state of that quantum system is given by a vector $\ket{\psi}\in \Hilb$ whose norm $\|\ket{\psi}\|$ is at most 1.\footnote{We consider subnormalized states rather than the usual normalized states in order to encode the probabilistic behavior of the language in the norms of states.} We write the set of states of $\Hilb$ as $St(\Hilb)$.

Each program variable in $\mathcal{X}$ will correspond to a distinct qubit system, which is given by a 2-dimensional Hilbert space. Accordingly, to each environment $\Gamma$ we assign the Hilbert space $\Hilb_{\Gamma}\triangleq\mathbb{C}^{2^{|\Gamma|}}$, so that each variable name of $\Gamma$ corresponds to a separate 2-dimensional subsystem.
The set of \emph{states} is defined by $\tb{States}\triangleq \cup_\Gamma St(\Hilb_{\Gamma})$.

As environments are sets (and therefore unordered), the order of subsystems is determined by the total ordering $\prec$ on $\mathcal{X}$.
For instance, if $\Gamma=\p,\q$ with $\p\prec \q$, we have $\Hilb_{\Gamma}=\mathbb{C}^4$ and we assume that $\p$ refers to the first qubit and $\q$ refers to the second qubit.

Given two environments $\Gamma$ and $\Delta$ s.t. $\Gamma \cap \Delta = \emptyset$, the intertwining  map $\iota_{\Gamma;\Delta}:St(\Hilb_{\Gamma})\otimes St(\Hilb_{\Delta}) \to St(\Hilb_{\Gamma,\Delta})$ is defined inductively as the linear map satisfying:
\begin{align*}
\iota_{\emptyset,\emptyset} &\triangleq id ,\\
\iota_{(\q_{min}, \Gamma); \Delta}(\ket{ax},\ket y) &\triangleq \ket a\otimes \iota_{\Gamma;\Delta}(\ket x,\ket y),\\
 \iota_{\Gamma;(\q_{min},\Delta)}(\ket{x},\ket {ay}) &\triangleq \ket a\otimes \iota_{\Gamma;\Delta}(\ket x,\ket y),
 \end{align*}
 with  $\forall \q \in \Gamma \cup \Delta,\ \q_{min} \prec \q$ and $a\in \{0,1\}$.

For $\ket{\psi}\in St(\Hilb_{\Gamma})$ and $\ket{\phi}\in St(\Hilb_{\Delta})$, we define:
\[
 \ket{\psi}_{\Gamma}\otimes \ket{\phi}_{\Delta} \triangleq \iota_{\Gamma;\Delta}(\ket{\psi},\ket{\phi}).
\]

Intuitively, this notation indicates that the variables of $\Gamma$ and $\Delta$ are put into the correct order according to $\prec$. As a particular case, the expression $\ket{\psi}_\q\otimes \ket{\phi}_{\Gamma}$ inserts the qubit $\q$ into the correct position. Using these notations, the choice of ordering $\prec$ will remain implicit in the definition of the operational and denotational semantics.

Given $\ket{\psi}\in St(\Hilb_{\q,\Gamma})$,  we write $\bra{0}_\q\ket{\psi}$ to denote the unique state $\ket{\psi_0}$ such that $\ket{\psi}=\ket{0}_\q\otimes\ket{\psi_0}_\Gamma+\ket{1}_\q\otimes\ket{\psi_1}_\Gamma$.
For an environment $\q,\Gamma$, let $U_{\q}$ be defined as the unitary map $U$ applied to the qubit $\q$ and identity elsewhere.

Finally, a \emph{configuration} is a tuple $[S,\ket{\psi}]_{\Gamma}$ where $(\Gamma;S)$ is a well-formed program and $\ket{\psi}$ is a state in $St(\Hilb_\Gamma)$. Let $\tb{Conf}$ denote the set of configurations.

\begin{figure}[t]
\hrulefill
\[
\begin{array}{c}
 \begin{prooftree}
  \hypo{\vphantom{S_0}}
  \infer1[(SK)]{[\tb{skip},\ket{\psi}]_\Gamma\stackrel{1}{\to} \ket{\psi}}
 \end{prooftree}
 \qquad
 \begin{prooftree}
  \hypo{\vphantom{S_0}}
  \infer1[(N)]{[\tb{new qbit }\q,\ket{\psi}]_{\Gamma}\stackrel{1}{\to} \ket{0}_{\q}\otimes \ket{\psi}_\Gamma}
 \end{prooftree}
\\[4ex]
 \begin{prooftree}
  \hypo{\vphantom{S_0}}
 \infer1[(D$_{0}$)]{[\tb{discard }\q,\ket{\psi}]_{\q,\Gamma}\stackrel{1}{\to} \bra{0}_{\q} \ket{\psi}}
 \end{prooftree}
 \qquad
 \begin{prooftree}
  \hypo{\vphantom{S_0}}
 \infer1[(D$_{1}$)]{[\tb{discard }\q,\ket{\psi}]_{\q,\Gamma}\stackrel{0}{\to} \bra{1}_{\q} \ket{\psi}}
 \end{prooftree}
\\[4ex]
 \begin{prooftree}
 \hypo{\vphantom{[S_0,\ket{\psi}]_\Gamma\stackrel{\nu_0}{\to} \ket{\psi'}}}
  \infer1[(U)]{[\q{*}{=}U,\ket{\psi}]_{\q,\Gamma}\stackrel{1}{\to} U_{\q}\ket{\psi}}
 \end{prooftree}
 \qquad
 \begin{prooftree}
 \hypo{[S_0,\ket{\psi}]_\Gamma\stackrel{\nu_0}{\to} \ket{\psi'}}
 \hypo{[S_1,\ket{\psi'}]_{\Gamma^{S_0}}\stackrel{\nu_1}{\to} \ket{\psi''}}
  \infer2[(S)]{[S_0;S_1,\ket{\psi}]_{\Gamma}\stackrel{\nu_0\nu_1}{\to} \ket{\psi''}}
 \end{prooftree}
\\[4ex]
 \begin{prooftree}
 \hypo{[S_0,\ket{0}\bra{0}_{\q}\ket{\psi}]_{\q,\Gamma}\stackrel{\nu}{\to} \ket{\psi'}}
 \infer1[(M$_0$)]{[\tb{meas}\ \q\ (0\rightarrow S_0,1 \rightarrow S_1),\ket{\psi}]_{\q,\Gamma}\stackrel{\nu}{\to} \ket{\psi'}}
 \end{prooftree}
\\[4ex]
 \begin{prooftree}
 \hypo{[S_1,\ket{1}\bra{1}_{\q}\ket{\psi}]_{\q,\Gamma}\stackrel{\nu}{\to} \ket{\psi'}}
 \infer1[(M$_1$)]{[\tb{meas}\ \q\ (0\rightarrow S_0,1 \rightarrow S_1),\ket{\psi}]_{\q,\Gamma}\stackrel{0}{\to} \ket{\psi'}}
 \end{prooftree}
\\[4ex]
 \begin{prooftree}
 \hypo{\vphantom{S_0}}
 \infer1[(W$_0$)]{[\tb{while }\q \tb{ do}\ S,\ket{\psi}]_{\q,\Gamma}\stackrel{1}{\to} \ket{0}\bra{0}_{\q}\ket{\psi}}
 \end{prooftree}
 \\[4ex]
 \begin{prooftree}
\hypo{[S;\tb{while} \ \q \ \tb{do}\ S,\ket{1}\bra{1}_{\q}\ket{\psi}]_{\q,\Gamma}\stackrel{\nu}{\to} \ket{\psi'}}
\infer1[(W$_1$)]{[\tb{while}\ \q\ \tb{do}\ S,\ket{\psi}]_{\q,\Gamma}\stackrel{0}{\to} \ket{\psi'}}
 \end{prooftree}
\\[4ex]
\begin{prooftree}
 \hypo{[S_0,\bra{0}_{\q}\ket{\psi}]_\Gamma\stackrel{\nu_0}{\to} \ket{\psi_0}}
 \hypo{[S_1,\bra{1}_{\q}\ket{\psi}]_\Gamma\stackrel{\nu_1}{\to} \ket{\psi_1}}
 \infer2[(Q)]{[\tb{qcase }\q\ (0\rightarrow S_0,1 \rightarrow S_1),\ket{\psi}]_{\q,\Gamma}\stackrel{\nu_0\nu_1}{\to} \nu_1 \ket{0}_{\q}\otimes \ket{\psi_0}_{\Gamma^{S_0}}+\nu_0\ket{1}_{\q}\otimes\ket{\psi_1}_{\Gamma^{S_1}}}
 \end{prooftree}
 \\[4ex]
 \end{array}
\]
\hrulefill
\caption{Operational Semantics}
\label{fig:semantics_operational}
\end{figure}

\subsection{Probabilistic Transition System}
\label{section:operational_semantics_def}

The big-step semantics is defined inductively  in Figure~\ref{fig:semantics_operational} as the probabilistic transition relation $\cdot \stackrel{\cdot}{\rightarrow} \cdot \subseteq \tb{Conf}\times \{0,1\}\times \tb{States}$.

The transition $[S,\ket{\psi}]_{\Gamma}\stackrel{\nu}{\to} \ket{\psi'}$ holds when the statement $S$ transforms the quantum state $\ket \psi \in St(\mathcal{H}_\Gamma)$ into $\ket {\psi'}\in St(\mathcal{H}_{\Gamma^S})$. The Boolean value $\nu$
answers the question, ``Is this the default transition?'' That is, $\nu=1$ on default transitions, and $\nu=0$ otherwise.
We remind the reader that the default transition describes how the system behaves when the input is empty.
The probability of a transition occurring is encoded in the norms of its input and output states, as is fairly standard (see for instance~\cite[Convention 3.3]{QPL}): the transition $[S,\ket{\psi}]_\Gamma \stackrel{\nu}{\to} \ket{\psi'}$ occurs with probability $\frac{\|\ket{\psi'}\|^2}{\|\ket{\psi}\|^2}$ if $\ket{\psi}\neq 0$, and 0 otherwise. Notice that the probability is independent of $\nu$.
In Figure~\ref{fig:semantics_operational}, the three probabilistic statements $\tb{discard}$, $\tb{meas}$, and $\tb{while}$ each give rise to two distinct rules, indexed by measurement outcomes. For each such statement, one of the two possible rules is set as the default one, namely (D$_0$), (M$_0$), and (W$_0$). (D$_0$) and (W$_0$) are marked as default by the annotation $\nu=1$, and (M$_0$) inherits its default transition from the statement $S_0$, hence the annotation $\nu$.
The choice of which transitions are set as default is discussed in more detail in Section~\ref{section:choice}.

In what follows, we will write $\pi\triangleright [S,\ket{\psi}]_\Gamma\stackrel{\nu}{\to}\ket{\psi'}$ to denote a derivation tree of root $[S,\ket{\psi}]_\Gamma\stackrel{\nu}{\to}\ket{\psi'}$ generated by applying the rules of Figure~\ref{fig:semantics_operational}.

We now provide some additional insights into the transition relation in Figure~\ref{fig:semantics_operational}.
Rule (SK) does nothing. Rule (N) rule prepares a new qubit in the state $\ket{0}$.
Rule (U) applies the unitary gate $U$ to the appropriate qubit. Rule (S) composes the action of statements $S_0$ and $S_1$. Moreover, its default transition is given by the default transitions for $S_0$ and $S_1$ in a sequence.
Rules (M$_0$) and (M$_1$) correspond to performing a computational basis measurement of the control qubit, and executing the branch corresponding to the measurement outcome.
Regarding the discard operation, as the operational semantics is defined on pure states, we follow a standard interpretation of $\tb{discard}~\q$ which consists in performing a computational basis measurement (equivalently to $\tb{meas}~\q\ (0\to \tb{skip}, 1\to \tb{skip})$) and then forbid the use of $\q$. Thus, Rules (D$_0$) and (D$_1$) project the qubit $\q$ onto their respective computational basis vectors.
Rules (W$_0$) and (W$_1$) implement a $\tb{while}$ loop by measuring the control qubit at every iteration, and exiting the loop when the outcome is $0$.
Finally, Rule (Q) for $\tb{qcase}$ statements reads as follows: for each statement $S_k$ ($k \in \{0,1\}$), we are given a possible transition $[S_k,\bra{k}_{\q}\ket{\psi}]_{\Gamma}\stackrel{\nu_k}{\to} \ket{\psi_k}$.
Intuitively, the output state of  $\tb{qcase}\ \q\ (0\rightarrow S_0,1\rightarrow S_1) $ is a superposition of the cases where the input is given to $S_0$ or to $S_1$ depending on the state of $\q$. When $\q$ is in state $\ket 0$, the input is routed to $S_0$, and no input is given to $S_1$. Thus,  $S_0$ will produce the state $\ket{\psi_0}$ if $S_1$ follows the default transition. However, the current transition of $S_1$ might not be the default one (when $\nu_1=0$), in this case this branch of the superposition will not contribute to the final state, leading to the term $\nu_1\ket{0}_{\q}\otimes \ket{\psi_0}_{\Gamma^{S_0}}$. Similarly, when $\q$ is in state $\ket 1$ we get $\nu_0\ket{1}_{\q}\otimes \ket{\psi_1}_{\Gamma^{S_1}}$. Hence the final state $\nu_1 \ket{0}_{\q}\otimes \ket{\psi_0}_{\Gamma^{S_0}}+\nu_0\ket{1}_{\q}\otimes\ket{\psi_1}_{\Gamma^{S_1}}$.
In particular, when both $S_0$ and $S_1$  implement unitary evolutions (see Example~\ref{ex:cnot_op}), then $\nu_0=\nu_1=1$.

As an example,  with $S = \tb{qcase}~\q~ (0\to \tb{discard}~ \p, 1\to \tb{discard}~ \p)$ and $\ket \varphi = (\frac{\ket 0+\ket 1}{\sqrt 2})_\q \otimes (\frac{\ket 0+\ket 1}{\sqrt 2})_\p$, the following transitions hold.
\begin{align*}
[S,\ket\varphi]_{\{\q,\p\}} &\stackrel{1}{\to}  \frac{\ket 0+\ket 1}{2} &[S,\ket\varphi]_{\{\q,\p\}} &\stackrel{0}{\to}  \frac{\ket 0}{2}\\
[S,\ket\varphi]_{\{\q,\p\}} &\stackrel{0}{\to}  \frac{\ket 1}{2} &[S,\ket\varphi]_{\{\q,\p\}} &\stackrel{0}{\to}  0
\end{align*}
Thus, the state of qubit $\q$ at the end of the computation is $\frac{\ket 0 + \ket 1}{\sqrt 2}$ with probability $1/2$ (when $\nu_0=\nu_1=1)$; $\ket 0$ with probability $1/4$ (when $\nu_0=1$ and $\nu_1=0$); and   $\ket 1$ with probability $1/4$ (when $\nu_0=0$ and $\nu_1=1$). Notice that the case $\nu_0=\nu_1=0$ always occurs with probability $0$ by definition of Rule (Q).

We now examine the probabilistic structure of the semantics in more detail.
Since the probability of a transition $[S,\ket{\psi}]_\Gamma \stackrel{\nu}{\to} \ket{\psi'}$ is defined as $\frac{\|\ket{\psi'}\|^2}{\|\ket{\psi}\|^2}$  if $\ket{\psi}\neq 0$, and $0$ otherwise, one must check that the transition probabilities are well defined, that is, the coefficients $\frac{\|\ket{\psi'}\|^2}{\|\ket{\psi}\|^2}$ for transitions originating from the same configuration sum to at most one (when $\ket{\psi}\neq 0$).

We first recall the definition of multiset. Multisets are a similar concept to sets, where elements may appear multiple times. This will be used to reflect the fact that multiple derivations from the same configuration may lead to the same state, but their transition probabilities must still be counted separately as they correspond to distinct program executions. Formally, a multiset over a set $A$ is a pair $(A,M)$ where $M:A\mapsto \mathbb{N}$ is the function that maps each element to its multiplicity. We will write multisets in brackets $\{\cdot\}$ similarly to sets, repeating each element according to its multiplicity. For example, the multiset $(\{a,b,c\},(a\mapsto 1,b\mapsto 2,c\mapsto 0))$ will be written as $\{a,b,b\}$. The support of $M$ is defined as the subset of elements of $A$ that appear in $M$ at least once:
\[
 \text{Supp}(M)\triangleq\{a\in A\ |\ M(A)> 0\}.
\]

Standard operations on sets are adapted to multisets as follows. Given multisets $M$ and $N$ over $A$, we say that $a\in M$ if $a\in \text{Supp}(M)$, and we say that $M\subseteq N$ if for all $a\in A$, $M(a)\leq N(a)$. The intersection of $M$ and $N$ is defined by $(M\cap N)(a)\triangleq\min(M(a),N(a))$ for all $a\in A$. The union of $M$ and $N$ is defined by $(M\cup N) (a)\triangleq\max(M(a),N(a))$ for all $a\in A$. The disjoint union of $M$ and $N$ is defined by $(M\uplus N) (a)\triangleq M(a)+N(a)$ for all $a\in A$, where copies of the same element in both $M$ and $N$ counted separately. Finally, we will sometimes write multisets as indexed families, for instance $\{a_i\}_{i\in I}$. This allows us to define the notions of summability and sums over multisets in the standard way for indexed families (see for instance~\cite{topology}).

For a fixed configuration $[S,\ket{\psi}]_{\Gamma}$, the output multiset $\mathscr{M}([S,\ket{\psi}]_\Gamma)$ is defined as the multiset of all possible output states, with their corresponding bit $\nu$, that can be reached from $[S,\ket{\psi}]_\Gamma$.

\begin{defi}[Multiset semantics]
\label{def:output_multiset}
The \emph{output multiset} of a configuration $[S,\ket{\psi}]_\Gamma \in \tb{Conf}$ is defined as the multiset
\[
 \mathscr{M}([S,\ket{\psi}]_\Gamma)\  \triangleq\biguplus_{\substack{\pi\triangleright [S,\ket{\psi}]_\Gamma\stackrel{\nu}{\to}\ket{\psi'}\\ (\ket{\psi'},\nu)\neq(0,0)}}\{(\ket{\psi'},\nu)\},
\]
where $\uplus$ denotes the disjoint union of multisets.\footnote{The element $(0,0)$ is removed because it does not necessarily have finite multiplicity, as required by the definition of a multiset, and is irrelevant to program behavior.}
\end{defi}

The following lemma ensures both that transition probabilities are well defined and that all multiplicities in $\mathscr{M}([S,\ket{\psi}]_\Gamma)$ are finite (and therefore that $\mathscr{M}([S,\ket{\psi}]_\Gamma)$ is well defined).

\begin{lem}
\label{lem:prob_distr}
If $\mathscr{M}([S,\ket{\psi}]_\Gamma)=\{(\ket{\psi'_i},\nu_i)\}_{i\in I}$, with $\ket{\psi}\neq 0$, then:
\begin{itemize}
 \item $\left(\frac{\|\ket{\psi'_i}\|^2}{\|\ket{\psi}\|^2}\right)_i$ is summable and $\sum_i \frac{\|\ket{\psi'_i}\|^2}{\|\ket{\psi}\|^2} \leq 1$. Therefore this family forms a probability subdistribution.
\item There is exactly one index $i$ such that $\nu_i=1$.
\end{itemize}
\end{lem}
\begin{proof}
We show by induction on the structure of $S$ that for all $\Gamma$ and $\ket{\psi}$, both results holds. This is clear in all cases except for the $\tb{while}$ statement.
Consider a statement $\tb{while}\ \q\ \tb{do}\ S$, and an environment $\Gamma$ such that the judgment $\q, \Gamma\vdash \tb{while}\ \q\ \tb{do}\ S\, \triangleright\q,\Gamma$ is derivable.
Given $\ket{\psi}\in St(\Hilb_{\q,\Gamma})$, derivations of transitions from $[\tb{while}\ \q \ \tb{do}\ S,\ket{\psi}]_{\q,\Gamma}$ can have arbitrary depth regardless of $S$, because Rule (W$_1$) can be applied arbitrarily many times. In order to bound the number of iterations of the $\tb{while}$ loop, we modify our language by adding a $\tb{while}$ statement indexed by a maximum number of iterations $n\geq 1$ (counting the step during which the loop is exited). New rules are added to the operational semantics:
\[
 \begin{prooftree}
 \hypo{\phantom{\bra{1}_{\q}}}
 \infer1[(W$_{n,0}$)]{[\tb{while}_n\ \q \tb{ do}\ S,\ket{\psi}]_{\q,\Gamma}\stackrel{1}{\to} \ket{0}\bra{0}_{\q}\ket{\psi}}
 \end{prooftree}\raisebox{0.5ex}{,\ \text{for all $n\geq 1$, and }}
\]
\[
 \begin{prooftree}
\hypo{[S;\tb{while}_{n-1} \ \q \ \tb{do}\ S, \ket{1}\bra{1}_{\q}\ket{\psi}]_{\q,\Gamma}\stackrel{\nu}{\to} \ket{\psi'}}
\infer1[(W$_{n,1}$)]{[\tb{while}_n\ \q\ \tb{do}\ S,\ket{\psi}]_{\q,\Gamma}\stackrel{0}{\to} \ket{\psi'}}
 \end{prooftree},\ \text{for all $n\geq 2$.}
\]
If for some $n$, $[\tb{while}_n \ \q\ \tb{do}\ S,\ket{\psi}]_{\q,\Gamma}\stackrel{\nu}{\to}\ket{\psi'}$ then $[\tb{while} \ \q\ \tb{do}\ S,\ket{\psi}]_{\q,\Gamma}\stackrel{\nu}{\to}\ket{\psi'}$ (this can be shown by induction on $n$). And conversely, if $[\tb{while} \ \q\ \tb{do}\ S,\ket{\psi}]_{\q,\Gamma}\stackrel{\nu}{\to}\ket{\psi'}$ then there exists $n$ such that $[\tb{while}_n \ \q\ \tb{do}\ S,\ket{\psi}]_{\q,\Gamma}\stackrel{\nu}{\to}\ket{\psi'}$. Moreover, there is a one-to-one correspondence between derivations of these transitions.
Therefore:
\begin{align*}
\mathscr{M}([\tb{while}\ \q \ \tb{do}\ S,\ket{\psi}]_{\q,\Gamma}) = \bigcup_{n\geq 1}\mathscr{M}([\tb{while}_n\ \q \ \tb{do}\ S,\ket{\psi}]_{\q,\Gamma}).
\end{align*}
We prove by induction on $n\geq 1$ that for all $\ket{\psi}\in St(\Hilb_{\q,\Gamma})$, if $\mathscr{M}([\tb{while}_n\ \q \ \tb{do}\ S,\ket{\psi}]_{\q,\Gamma}) = \{(\ket{\psi_i'},\nu_i)\}_i$ then $\left(\frac{\|\ket{\psi'_i}\|^2}{\|\ket{\psi}\|^2}\right)_i$ is summable and $\sum_i \frac{\|\ket{\psi'_i}\|^2}{\|\ket{\psi}\|^2} \leq 1$.

For $n=1$ the result is clear.
Suppose that the result holds for a fixed $n\geq 1$. For a multiset $\{(\ket{\phi_i},\mu_i)\}_i$ and $\lambda\in\mathbb{C}$, define $\{(\ket{\phi_i},\mu_i)\}_i\cdot \lambda \triangleq \{(\ket{\phi_i},\lambda \mu_i)\}_i$, where all occurrences of the element $(0,0)$ are removed.
We have for all $\ket{\psi}$:

\[
 \mathscr{M}([\tb{while}_{n+1}\ \q \ \tb{do}\ S,\ket{\psi}]_{\q,\Gamma})  = \left\{\left(\ket{0}\bra{0}_{\q}\ket{\psi},1\right)\right\}\uplus \mathscr{M}\left([S;\tb{while}_{n}\ \q \ \tb{do}\ S,\ket{1}\bra{1}_\q\ket{\psi}]_{\q,\Gamma}\right) \cdot 0.
\]
Let $\mathscr{M}\left([S,\ket{1}\bra{1}_\q\ket{\psi}]_{\q,\Gamma}\right) = \{\ket{\varphi_j},\mu_j\}_j$. Then,
\[
 \mathscr{M}\left([S;\tb{while}_{n}\ \q \ \tb{do}\ S,\ket{1}\bra{1}_\q \ket{\psi}]_{\q,\Gamma}\right)\  =\ \biguplus_j  \mathscr{M}\left([\tb{while}_{n}\ \q \ \tb{do}\ S,\ket{\varphi_j}]_{\q,\Gamma}\right) \cdot \mu_j.
\]
Therefore:
\[
\mathscr{M}([\tb{while}_{n+1}\ \q \ \tb{do}\ S,\ket{\psi}]_{\q,\Gamma})
 = \left\{\left(\ket{0}\bra{0}_{\q}\ket{\psi},1\right)\right\}\uplus \Big(  \biguplus_j  \mathscr{M}\left([\tb{while}_{n}\ \q \ \tb{do}\ S,\ket{\varphi_j}]_{\q,\Gamma}\right)\cdot 0 \Big).
\]
For all $j$, let $\mathscr{M}\left([\tb{while}_{n}\ \q \ \tb{do}\ S,\ket{\varphi_j}]_{\q,\Gamma}\right) = \{(\ket{\phi_j^l},\beta_{j}^l)\}_l$. Then, the following inequalities hold (in $\mathbb{R}_+\cup \{+\infty\}$):
\begin{align*}
 \|\ket{0}\bra{0}_q\ket{\psi}\|^2 + \sum\nolimits_j\sum\nolimits_l \|\ket{\phi_j^l}\|^2  & \leq \|\ket{0}\bra{0}_q\ket{\psi}\|^2 + \sum\nolimits_j\|\ket{\varphi_j}\|^2
 & \text{\small (by induction hypothesis)} \\
 & \leq \|\ket{0}\bra{0}_q\ket{\psi}\|^2 + \|\ket{1}\bra{1}_q\ket{\psi}\|^2
 & \text{\small (by induction hypothesis)} \\
 & \leq \|\ket{\psi}\|^2,
\end{align*}
which proves the result for $n+1$.

Therefore, by writing $\mathscr{M}([\tb{while}\ \q \ \tb{do}\ S,\ket{\psi}]_{\q,\Gamma}) = \{(\ket{\psi_i'},\nu_i)\}_i$, we have $\sum_i \frac{\|\ket{\psi'_i}\|^2}{\|\ket{\psi}\|^2} \leq 1$. To see that the second point holds, the derivation of any transition from $[\tb{while}\ \q \ \tb{do}\ S,\ket{\psi}]_{\q,\Gamma}$ must end with either Rule (W$_0$) or Rule (W$_1$) rule. Therefore there is exactly one outcome with $\nu_i=1$, obtained by applying Rule (W$_0$).
\end{proof}

The second point can be interpreted as meaning that if no input is given, a program always terminates. In other words, there always exists a default transition.
Lastly, we define the probability of termination of programs.
\begin{defi}[Probability of termination]
\label{def:proba_termination}
For a configuration $[S,\ket{\psi}]_{\Gamma} \in \tb{Conf}$ such that $\ket{\psi}\neq 0$, the \emph{probability of termination} of $S$ on an input state $\ket{\psi}$ is defined as
\[p(S,\ket{\psi})_\Gamma\ \triangleq \sum_{(\ket{\psi'}, \nu ) \in \mathscr{M}([S,\ket{\psi}]_\Gamma)}  \frac{\|\ket{\psi'}\|^2}{\|\ket{\psi}\|^2}\raisebox{0.4ex}{.}\]
\end{defi}

From now on, we will sometimes omit the subscript $\Gamma$ in configurations $[S,\ket{\psi}]_{\Gamma}$ and probabilities $p(S,\ket{\psi})_\Gamma$ when it is clear from the context.

\subsection{Examples}
\label{section:ex_op}

\begin{exa}[CNOT and SWAP]
\label{ex:cnot_op}
Let CNOT and SWAP be the standard 2-qubit unitary gates defined as follows on computational basis states:
\[
\text{CNOT}:\ket{x,y}\mapsto \ket{x,y\oplus x} \qquad
\text{SWAP}: \ket {x,y}\mapsto \ket{y,x},
\]
with $\oplus$ being the addition modulo 2. The CNOT acts by flipping the second qubit, called the target qubit, only when the first qubit, called the control qubit, is in state $\ket{1}$. The SWAP exchanges the states of the two qubits. Although the language only has single-qubit unitary gates built-in, the CNOT and SWAP gates can be implemented using the $\tb{qcase}$ primitive. For the CNOT, consider the following statement:
\begin{align*}
 S_{\text{CNOT}}^{\qc,\qt}\ \triangleq \ \tb{qcase}\ \qc \ (0 \rightarrow \tb{skip},1 \rightarrow \qt\ {*}{=} X),
\end{align*}
where $X$ is the Pauli matrix defined on computational basis states as $X:\ket x \mapsto \ket{1-x}$, $\qc$ represents the control qubit and $\qt$ represents the target qubit. Hence, one can derive the judgment
$\qc,\qt\vdash S_{\text{CNOT}}^{\qc,\qt} \triangleright \qc,\qt$. For all $\alpha\ket{00}+\beta\ket{01}+\gamma\ket{10}+\delta\ket{11}\in St(\Hilb_{\qc,\qt})$, assuming $\qc \prec \qt$ (i.e., $\qc$ is the first qubit), we have:
\[
 \begin{prooftree}
 \hypo{[\tb{skip},\alpha\ket{0}+\beta\ket{1}]
 \stackrel{1}{\to} \alpha\ket{0}+\beta\ket{1}
 }
 \hypo{[\qt \ {*}{=}X, \gamma\ket{0}+\delta\ket{1}]
  \stackrel{1}{\to} \gamma\ket{1}+\delta\ket{0}}
 \infer2[(Q)]{[S_{\text{CNOT}}^{\qc,\qt},\alpha\ket{00}+\beta\ket{01}+\gamma\ket{10}+\delta\ket{11}]
 \stackrel{1}{\to} \alpha\ket{00}+\beta\ket{01}+\gamma\ket{11}+\delta\ket{10}}
 \end{prooftree}\raisebox{0.4ex}{,}
\]
which matches the definition of the CNOT gate. Furthermore, since the SWAP gate can be decomposed into a sequence of three CNOT gates, it is implemented by the following statement:
\[
 S_{\text{SWAP}}^{\p,\q}\ \triangleq \ S_{\text{CNOT}}^{\p,\q}\ ;\ S_{\text{CNOT}}^{\q,\p}\ ;\ S_{\text{CNOT}}^{\p,\q}.
\]\end{exa}

\begin{exa}[Repeated coin flip]
\label{ex:coin_op}
Consider the program $(\q;\text{COIN})$, with
\begin{align*}
 \text{COIN}\ \triangleq \   \tb{while} \ \q \tb{ do } \q {*}{=} H,
 \end{align*}
and where $H\triangleq \ket{x} \mapsto \sfrac{1}{\sqrt{2}}(\ket{0}+(-1)^x \ket{1})$ is the Hadamard gate. The program performs a $\tb{while}$ loop on a qubit $\q$, applying $H$ with each iteration. The judgment $\q \vdash \text{COIN}\triangleright \q$ is derivable.
The execution of the program proceeds as follows: with each iteration, the loop is exited with a certain probability, so that by the end of the execution, $\q$ is in state $\ket{0}$. More precisely, given an input state $\alpha \ket{0}+\beta\ket{1}\in St(\Hilb_\q)$, the first iteration is exited with probability $|\alpha|^2$, and every subsequent iteration with probability $\frac{1}{2}$. Formally, we consider the transitions from the configuration $[\text{COIN},\ket{\psi}]_\q$. First, by applying Rule (W$_0$) rule, we~have
\[
\begin{prooftree}[center=true]
\hypo{\phantom{l}}
 \infer1[(W$_0$)]{[\text{COIN},\alpha\ket{0}+\beta\ket{1}]\stackrel{1}{\to} \alpha \ket{0}_{\q}}
\end{prooftree}\raisebox{0.9ex}{,}
\]
which corresponds to exiting the loop immediately. By applying (W$_1$) once, then (W$_0$), we~have
\[
 \begin{prooftree}
 \infer0[(U)]{[\q {*}{=} H,\beta \ket{1}]
 \stackrel{1}{\to} \sfrac{\beta}{\sqrt 2}(\ket 0{-}\ket 1)
 }
 \infer0[(W$_0$)]{[\text{COIN},\sfrac{\beta}{\sqrt 2}(\ket 0{-}\ket 1)]
\stackrel{1}{\to} \sfrac{\beta}{\sqrt{2}} \ket{0}
}
 \infer2[(S)]{[\q {*}{=} H ;\text{COIN},\beta\ket{1}]\stackrel{1}{\to} \sfrac{\beta}{\sqrt{2}}\ket{0}}
\infer1[(W$_1$)]{[\text{COIN},\alpha\ket{0}+\beta\ket{1}]\stackrel{0}{\to} \sfrac{\beta}{\sqrt{2}}\ket{0}}
 \end{prooftree}
\]
which corresponds to exiting the loop after one iteration, and so on. Generally, the transition $[\text{COIN},\alpha\ket{0}+\beta\ket{1}]_\q\stackrel{0}{\to} -\beta\left(-\frac{1}{\sqrt{2}}\right)^k\ket{0}_\q $
can be obtained by applying (W$_1$) $k\geq 1$ times, for exiting after $k$ iterations. Therefore:
\[
 \mathscr{M}([\text{COIN},\alpha\ket{0}+\beta\ket{1}]_\q)= \left\{\left(\alpha \ket{0},1\right)\right\}\uplus \left\{\left(-\beta\left(-\frac{1}{\sqrt{2}}\right)^k\ket{0},0\right)\right\}_{k\geq 1}\!\!.
\]
Consequently, the probability of termination is $p(\text{COIN},\alpha\ket{0}+\beta\ket{1})_\q=1$, regardless of the choice of $\alpha$ and $\beta$.
\end{exa}

\begin{exa}[Infinite loop]
\label{ex:loop_op}
The following statement implements an infinite loop:
 \begin{align*}
 \text{LOOP}\ \triangleq \  & \tb{new qbit }\qr\ ;\ \qr{*}{=}X\ ;\ \tb{while} \ \qr \tb{ do skip}\ ;\ \tb{discard} \ \qr,
 \end{align*}
 with $X$ the Pauli matrix defined on computational basis states as $X:\ket x \mapsto \ket{1-x}$. The statement acts by initializing a new qubit $\qr$, setting its value to $\ket{1}$, and using the qubit to control a $\tb{while}$ loop. Since the loop does not modify the state of $\qr$, the program cannot terminate.

 Formally, for all environment $\Gamma$ such that $\qr \notin \Gamma$, we can derive $\Gamma\vdash \text{LOOP}\triangleright \Gamma$. We consider transitions from the configuration $[\text{LOOP},\ket{}]_{\emptyset}$, where $\ket{}$ denotes the state of a one-dimensional input system. Consider the following transitions:
 \[
 \begin{prooftree}
 \infer0[(W$_0$)]{[\tb{while }\qr \tb{ do}\ \tb{skip},\ket{1}]_\qr\stackrel{1}{\to} 0}
 \end{prooftree}
\]
and
\[
 \begin{prooftree}
 \infer0[(SK)]{[\tb{skip},\ket{1}]_{\qr}\stackrel{1}{\to} \ket{1}_\qr}
 \infer0[(W$_0$)]{[\tb{while} \ \qr \ \tb{do}\ \tb{skip},\ket{1}]_{\qr} \stackrel{1}{\to} 0}
\infer2[(S)]{[\tb{skip};\tb{while} \ \qr \ \tb{do}\ \tb{skip},\ket{1}]_{\qr}\stackrel{1}{\to} 0}
\infer1[(W$_1$)]{[\tb{while}\ \qr\ \tb{do}\ \tb{skip},\ket{1}]_{\qr}\stackrel{0}{\to} 0}
 \end{prooftree}
\]
By induction, we have $\mathscr{M}([\tb{while}\ \qr\ \tb{do}\ \tb{skip},\ket{1}]_{\qr})=\left\{(0,1) \right\}$, and therefore by applying Rule (S) several times, we obtain
$\mathscr{M}([\text{LOOP},\ket{}]_{\emptyset})=\left\{(0,1)\right\}$.
As expected, the probability of termination is $p(\text{LOOP},\ket{})=0$.
\end{exa}

\begin{exa}[Measurement encoding]\label{ex:meas}
The statement
\begin{align*}
\text{$M_{S_0,S_1}$}\ \triangleq \ &
        \tb{new qbit } \q' \ ;\  S_{\text{CNOT}}^{\q,\q'}\ ;  \  \tb{qcase}\ \q'\ (0\rightarrow S_0,1\rightarrow S_1)\ ;\ \tb{discard }\q'
\end{align*}
implements the measurement $ \tb{meas}\ \q\ (0\rightarrow S_0, \ 1\rightarrow S_1)$. The statement reads as follows: the $\tb{new qbit}$ and $S_{\text{CNOT}}^{\q,\q'}$ statements copy the qubit $\q$ onto a new qubit $\q'$ along computational basis states. This new qubit is then used as the control for the $\tb{qcase}$ statement, whose branches $S_0$ and $S_1$ are the same as those of the measurement statement we aim to implement. Note that $S_0$ and $S_1$ can act on $\q$. Afterwards, the auxiliary qubit $\q'$ is discarded. One can show that for all input state $\ket{\psi}$,
\begin{align}
 \label{eq:meas}
 \mathscr{M}([M_{S_0,S_1},\ket{\psi}]_{\q,\Gamma})=\mathscr{M}([\tb{meas }\q\ (0\rightarrow S_0,1\rightarrow S_1),\ket{\psi}]_{\q,\Gamma}),
\end{align}
justifying the equivalence with the measurement statement.

Indeed, suppose that $\q,\Gamma\vdash S_0 \triangleright \Gamma'$ and $\q,\Gamma\vdash S_1 \triangleright\Gamma'$ are derivable and $\q' \notin \text{Var}(S_0) \cup\text{Var}(S_1) \cup\Gamma \cup \Gamma'$, in which case $\q,\Gamma \vdash M_{S_0,S_1} \triangleright\Gamma'$ is derivable. Let $\ket{\psi}=\ket{0}_{\q}\otimes \ket{\psi_0}_{\Gamma}+\ket{1}_\q\otimes \ket{\psi_1}_{\Gamma}\in St(\Hilb_{\q,\Gamma})$. Transitions from the configuration $[M_{S_0,S_1},\ket{\psi}]_{\q,\Gamma}$ can be derived as follows. For the two first instructions, we have:
\[
 \begin{prooftree}
  \infer0[(N)]{[\tb{new qbit }\q',\ket{\psi}]\stackrel{1}{\to} \ket{0}_{\q'}\otimes \ket{\psi}_{\q,\Gamma}}
 \end{prooftree}
\]
and
\[
 \begin{prooftree}
 \infer0[(SK)]{[\tb{skip},\ket{0}_{\q'}\otimes \ket{\psi_0}_{\Gamma}]
 \stackrel{1}{\to} \ket{0}_{\q'}\otimes \ket{\psi_0}_{\Gamma}}
 \infer0[(U)]{[\q'{*}{=}X,\ket{0}_{\q'}\otimes \ket{\psi_1}_{\Gamma}]
\stackrel{1}{\to} \ket{1}_{\q'}\otimes \ket{\psi_0}_{\Gamma}}
 \infer2[(Q)]{[\tb{qcase}\ \q \ (0 \rightarrow \tb{skip},1\rightarrow \q'{*}{=}X),\ket{0}_{\q'}\otimes \ket{\psi}_{\q,\Gamma}]
\stackrel{1}{\to} \ket{00}_{\q',\q}\otimes \ket{\psi_0}_{\Gamma}+\ket{11}_{\q',\q}\otimes\ket{\psi_1}}
 \end{prooftree}.
\]
Then, for every pair of transitions $[S_k,\ket{k}_{\q}\otimes\ket{\psi}_{\Gamma}]\rightarrow (\ket{\phi_k}_{\Gamma'},\nu_k)$, with $k\in\{0,1\}$, we have:
\[
 \begin{prooftree}
 \hypo{\forall k\in\{0,1\}:[S_k,\ket{k}_{\q}\otimes\ket{\psi}_{\Gamma}]\stackrel{\nu_k}{\to} \ket{\phi_k}_{\Gamma'}}
 \infer1[(Q)]{\!\!\!\begin{array}{r}[\tb{qcase }\q'\ (0\rightarrow S_0,1 \rightarrow S_1),\ket{00}_{\q',\q}\otimes \ket{\psi_0}_{\Gamma}+\ket{11}_{\q',\q}\otimes\ket{\psi_1}_\Gamma]\\
 \stackrel{\nu_0\nu_1}{\to} \nu_1 \ket{0}_{\q'}\otimes \ket{\phi_0}_{\Gamma'}+\nu_0\ket{1}_{\q'}\otimes\ket{\phi_1}_{\Gamma'}
\end{array}\!\!\!
}
 \end{prooftree}\raisebox{1.8ex}{.}
\]
Lastly, the derivations below hold.
\[
 \begin{prooftree}
 \infer0[(D$_{0}$)]{[\tb{discard }\q',\nu_1 \ket{0}_{\q'}\otimes \ket{\phi_0}_{\Gamma'}+\nu_0\ket{1}_{\q'}\otimes\ket{\phi_1}_{\Gamma'}]\stackrel{1}{\to} \nu_1\ket{\phi_0}_{\Gamma'}}
 \end{prooftree}
\]
\[
 \begin{prooftree}
 \infer0[(D$_{1}$)]{[\tb{discard }\q,\nu_1 \ket{0}_{\q'}\otimes \ket{\phi_0}_{\Gamma'}+\nu_0\ket{1}_{\q'}\otimes\ket{\phi_1}_{\Gamma'}]\stackrel{0}{\to} \nu_0\ket{\phi_1}_{\Gamma'}}
 \end{prooftree}
\]

By applying the (S) rule several times, we obtain the transitions below.
\begin{align*}
 [M_{S_0,S_1},\ket{\psi}]_{\q,\Gamma}&\stackrel{\nu_0\nu_1}{\to} \nu_1\ket{\phi_0} \\
 [M_{S_0,S_1},\ket{\psi}]_{\q,\Gamma}&\stackrel{0}{\to} \nu_0\ket{\phi_1}
\end{align*}
There always exists exactly one default transition, for which $\nu_k=1$ (see Lemma~\ref{lem:prob_distr}),
hence we obtain Equation~\ref{eq:meas}.

In conclusion, the statement $M_{S_0,S_1}$ generates exactly the same transitions as the corresponding classical measurement, for all input state $\ket{\psi}$, ignoring occurrences of $(0,0)$ which have no effect.
In particular, the $\tb{meas}$ statement could be removed from the programming language without compromising its expressivity. We note that this is only possible because the $\tb{discard}$ statement still implicitely contains a computational basis measurement.
\end{exa}

\section{Denotational Semantics}
\label{section:denotational_semantics}

We now give the denotational semantics for our language. The aim is to assign to each program some mathematical operation that fully describes its behavior. Standard denotational semantics for quantum programming languages are based on \emph{quantum operations} (i.e., completely positive trace non-increasing maps), which are known to be insufficient to describe the coherent control of arbitrary operations~\cite{badescu}. Therefore, we introduce the semantical domain of \emph{coherent quantum operations}, inspired by the so-called vacuum extension approach~\cite{kristjansson_1,kristjansson}, which we show to be well-suited to defining a denotational semantics in the presence of quantum control, general recursion, and general quantum evolutions.

We begin with an overview of the standard notions of quantum channels and quantum operations. Then we present the notion of vacuum extension and introduce the notion of coherent quantum operations. Using this definition, we give the denotational semantics of programs.

\subsection{Quantum Operations and Channels}
\label{section:prelim_denotational}

Given two Hilbert spaces $\Hilb_A$ and $\Hilb_B$, let $\mathcal{L}(\Hilb_A, \Hilb_B)$ be the set of linear maps from $\Hilb_A$ to $\Hilb_B$. When $\Hilb_A=\Hilb_B$, we simply write $\mathcal{L}(\Hilb_A)$. The set of Hermitian maps over $\Hilb_A$ is written as $\text{Herm}(\Hilb_A)$.

To study probability distributions of quantum states, we use the density matrix representation, which captures both quantum superpositions and classical probability distributions (see for instance~\cite{nielsen}).
A density matrix over a finite-dimensional Hilbert space $\mathcal{H}$ is a positive semi-definite element of $\mathcal{L}(\Hilb)$ with trace at most one.
We write the set of density matrices over $\Hilb$ as $D(\Hilb)$.

The Löwner ordering $\leq$ is defined on the set of Hermitian matrices Herm($\Hilb_A$) as follows: $M\leq N$ if $M-N$ is positive semi-definite. The set of density matrices $D(\Hilb_A)\subseteq \text{Herm}(\Hilb_A)$, equipped with the Löwner ordering, is a pointed directed-complete partial order (pointed DCPO).

\begin{prop}
 \label{prop:dcpo_dens_mat}
$(D(\Hilb_A),\leq)$ is a pointed DCPO.
\end{prop}
\begin{proof}
For $N,P\in D(\Hilb_A)$, notice that $N\leq P$ is and only if for all $\ket{x}\in \Hilb_A$, $\bra{x}N\ket{x}\leq \bra{x}P\ket{x}$. Moreover, $\bra{x}N\ket{x}\leq \| \ket{x}\|^2$, using diagonalization and the fact that $Tr(N)\leq 1$. Then, given a directed subset $\{M_k\}_k$ of $D(\Hilb_A)$, let $M$ be the matrix defined by
\begin{align}
 \label{eq:def_sup}
\forall\, \ket{x} \in \Hilb_A, \ \bra{x}M\ket{x}=\sup_k \bra{x}M_k\ket{x}.
\end{align}
The supremum is clearly well defined for each $\ket{x}$. We check that $M$ is well defined: suppose $M$ is a linear map satisfying the above expression. Then for all $\ket{x},\ket{y}\in \Hilb_A$, we have $\bra{x}M\ket{y} = B(\ket{x},\ket{y})$ where:
\[
 B(\ket{x},\ket{y}) =\frac{1}{2}\left[f(\ket{x}+\ket{y})+if(i\ket{x}+\ket{y})-(f(\ket{x})+f(\ket{y}))(1+i)\right],
\]
with $f(\ket{z})\triangleq\sup_k \bra{z}M_k\ket{z}$.  Conversely, $B$ defines a Hermitian sesquilinear form, for which the corresponding matrix $M$ satisfies Equation~\ref{eq:def_sup}. To see that $B$ is sesquilinear, use the fact that the suprema can be replaced by limits for a same sequence of elements of $\{M_k\}_{k}$, due to it being a directed set.
Therefore $M$ is well defined.

Thus, $M$ is the supremum of $\{M_k\}_k$ in Herm($\Hilb_A$).
$M$ is also an element of $D(\Hilb_A)$: $M$ is a positive semi-definite Hermitian matrix by definition, and one can check that the trace is a Scott-continuous function from Herm$(\Hilb_A)$ to $\mathbb{R}$ (equipped with the usual ordering), and therefore $Tr(M)=Tr(\bigvee_k M_k)=\bigvee_k Tr(M_k) \leq 1$.

Finally, $(D(\Hilb_A),\leq)$ has least element 0 and is therefore a pointed DCPO.
\end{proof}

\begin{rem}
 \label{rem:topological_limit}
 If $M$ is the supremum of a directed subset $\{M_k\}_k$ of $D(\Hilb_A)$, then it is the topological limit of a sequence of elements in that directed set, with respect to the standard Euclidian topology on $\mathbb{C}^{d_A\times d_A}$, where $d_A$ is the dimension of $\Hilb_A$.
\end{rem}

The most general evolutions of density matrices are given by \emph{quantum operations}.
A quantum operation from $\Hilb_A$ to $\Hilb_B$ is a higher-order linear map from $\mathcal{L}(\Hilb_A)$ to $\mathcal{L}(\Hilb_B)$ that is completely positive and trace non-increasing. These notions are defined as follows: $\mathcal{C}\in\mathcal{L}(\mathcal{L}(\Hilb_A), \mathcal{L}(\Hilb_B))$ is \emph{completely positive} if it preserves positive semi-definiteness when tensored with the identity.
That is, for all finite-dimensional Hilbert space $\Hilb_E$ and for all positive semi-definite $\rho\in \mathcal{L}(\Hilb_A\otimes \Hilb_E)$, the operator $(\mathcal{C}\otimes \mathcal{I}_E)(\rho)$ is also positive semi-definite, where $\mathcal{I}_E$ is the identity map over $\mathcal{L}(\Hilb_E)$. $\mathcal{C}$ is \emph{trace non-increasing} if for all positive semidefinite $\rho\in \mathcal{L}(\Hilb_A)$, we have $Tr(\mathcal{C}(\rho))\leq Tr(\rho)$.
A \emph{quantum channel} from $\Hilb_A$ to $\Hilb_B$ is a quantum operation from $\Hilb_A$ to $\Hilb_B$ that is trace-preserving. Intuitively, quantum channels map probability distributions of quantum states (represented as density matrices) to probability distributions, whereas quantum operations map probability distributions of quantum states to probability subdistributions.

Throughout this paper, we will use the standard font $F$, $G$, $\ldots$ to denote linear maps, except when referring to higher-order maps such as quantum channels and operations, where we will use calligraphic letters $\mathcal{C}$, $\mathcal{D}$, $\mathcal{C}_0$, $\mathcal{C}_1$, $\dots$ In particular, we write the identity map over $\Hilb_A$ as $I_A$, and the identity channel over $\Hilb_A$ (i.e., the identity map over $\mathcal{L}(\Hilb_A)$) as $\mathcal{I}_A$.

Quantum operations and quantum channels form symmetric monoidal categories, which we write as $\tb{QO}$ and $\tb{QC}$, respectively. In both $\tb{QO}$ and $\tb{QC}$, objects are finite-dimensional Hilbert spaces. In $\tb{QO}$, the morphisms $\tb{QO}(\Hilb_A, \Hilb_B)$ are given by quantum operations from $\Hilb_A$ to $\Hilb_B$; and in $\tb{QC}$, the morphisms $\tb{QC}(\Hilb_A, \Hilb_B)$ are given by quantum channels from $\Hilb_A$ to $\Hilb_B$. In both categories, the composition rule corresponds to the usual composition of linear maps, and the categorical operation $\otimes$ coincides with the linear-algebraic tensor product.

The Löwner ordering is extended to quantum operations as follows: for all $\mathcal{C},\mathcal{D}\in\tb{QO}(\Hilb_A, \Hilb_{B})$, $\mathcal{C}\leq \mathcal{D}$ if for all Hilbert space $\Hilb_E$ and for all $\rho\in D(\Hilb_A\otimes \Hilb_E)$, we have $(\mathcal{C}\otimes \mathcal{I}_E)(\rho)\leq(\mathcal{D}\otimes \mathcal{I}_E)(\rho)$.

\begin{prop}
 \label{prop:dcpo_q_operations}
 $(\tb{QO}(\Hilb_A, \Hilb_{B}),\leq)$ is a pointed DCPO. Moreover, the supremum of a directed subset is given by the pointwise supremum, taken with respect to the Löwner ordering on Hermitian matrices.
\end{prop}
\begin{proof}
Let $\{\mathcal{C}_k\}_k$ be a directed subset of $\tb{QO}(\Hilb_A, \Hilb_B)$. Then for all $\rho\in D(\Hilb_A)$, $\{\mathcal{C}_k(\rho)\}_k$ is a directed subset of $(D(\Hilb_B),\leq)$. Consider the pointwise supremum $\mathcal{C}:\rho \in D(\Hilb_A) \mapsto \bigvee_k \mathcal{C}_k(\rho)$, which is well defined by Proposition~\ref{prop:dcpo_dens_mat}. We show that $\mathcal{C}$ is the supremum of $\{\mathcal{C}_k\}_k$ in $\tb{QO}(\Hilb_A,\Hilb_B)$. First, $\mathcal{C}$ is a linear map from Herm$(\Hilb_A)$ to Herm($\Hilb_{B}$). Indeed, for all states $\rho, \sigma\in D(\Hilb_A)$ and $\lambda\in \mathbb{C}$ such that $\rho+\lambda \sigma\in D(\Hilb_A)$, we have:
\[
 \mathcal{C}(\rho+\lambda \sigma)  = \bigvee_k \mathcal{C}_k(\rho +\lambda \sigma ) = \bigvee_k \mathcal{C}_k(\rho)+\lambda \mathcal{C}_k(\sigma)
 \leq  \bigvee_{k,k'} \mathcal{C}_k(\rho)+\lambda \mathcal{C}_{k'}(\sigma),
\]
where the last supremum is well defined because $\{\mathcal{C}_k(\rho)+\lambda \mathcal{C}_{k'}(\rho)\}_{k,k'}$ is a directed subset of $(D(\Hilb_{B}),\leq)$. Moreover since $\{\mathcal{C}_k\}_k$ is a directed subset, for each pair $k,k'$, there exists and index $l$ such that $\mathcal{C}_l\geq \mathcal{C}_k$ and $\mathcal{C}_l\geq\mathcal{C}_{k'}$, and therefore $\mathcal{C}_l(\rho)+\lambda \mathcal{C}_l(\sigma)\geq\mathcal{C}_k(\rho)+\lambda \mathcal{C}_{k'}(\sigma)$.
Consequently, the last inequality is an equality, and $\mathcal{C}(\rho+\lambda \sigma) = \mathcal{C}(\rho)+\lambda \mathcal{C}(\sigma)$, and therefore $\mathcal{C}$ is linear on $D(\Hilb_A)$. Since $D(\Hilb_A)$ spans Herm($\Hilb_A$), $\mathcal{C}$ can be extended to a linear map on Herm$(\Hilb_A)$.
$\mathcal{C}$ is trace non-increasing: for all $\rho\in D(\Hilb_A)$, we have $Tr(\mathcal{C}(\rho))=Tr\left(\bigvee_k\mathcal{C}_k(\rho)\right) = \bigvee_k Tr(\mathcal{C}_k(\rho)) \leq Tr(\rho)$ by Scott-continuity of the trace.
$\mathcal{C}$ is completely positive: for all auxiliary Hilbert space $\Hilb_E$ and $\rho\in D(\Hilb_A\otimes \Hilb_E)$, we have $(\mathcal{C}\otimes \mathcal{I}_E)(\rho) = \bigvee_k (\mathcal{C}_k\otimes \mathcal{I}_E)(\rho)$.
By complete positivity of $\mathcal{C}_k$, $(\mathcal{C}_k\otimes \mathcal{I}_E)(\rho)$ is positive for all $k$, and therefore $\mathcal{C}$ is completely positive. Therefore $\mathcal{C}\in\tb{QO}(\Hilb_A, \Hilb_{B})$. To see that $\mathcal{C}$ is the supremum of $\{\mathcal{C}_k\}_k$, for all Hilbert space $\Hilb_E$ and $\rho\in D(\Hilb_A\otimes \Hilb_E)$ we have $(\mathcal{C}\otimes \mathcal{I}_E)(\rho) = \bigvee_k(\mathcal{C}_k\otimes \mathcal{I}_E)(\rho) \geq (\mathcal{C}_k\otimes \mathcal{I}_E)(\rho)$, and therefore $\mathcal{C}\geq\mathcal{C}_k$ for all $k$.

Lastly, $\tb{QO}(\Hilb_A,\Hilb_{B})$ has least element $0$, and is therefore a pointed DCPO.
\end{proof}

\begin{lem}
 \label{lem:topological_limit_channel}
 If $\mathcal{C}$ is the supremum of a directed subset $\{\mathcal{C}_k\}_k$ of $\tb{QO}(\Hilb_A, \Hilb_{B})$, then it is the limit of a sequence of elements in that directed subset, with respect to the Euclidian topology.
\end{lem}

\begin{proof}
 By Proposition~\ref{prop:dcpo_q_operations}, for all quantum state $\rho$, we have $\mathcal{C}(\rho)=\bigvee_k \mathcal{C}_k (\rho)$. Let $\rho_1,\ldots,\rho_N$ be a finite set of states that spans $\mathcal{L}(\Hilb_A)$.
 By Remark~\ref{rem:topological_limit}, for all $i\in\llbracket 1,N\rrbracket$, there exists a sequence of indices $(k_i(n))_n$ such that $\mathcal{C}(\rho_i)$ is given by the limit $\mathcal{C}(\rho_i)=\lim_{n\rightarrow +\infty} \mathcal{C}_{k_i(n)} (\rho_i)$ in the standard Euclidian topology. For all $n$, choose $k(n)$ such that $\forall i\in\llbracket 1,N\rrbracket : \mathcal{C}_{k(n)}\geq\mathcal{C}_{k_{i}(n)}$ (such an index always exists because $\{\mathcal{C}_k\}_k$ is a directed set). Then for all $i$, $\mathcal{C}(\rho_i)=\lim_{n\rightarrow +\infty} \mathcal{C}_{k(n)} (\rho_i)$. In $\mathcal{L}(\mathcal{L}(\Hilb_A), \mathcal{L}(\Hilb_{B}))$, a finite-dimensional Hilbert space, all norms are equivalent and thus induce the same topology (the Euclidian topology). To prove the result, we choose to show that $\mathcal{C}$ is the limit of $(\mathcal{C}_{k(n)})_n$ for the operator norm $\op{\cdot}$ induced by the Frobenius norm $\|\cdot\|_2$.
 For all $n$:
 \[
  \op{\mathcal{C} - \mathcal{C}_{k(n)}} = \sup_{\substack{X\in \mathcal{L}(\Hilb_A): \\ \|X\|_2=1}} \|(\mathcal{C} - \mathcal{C}_{k(n)})(X)\|_2.
 \]
For all $X\in \mathcal{L}(\Hilb_A)$, with $\lambda_1,\ldots,\lambda_N\in \mathbb{C}$ such that $X=\sum_i\lambda_i \rho_i$, we have $\|(\mathcal{C} - \mathcal{C}_{k(n)})(X)\|_2 \leq \sum_{i=1}^N |\lambda_i| \|(\mathcal{C} - \mathcal{C}_{k(n)})(\rho_i)\|_2 \xrightarrow[n\rightarrow +\infty]{} 0$. Let $X_1,\ldots,X_N$ be an orthonormal basis of $\mathcal{L}(\Hilb_A)$, with respect to the Frobenius inner product. Then for all $X\in\mathcal{L}(\Hilb_A)$ with $\|X\|_2=1$, $\|(\mathcal{C} - \mathcal{C}_{k(n)})(X)\|_2\leq \sqrt{N} \max_i \|(\mathcal{C}-\mathcal{C}_{k(n)})(X_i)\|_2$.
Therefore $\mathcal{C}_{k(n)}\xrightarrow[]{\op{\cdot}} \mathcal{C}$.
\end{proof}

Finally, the Kraus representation is a convenient way to describe quantum channels and operations. Kraus's theorem states that any quantum operation $\mathcal{C}\in\tb{QO}(\Hilb_A, \Hilb_B)$ can be written as $\mathcal{C}:\rho\mapsto \sum_i K_i \rho K_i^{\dag}$, where the $K_i$ are linear maps in $\mathcal{L}(\Hilb_A, \Hilb_B)$ satisfying $\sum_{i}K_i^{\dag}K_i\leq I_A$. Moreover $\sum_{i}K_i^{\dag}K_i = I_A$ if and only if $\mathcal{C}$ is a quantum channel.

Because they form a pointed DCPO, quantum operations are standard candidates to define the denotational semantics of quantum programs in the absence of quantum control. However, they are insufficient to defining the behavior of the $\tb{qcase}$ operation. To see this, consider the following. The quantum case is defined on unitary matrices as $(U,V)\mapsto \ket{0}\bra{0}\otimes U +\ket{1}\bra{1}\otimes V
$. To extend this function to quantum operations, we would like to define an operation that maps the corresponding unitary channels $\mathcal{C}_U \triangleq U(\cdot)U^{\dag}$ and $\mathcal{C}_V\triangleq V(\cdot)V^{\dag}$ to $W(\cdot)W^{\dag}$ with $W\triangleq\ket{0}\bra{0}\otimes U +\ket{1}\bra{1}\otimes V$.
However, it is not possible to define such a process~\cite{grenoble_coherent_control}. Indeed, with $U\triangleq I$ and $V\triangleq e^{i\theta}I$ for some phase $\theta\in[0,2\pi[$, the corresponding channels are $\mathcal{C}_U=\mathcal{C}_V=\mathcal{I}$, regardless of $\theta$. However, the channel $W(\cdot)W^{\dag}$ does depend on $\theta$. This is because the global phase of $V$ becomes a local phase of $W$, and thus becomes observable. Therefore it is not possible to define such a protocol on quantum operations.

In order to define an adequate quantum case operation, it is necessary to specify more information about the physical implementation of each input channel, similarly to the operational semantics. This can be achieved using the framework of \emph{vacuum extensions}, introduced in~\cite{kristjansson,kristjansson_1}.

\subsection{Vacuum Extensions and Coherent Quantum Operations}
\label{section:vacuum}

The idea of vacuum extensions is to extend a quantum operation to describe its action when it receives a regular input, when its input is empty, and when it receives any superposition of regular and empty inputs.
The concept of ``empty input'' is formalized by adding an extra basis state, written as $\ket{\text{vac}}$.
Thus, for any quantum operation $\mathcal C\in \tb{QO}(\Hilb_A, \Hilb_B)$, one can consider a vacuum extension $ \tilde{\mathcal C}\in \tb{QO}(\Hilb_A\oplus \text{Vac}, \Hilb_B\oplus \text{Vac})$ where $\text{Vac}$ is the one-dimensional Hilbert space spanned by $\ket {\text{vac}}$, and $\oplus$ is the direct sum (in particular, $\mathcal{H}_A\subseteq \mathcal{H}_A\oplus \text{Vac}$ and $\mathcal{H}_B\subseteq \mathcal{H}_B\oplus \text{Vac}$). A vacuum extension of $\mathcal{C}$ is not an arbitrary quantum evolution between these larger Hilbert spaces—it is required to map the vacuum state to the vacuum state ($\tilde{\mathcal{C}}(\ket {\text{vac}}\bra {\text{vac}}) =\ket {\text{vac}}\bra {\text{vac}}$) and must coincide with $\mathcal{C}$ on non-vacuum states ($\forall \ket \phi \in \Hilb_A, \tilde{\mathcal{C}}(\ket \phi\bra \phi) = \mathcal{C}(\ket \phi\bra \phi)$). Accordingly, we define the category of vacuum-extended quantum operations:

\begin{defi}\label{def:ve}
Let \tb{QO}$_\text{\bf vac}$ be the category of \emph{vacuum-extended quantum operations}, whose objects are finite-dimensional Hilbert spaces, and whose morphisms are defined as follows. A vacuum-extended quantum operation $\mathcal{D}$ from $\Hilb_A$ to $\Hilb_B$ is a quantum operation $\mathcal D\in \tb{QO}(\Hilb_A\oplus \text{Vac}, \Hilb_B\oplus \text{Vac})$ such that $\mathcal{D}(\ket {\text{vac}}\bra {\text{vac}}) = \ket {\text{vac}}\bra {\text{vac}}$ and $\forall \ket{\phi} \in \Hilb_A, \bra{\text{vac}}\mathcal D(\ket\phi \bra \phi)\ket{\text{vac}} = 0$.\footnote{The composition of morphisms in \tb{QO}$_\text{\bf vac}$ is nothing but the composition in  \tb{QO}.}
\end{defi}

Vacuum-extended quantum operations can be characterized in the following way.
For any $\mathcal C\in \tb{QO}(\Hilb_A, \Hilb_B)$ and $F\in \mathcal L(\Hilb_A, \Hilb_B)$, one can define $\tilde{\mathcal{C}}^F: \mathcal{L}(\mathcal{H}_A\oplus \text{Vac})\rightarrow\mathcal{L}(\mathcal{H}_B\oplus\text{Vac})$~as
\[ \tilde{\mathcal{C}}^F(\rho) \triangleq   \mathcal{C}(P_A\rho P_A^\dag) + \ket {\text{vac}}\bra {\text{vac}}\rho\ket{\text{vac}}\bra{\text{vac}}
+ FP_A\rho \ket {\text{vac}}\bra {\text{vac}}+ \ket {\text{vac}}\bra {\text{vac}}\rho P_A^\dag F^\dagger,
\]
where $P_A\in\mathcal{L}(\mathcal{H}_A\oplus\text{Vac},\mathcal{H}_A)$ is the orthogonal projector onto $\Hilb_A$.
Notice that $\tilde{\mathcal{C}}^F(\ket {\text{vac}}\bra {\text{vac}}) = \ket {\text{vac}}\bra {\text{vac}}$, and $\forall \ket{\phi} \in \Hilb_A$, $\bra{\text{vac}}\tilde{\mathcal{C}}^F(\ket\phi \bra \phi)\ket{\text{vac}} = 0$.
While for arbitrary $\mathcal C$ and $F$, $\tilde{\mathcal{C}}^F$ is not necessarily a completely positive map, any vacuum-extended quantum operation is of this form: for all $\mathcal D \in \tb{QO}_\text{\bf vac} (\Hilb_A, \Hilb_B)$, there exist unique  $\mathcal{C} \in \tb{QO}(\Hilb_A, \Hilb_B)$ and $F\in \mathcal L(\Hilb_A, \Hilb_B)$ (generally called transformations matrix) such that $\mathcal D = \tilde{\mathcal C}^F$~\cite{grenoble_coherent_control}.
We call such a pair $(\mathcal C,F)$ a \emph{coherent quantum operation}.

In particular, for a fixed $\mathcal{C}\in \tb{QO}(\Hilb_A, \Hilb_B)$, its vacuum extension $\tilde{\mathcal{C}} \in \tb{QO}_\text{\bf vac} (\Hilb_A, \Hilb_B)$, which corresponds to a specific choice of $F$, is generally not unique.

\begin{defi}[Coherent quantum operations]\label{def:cqo}
Let $\tb{CQO}$ be the category of \emph{coherent quantum operations}, whose objects are finite-dimensional Hilbert spaces and whose morphisms are defined as follows.  A coherent quantum operation from $\Hilb_A$ to $\Hilb_B$ is a pair $(\mathcal C, F)$ with $\mathcal C\in   \tb{QO} (\Hilb_A, \Hilb_B)$ and $F\in \mathcal L(\Hilb_A, \Hilb_B)$ such that $\tilde{\mathcal{C}}^F\in  \tb{QO}_\text{\bf vac} (\Hilb_A, \Hilb_B)$ is a vacuum extension. Composition is defined pointwise: \[(\mathcal C_1,F_1)\circ  (\mathcal C_0,F_0) \triangleq  (\mathcal C_1 \circ \mathcal C_0,F_1 F_0).\]
\end{defi}

The following lemma shows that composition in $\tb{CQO}$ matches the usual notion of composition for vacuum-extended operations.
\begin{lem}
\label{lem:expr_vac_ext_comp}
Let $(\mathcal{C},F)\in\tb{CQO}(\mathcal{H}_A,\mathcal{H}_B)$ and $(\mathcal{D},G)\in\tb{CQO}(\mathcal{H}_B,\mathcal{H}_C)$. For $(\mathcal{E},H) \triangleq(\mathcal{D},G)\circ (\mathcal{C},F)$, we have $\tilde{\mathcal{E}}^H = \tilde{\mathcal{D}}^G \circ \tilde{\mathcal{C}}^F\!.$
\end{lem}

We use coherent quantum operations to define the denotational semantics of our language. Indeed, a coherent quantum operation $(\mathcal C, F)$ provides an intuitive representation: $\mathcal C$ is the quantum operation representing the action of the program on non-vacuum inputs and $F$, roughly speaking, corresponds to the evolution of the default transition from the operational semantics. This intuition will be made formal by the adequacy theorem (Theorem~\ref{thm:adequacy}).

In the following, we consider a few illustrative examples of coherent quantum operations.

 \paragraph{Unitary evolution.} It is well known that the quantum operation associated with a unitary map $U$ is $\mathcal{C}_U\triangleq \rho \mapsto U\rho U^\dagger$, or simply $U(\cdot) U^\dagger$. Thus a coherent quantum operation for a unitary map must be $(\mathcal C_U,F)$ for some $F$. One can show that the only $F$ that makes $\tilde{\mathcal{C}}_U^F$ a vacuum-extended quantum operation is $U$ (up to a global phase). Intuitively, this is because in this case, the default transition is the only possible evolution, which consists in applying the unitary $U$. As a consequence, a coherent quantum operation corresponding to the application of a unitary map is of the form $(U(\cdot) U^\dagger, U)$.

\paragraph{Initialization.} Similarly to the unitary case, the initialization of a qubit in a $\ket \phi$ state is deterministic, as a consequence the transformation matrix must be $\ket \phi$ (up to a global phase). Thus the coherent quantum operations corresponding to state initialization are of the form $(\ket \phi\bra \phi, \ket \phi)$.

 \paragraph{Discard.} The quantum operation associated with discarding is the trace, so the corresponding coherent quantum operation is $(Tr(\cdot),\bra \phi)$ where $\ket \phi$ is any arbitrary normalized  state. This aligns with the intuition that discarding a qubit can be implemented as a measurement in any arbitrary orthonormal  basis $\{\ket {\phi}, \ket {\phi^\perp}\}$ (since $\bra \phi \rho\ket \phi + \bra {\phi^\perp}\rho\ket {\phi^\perp} = Tr(\rho)$) for which the first possible outcome ($\ket \phi$) is chosen as the default transition.

\paragraph{Measurement and classical control.} We now consider a scenario where a qubit $\q$ is measured  in the standard basis, then depending on the classical outcome, one of two coherent quantum operations $(\mathcal C_0,F_0),(\mathcal C_1,F_1)\in \tb{CQO}(\mathcal{H}_{\q,\Gamma},\mathcal{H}_{\Gamma'})$ is applied, for some environments $\Gamma,\Gamma'$. The corresponding quantum operation is $\mathcal C_0\circ \mathcal{P}_0^\q + \mathcal C_1\circ \mathcal{P}_1^\q$ where $ \mathcal{P}_k^\q \triangleq \ket k\bra k_\q(\cdot) \ket k\bra k_\q$ correspond to the two possible measurement outcomes. One can choose for instance $F_0\ket0\bra 0_\q$ as a transformation matrix, leading to the coherent quantum operation
\begin{align}\overline{\tb{meas}}_{\q}[(\mathcal C_0,F_0),(\mathcal C_1,F_1)] \triangleq (\mathcal C_0\circ \mathcal{P}_0^\q + \mathcal C_1\circ \mathcal{P}_1^\q,F_0\ket0\bra 0_\q).\end{align}

To justify the definition, we check that the result is indeed a coherent quantum operation.
\begin{lem}
\label{lem:expr_vac_ext_meas}
Let $(\mathcal{C}_0,F_0)$, $(\mathcal{C}_1,F_1)\in\tb{CQO}(\mathcal{H}_{\q,\Gamma},\mathcal{H}_{\Gamma'})$. If $(\mathcal{E},H)=\overline{\tb{meas}}_{\q}[(\mathcal{C}_0,F_0),(\mathcal{C}_1,F_1)]$, then for all $\rho\in D(\Hilb_{\q,\Gamma} \oplus \text{\emph{Vac}})$,
 \[\tilde{\mathcal{E}}^H (\rho) = \tilde{\mathcal{C}_0}^{F_0}\circ \mathcal{T}_0^\q (\rho) +  \tilde{\mathcal{C}_1}^{F_1}\circ\mathcal{T}_1^\q(\rho),\]
where:
\begin{align*}
 \mathcal{T}_0^\q &\triangleq \left(\ket{0}\bra{0}_\q \oplus \ket{\text{\emph{vac}}}\bra{\text{\emph{vac}}}\right)(\cdot) \left(\ket{0}\bra{0}_\q \oplus \ket{\text{\emph{vac}}}\bra{\text{\emph{vac}}}\right)\!,
 \\
 \mathcal{T}_1^\q &\triangleq \ket{1}\bra{1}_\q\, (\cdot) \, \ket{1}\bra{1}_\q.
\end{align*}
\end{lem}
The proof is straightforward. This shows that the $\overline{\tb{meas}}_\q$ operation amounts to performing the measurement defined by the Kraus operators
$\bigl\{\ket{0}\bra{0}_\q \oplus \ket{\text{vac}}\bra{\text{vac}}, $ $\ket{1}\bra{1}_\q \bigr\}$ on $\Hilb_{\q,\Gamma}\oplus \text{Vac}$, and applying $\tilde{\mathcal{C}}_0^{F_0}$ or $\tilde{\mathcal{C}}_1^{F_1}$ depending on the outcome. The fact that $\overline{\tb{meas}}_\q$ preserves coherent quantum operations follows from this: the expression of $\tilde{\mathcal{E}}^H$ defines a vacuum-extended operation.

\paragraph{Quantum case with unitaries.}  We first consider a quantum case scenario on unitaries, which can be described using pure operations only. The quantum case of two unitaries $U$ and $V$ with a control qubit $\q$ corresponds to the unitary $\ket{0}\bra{0}_\q\otimes U + \ket{1}\bra{1}_\q\otimes V$. We now describe this scenario using coherent quantum operations
$(\mathcal C_U,U),(\mathcal C_V,V)\in\tb{CQO}(\mathcal{H}_{\Gamma},\mathcal{H}_{\Gamma})$ representing the unitary evolutions $U$ and $V$ respectively (where $\mathcal C_W = \rho\mapsto  W\rho W^\dagger$ and $\Gamma$ is an environment). The resulting coherent quantum operation $\overline{\tb{qcase}}_{\q}[(\mathcal C_U,U),(\mathcal C_V,V)] = (\mathcal D,F)$ must represent the unitary evolution $\ket{0}\bra{0}_\q\otimes U + \ket{1}\bra{1}_\q\otimes V$. Thus we can choose $F = \ket{0}\bra{0}_\q\otimes U + \ket{1}\bra{1}_\q\otimes V$ (which is the only possible choice up to a global phase), and $\mathcal D$ must be equal to $\mathcal C_{\ket{0}\bra{0}_\q\otimes U + \ket{1}\bra{1}_\q\otimes V}$. Then, for any $\rho$,
\begin{align*}
\mathcal D(\rho )=\ & (\ket{0}\bra{0}_\q\otimes U + \ket{1}\bra{1}_\q\otimes V)\rho (\ket{0}\bra{0}_\q\otimes U^\dagger  + \ket{1}\bra{1}_\q\otimes V^\dagger)\\
=\ & (\ket{0}\bra{0}_\q\otimes U)\rho( \ket{0}\bra{0}_\q\otimes U^\dagger )+( \ket{0}\bra{0}_\q\otimes U)\rho( \ket{1}\bra{1}_\q\otimes V^\dagger )\\
&+ (\ket{1}\bra{1}_\q\otimes V)\rho( \ket{0}\bra{0}_\q\otimes U^\dagger)+(\ket{1}\bra{1}_\q\otimes V)\rho (\ket{1}\bra{1}_\q\otimes V^\dagger)\\
=\ &  (\mathcal{P}_0^\q\otimes \mathcal{C}_U + \mathcal{P}_1^\q \otimes \mathcal{C}_V) (\rho)
+ (\ket{0}\bra{0}_\q{\otimes} U)\rho (\ket{1}\bra{1}_\q{\otimes} V^\dagger) + (\ket{1}\bra{1}_\q{\otimes} V)\rho (\ket{0}\bra{0}_\q{\otimes} U^\dagger).
\end{align*}

We notice that $\mathcal D$ is the linear combination of $\mathcal{P}_0^\q \otimes \mathcal{C}_U + \mathcal{P}_1^\q \otimes \mathcal{C}_V$ with two extra terms depending on $U$ and $V$. To sum up,
\begin{align*}
 \overline{\tb{qcase}}_{\q}[(\mathcal{C}_U,U),(\mathcal{C}_V,V)] \triangleq \, & \bigl(\mathcal{P}_0^\q \otimes \mathcal{C}_U + \mathcal{P}_1^\q \otimes \mathcal{C}_V + (\ket{0}\bra{0}_\q \otimes U)(\cdot)(\ket{1}\bra{1}_\q\otimes V^{\dag}) \\
  &+(\ket{1}\bra{1}_\q \otimes V)(\cdot)(\ket{0}\bra{0}_\q\otimes U^{\dag}), \ket{0}\bra{0}_\q\otimes U + \ket{1}\bra{1}_\q\otimes V\bigr).
\end{align*}

\paragraph{Arbitrary quantum case.}
We now consider the quantum case of two arbitrary coherent quantum operations $(\mathcal C_0,F_0),(\mathcal C_1,F_1)\in\tb{CQO}(\mathcal{H}_{\Gamma},\mathcal{H}_{\Gamma'})$. The unitary case described above can be extended in the straightforward way as follows:
\begin{align*}
 \overline{\tb{qcase}}_{\q}[(\mathcal{C}_0,F_0),(\mathcal{C}_1,F_1)] \triangleq \,& \bigl(\mathcal{P}_0^\q \otimes \mathcal{C}_0 + \mathcal{P}_1^\q \otimes \mathcal{C}_1 + (\ket{0}\bra{0}_\q \otimes F_0)(\cdot)(\ket{1}\bra{1}_\q\otimes F_1^{\dag}) \\
  &+(\ket{1}\bra{1}_\q \otimes F_1)(\cdot)(\ket{0}\bra{0}_\q\otimes F_0^{\dag}), \ket{0}\bra{0}_\q\otimes F_0 + \ket{1}\bra{1}_\q\otimes F_1\bigr).
\end{align*}

For more clarity, if we write the input state $\rho$ as a block matrix $\rho=
 \begin{bmatrix}
  \rho_1 & \!\!\!\rho_2 \\
  \rho_3 & \!\!\!\rho_4
 \end{bmatrix}
$, where the first qubit is assumed to be the control qubit $\q$, then the $\overline{\tb{qcase}}_\q$ operation can be rewritten~as:
\[
 \overline{\tb{qcase}}_{\q}[(\mathcal{C},F),(\mathcal{D},G)] = \left(\begin{bmatrix}
  \rho_1 & \!\!\!\rho_2 \\
  \rho_3 & \!\!\!\rho_4
 \end{bmatrix} \!\mapsto\! \begin{bmatrix}
  \mathcal{C}(\rho_1) & \!\!\!F\rho_2G^{\dag} \\
  G\rho_3F^{\dag} & \!\!\!\mathcal{D}(\rho_4)
 \end{bmatrix},
 \begin{bmatrix}
  F & \!\!0 \\
  0 & \!\!G
 \end{bmatrix}
 \right)\!.
\]

One can verify that if $(\mathcal C_0,F_0)$ and  $(\mathcal C_1,F_1)$ are  coherent quantum operations,  then so is $\overline{\tb{qcase}}_{\q}[(\mathcal{C}_0,F_0),$ $(\mathcal{C}_1,F_1)]$:
\begin{lem}
 Let $(\mathcal{C}_0,F_0),(\mathcal{C}_1,F_1)\in\tb{CQO}(\mathcal{H}_\Gamma,\mathcal{H}_{\Gamma'})$ and $\q\notin\Gamma\cup\Gamma'$. Then $\overline{\tb{qcase}}_\q[(\mathcal{C}_0,F_0),$ $(\mathcal{C}_1,F_1)]$ $\in \tb{CQO}(\mathcal{H}_{\q,\Gamma},\mathcal{H}_{\q,\Gamma'})$.
\end{lem}
\begin{proof}
Let $\{K_i \oplus \nu_i \ket{\text{vac}}\bra{\text{vac}}\}_i$ and $\{L_j \oplus \mu_j \ket{\text{vac}}\bra{\text{vac}}\}_j$ be Kraus decompositions of $\tilde{\mathcal{C}_0}^{F_0}$ and $\tilde{\mathcal{C}_1}^{F_1}$, respectively. Note that Kraus decompositions of vacuum extensions necessarily have this form~\cite{kristjansson}.
We write $(\mathcal{D},G)= \overline{\tb{qcase}}_\q[(\mathcal{C}_0,F_0),(\mathcal{C}_1,F_1)]$.
Then, one can check that
\[
\bigl\{\bigl(\ket{0}\bra{0} \otimes \mu_jK_i + \ket{1}\bra{1} \otimes \nu_iL_j\bigr)\oplus\nu_i\mu_j\ket{\text{vac}}\bra{\text{vac}}\bigr\}_{i,j}
\]
is a Kraus decomposition of $\tilde{\mathcal{D}}^G$. It follows that $\tilde{\mathcal{D}}^G\in\tb{QO}_{\text{\bf{vac}}}(\mathcal{H}_{\q,\Gamma},\mathcal{H}_{\q,\Gamma'})$ and therefore that $(\mathcal{D},G)\in\tb{CQO}(\mathcal{H}_{\q,\Gamma},\mathcal{H}_{\q,\Gamma'})$.
\end{proof}

Notice that the quantum case satisfies the following interchange property:
\begin{prop}[Sequentiality]
\label{prop:factorize_qcase}
For any coherent quantum operations $(\mathcal{C}_0,F_0),(\mathcal{C}_1,F_1)\in\tb{CQO}(\mathcal{H}_\Gamma,\mathcal{H}_\Gamma')$ and $(\mathcal{D}_0,G_0)$, $(\mathcal{D}_1,G_1)\in\tb{CQO}(\mathcal{H}_{\Gamma'},\mathcal{H}_{\Gamma''})$ and for any $\q\notin \Gamma\cup\Gamma'\cup\Gamma''$, we have:
\[
 \begin{array}{c}
  \overline{\tb{qcase}}_{\q}((\mathcal{D}_0,G_0),(\mathcal{D}_1,G_1))\circ \overline{\tb{qcase}}_{\q}((\mathcal{C}_0,F_0),(\mathcal{C}_1,F_1))
  \\
  =\overline{\tb{qcase}}_{\q}((\mathcal{D}_0,G_0)\circ (\mathcal{C}_0,F_0),(\mathcal{D}_1,G_1)\circ (\mathcal{C}_1,F_1)).
 \end{array}
\]
\end{prop}

\paragraph{General recursion.} Finally, in order to define the semantics of loops, we need general recursion. We crucially show that  coherent quantum operations form a pointed DCPO:

\begin{prop}
\label{prop:cqoDCPO}
For all Hilbert spaces $\Hilb_A$, $\Hilb_B$, $(\tb{CQO}(\Hilb_A,\Hilb_B), \sqsubseteq )$ is a pointed DCPO with least element $(0,0)$, where $(\mathcal C,F)\sqsubseteq (\mathcal D,G)$ iff $\tilde C^F \leq  \tilde D^G$,  $\leq$ being the Löwner partial ordering.
\end{prop}
\begin{proof}
Let $\{(\mathcal{C}_k,F_k)\}_k$ be a directed subset of $\tb{CQO}(\mathcal{H}_A,\mathcal{H}_B)$.
Then, the set of corresponding vacuum-extended operations $\{\tilde{\mathcal{C}}_k^{F_k}\}_k$ is a directed subset of $\tb{QO}(\Hilb_A \oplus \text{Vac}, \Hilb_B\oplus \text{Vac})$, equipped with the Löwner ordering.
By Proposition~\ref{prop:dcpo_q_operations}, it has a supremum $\mathcal{Z}\in \tb{QO}(\Hilb_A \oplus \text{Vac},\Hilb_B\oplus \text{Vac})$, defined as the pointwise supremum with respect to the Löwner ordering on density matrices. We must show that $\mathcal{Z}$ is a vacuum-extended operation. For all $\rho \in D(\Hilb_A \oplus \text{Vac})$, $\mathcal{Z}(P_A\rho P_A) = \bigvee_k \tilde{\mathcal{C}}_k^{F_k}(P_A\rho P_{A})$, with $P_{A}$ the orthogonal projector onto $\Hilb_{A}$. For all $k$, since $\tilde{\mathcal{C}}_k^{F_k}$ is a vacuum extension, we have $\tilde{\mathcal{C}}_k^{F_k}(P_{A}\rho P_{A})\in \mathcal{L}(\Hilb_B)$. By Remark~\ref{rem:topological_limit}, $\mathcal{Z}(P_{A}\rho P_{A})$ is the limit of a sequence of elements of this form in the standard Euclidian topology. Since $\mathcal{L}(\Hilb_B)$ is a closed subset of $\mathcal{L}(\Hilb_B\oplus \text{Vac})$, we have $\mathcal{Z}(P_{A}\rho P_{A})\in D(\Hilb_B)$. Therefore, the restriction of $\mathcal{Z}$ to inputs in $St(\Hilb_A)$ is a quantum operation in $\tb{QO}(\Hilb_A, \Hilb_B)$. Moreover, we have $\mathcal{Z}(\ket{\text{vac}}\bra{\text{vac}})=\bigvee_k \tilde{\mathcal{C}}_k^{F_k}(\ket{\text{vac}}\bra{\text{vac}})=\ket{\text{vac}}\bra{\text{vac}}$.
Therefore $\mathcal{Z}$ is a vacuum extension, and its representation as a coherent quantum operation  $(\mathcal{C}_{\mathcal{Z}},F)$ (which is uniquely defined) is the supremum in $\tb{CQO}(\mathcal{H}_A,\mathcal{H}_B)$ of $\{(\mathcal{C}_k,F_k)\}_k$. Finally, the least element of $\tb{CQO}(\mathcal{H}_A,\mathcal{H}_B)$ is $(0,0)$.
\end{proof}

\subsection{Interpretation of Programs}

We can now define the denotational semantics of the language. Each well-formed statement will be interpreted as a coherent quantum operation. Formally, given a well-formed program $(\Gamma;S)$, we define the interpretation $\sem{S}_{\Gamma}$ as a coherent quantum operation in $\tb{CQO}(\Hilb_\Gamma, \Hilb_{\Gamma^S})$, written as $\tb{CQO}_{\Gamma;\Gamma^S}$ for short. The semantics of statements is defined inductively in Figure~\ref{fig:semantics_denotational}. The operation $Tr_{\q}$ denotes the partial trace over qubit $\q$, defined as $Tr_{\q}:\rho \mapsto \bra{0}_\q \rho \ket{0}_q +\bra{1}_\q \rho \ket{1}_\q$. As before, $U_{\q}$ denotes the application of the unitary $U$ to the qubit $\q$, and identity elsewhere. For the $\tb{meas}$ statement, the statements $S_0$ and $S_1$ in either branch can affect the control qubit $\q$, whereas those of $\tb{qcase}$ cannot. This is clear from the environment indices. In the interpretation of $\tb{while}$, lfp denotes the least fixed point operator and  $\mathscr{F}_{\q}^{S}$ is defined for all $(\mathcal{C},F)\in \tb{CQO}_{\q,\Gamma;\q,\Gamma}$ as:
\[
 \mathscr{F}_{\q}^{S}(\mathcal{C},F)\triangleq\overline{\tb{meas}}_{\q}\bigl[(\mathcal{I}_{\q,\Gamma},I_{\q,\Gamma}), (\mathcal{C},F)\circ \sem{S}_{\q,\Gamma}\bigr]\!.
 \]
 This reads as follows:
 after obtaining outcome 0 the loop is exited, and therefore the program does nothing; and after obtaining outcome 1 the program executes $S$ then executes the loop.

\begin{figure}
\hrulefill
\[
\begin{array}{rcl}
 \sem{\tb{skip}}_{\Gamma}&\triangleq&\bigl(\mathcal{I}_{\Gamma},I_{\Gamma}\bigr)
 \\[0.4ex]
 \sem{\tb{new qbit }\q}_{\Gamma}&\triangleq&\bigl(\ket{0}\bra{0}_{\q}\otimes \mathcal{I}_{\Gamma},\ket{0}_\q\otimes I_{\Gamma}\bigr)
 \\[0.4ex]
  \sem{\tb{discard }\q}_{\q,\Gamma}&\triangleq&\bigl(Tr_{\q},\bra{0}_\q\otimes I_{\Gamma}\bigr)
 \\[0.4ex]
  \sem{\q{*}{=}U}_{\q,\Gamma}&\triangleq&\bigl(U_{\q}(\cdot) U_{\q}^{\dag},U_{\q}\bigr)
 \\[0.4ex]
 \sem{S_0;S_1}_{\Gamma}&\triangleq&\sem{S_1}_{\Gamma^{S_0}}\circ \sem{S_0}_{\Gamma}
 \\[0.4ex]
 \sem{\tb{meas }\q \ (0\rightarrow S_0,1\rightarrow S_1)}_{\q,\Gamma}&\triangleq&\overline{\tb{meas}}_{\q}[\sem{S_0}_{\q,\Gamma},\sem{S_1}_{\q,\Gamma}]
 \\[0.4ex]
 \sem{\tb{while}\ \q \ \tb{do}\ S}_{\q,\Gamma}&\triangleq&\text{lfp}(\mathscr{F}_{\q}^{S})
 \\[0.4ex]
 \sem{\tb{qcase}\ \q\ (0\rightarrow S_0,1\rightarrow S_1)}_{\q,\Gamma} &\triangleq&\overline{\tb{qcase}}_{\q}[\sem{S_0}_{\Gamma},\sem{S_1}_{\Gamma} ]
\end{array}\]
\makebox[\textwidth][l]{where}
\begin{align*}
\overline{\tb{meas}}_{\q}[(\mathcal C_0,F_0),(\mathcal C_1,F_1)] &= (\mathcal C_0\circ \mathcal{P}_0^\q + \mathcal C_1\circ \mathcal{P}_1^\q,F_0\ket0\bra 0_\q)\\
\mathscr{F}_{\q}^{S}(\mathcal{C},F) &= \overline{\tb{meas}}_{\q}\bigl[(\mathcal{I}_{\q,\Gamma},I_{\q,\Gamma}), (\mathcal{C},F)\circ \sem{S}_{\q,\Gamma}\bigr]\\
 \overline{\tb{qcase}}_{\q}[(\mathcal{C}_0,F_0),(\mathcal{C}_1,F_1)] &=  \bigl(\mathcal{P}_0^\q \otimes \mathcal{C}_0 + \mathcal{P}_1^\q \otimes \mathcal{C}_1 + (\ket{0}\bra{0}_\q \otimes F_0)(\cdot)(\ket{1}\bra{1}_\q\otimes F_1^{\dag}) \\
  &\quad +(\ket{1}\bra{1}_\q \otimes F_1)(\cdot)(\ket{0}\bra{0}_\q\otimes F_0^{\dag}), \ket{0}\bra{0}_\q\otimes F_0 + \ket{1}\bra{1}_\q\otimes F_1\bigr)
\end{align*}
\hrulefill
\caption{Denotational Semantics}
\label{fig:semantics_denotational}
\end{figure}

It is not immediately clear that the denotational semantics is well defined. In particular, we must check that each operation preserves coherent quantum operations, and prove the existence of the least fixed point. To this end, we prove the following lemma.
\begin{lem}
\label{lem:scott_cont}
The following functions are Scott-continuous:
\begin{itemize}
  \item Fix $(\mathcal{D},G)\in\tb{CQO}_{\Gamma;\Gamma}$. The function $(\mathcal{C},F)\mapsto (\mathcal{C},F)\circ (\mathcal{D},G)$ from $\tb{CQO}_{\Gamma;\Gamma}$ to itself is Scott-continuous.
  \item Fix $(\mathcal{D},G)\in\tb{CQO}_{\q,\Gamma;\q,\Gamma}$. The function $(\mathcal{C},F)\mapsto \overline{\tb{meas}}_{\q}[(\mathcal{D},G),(\mathcal{C},F)]$ from $\tb{CQO}_{\q,\Gamma;\q,\Gamma}$ to itself is Scott-continuous.
 \end{itemize}
\end{lem}
\begin{proof}

To show that composition is Scott-continuous on the left, consider a directed subset $\{(\mathcal{C}_k,F_k)\}_k$ of $\tb{CQO}_{\Gamma;\Gamma}$. Let $(\mathcal{C},F)=\bigvee_k (\mathcal{C}_k,F_k)$ (well defined by Proposition~\ref{prop:cqoDCPO}) and $(\mathcal{E},H)= (\mathcal{C},F)\circ(\mathcal{D},G)$. Then by Lemma~\ref{lem:expr_vac_ext_comp}, $\tilde{\mathcal{E}}^H = \tilde{\mathcal{C}}^F \circ \tilde{\mathcal{D}}^G$. By definition, $\tilde{\mathcal{C}}^F=\bigvee_k \tilde{\mathcal{C}}_k^{F_k}$. By Lemma~\ref{lem:topological_limit_channel}, $\tilde{\mathcal{C}}^F$ is the limit of some sequence in $\{\tilde{\mathcal{C}}_k^{F_k}\}_k$, in the Euclidian topology. Then by continuity of the composition function $-\circ \tilde{\mathcal{D}^G}$, $\tilde{\mathcal{E}}^H$ is the topological limit of some sequence in $\{\tilde{\mathcal{C}}_k^{F_k}\circ \tilde{\mathcal{D}}^G\}_k$. Consequently, $\tilde{\mathcal{E}}^H \leq \bigvee_k\left( \tilde{\mathcal{C}}_k^{F_k} \circ \tilde{\mathcal{D}}^G\right)$.
 To obtain the reverse inequality, we use the fact that $-\circ \tilde{\mathcal{D}^G}$ is an increasing function. Therefore, as they correspond to the same vacuum-extended operation, $\left(\bigvee_k (\mathcal{C}_k,F_k)\right)\circ (\mathcal{D},G)=\bigvee_k\left( (\mathcal{C}_k,F_k)\circ (\mathcal{D},G)\right)$.

To show that measurement is Scott-continuous on the right, consider a directed subset $\{(\mathcal{C}_k,F_k)\}_k$ of $\tb{CQO}_{\q,\Gamma;\q,\Gamma}$. Let $(\mathcal{C},F)=\bigvee_k (\mathcal{C}_k,F_k)$ and $(\mathcal{E},H)= \overline{\tb{meas}}_{\q}[(\mathcal{D},G),(\mathcal{C},F)]$. Then by Lemma~\ref{lem:expr_vac_ext_meas}, $\tilde{\mathcal{E}}^H = \tilde{\mathcal{D}}^G\circ \mathcal{T}_0^\q +  \tilde{\mathcal{C}}^F\circ\mathcal{T}_1^\q$, with $\mathcal{T}_0^\q,\mathcal{T}_1^\q$ defined as in the lemma. Similarly to the composition case, by Lemma~\ref{lem:topological_limit_channel}, $\tilde{\mathcal{C}}^F$ is the limit of some sequence in $\{\tilde{\mathcal{C}}_k^{F_k}\}_k$ in the Euclidian topology. Then, by continuity of $(\tilde{\mathcal{D}}^G\circ \mathcal{T}_0^\q) +  (- \circ\mathcal{T}_1^\q)$, $\tilde{\mathcal{E}}^H$ is the limit of some sequence in $\{\tilde{\mathcal{D}}^{G}\circ \mathcal{T}_0^\q +  \tilde{\mathcal{C}}_k^{F_k}\circ\mathcal{T}_1^\q\}_k$. Therefore, $\tilde{\mathcal{E}}^H \leq \bigvee_k \tilde{\mathcal{D}}^{F}\circ \mathcal{T}_0^\q +  \tilde{\mathcal{C}}_k^{F_k}\circ\mathcal{T}_1^\q$. Again, the reverse inequality is obtained using the fact that $(\tilde{\mathcal{D}}^G\circ \mathcal{T}_0^\q) +  (- \circ\mathcal{T}_1^\q)$ is an increasing function. Therefore, $\overline{\tb{meas}}_{\q}[(\mathcal{D},G),\bigvee_k(\mathcal{C}_k,F_k)]=\bigvee_k \overline{\tb{meas}}_{\q}[(\mathcal{D},G),(\mathcal{C}_k,F_k)]$.
\end{proof}

\begin{rem}
The above lemma on gives the results necessary to prove the following proposition, but one can also show Scott-continuity of composition on the right and measurement on the left.
\end{rem}

Using this result, we can now prove that the denotational semantics is well-defined.
\begin{prop}
\label{prop:denotational_semantics}
 For all well-formed program $(\Gamma;S)$, the interpretation $\sem{S}_{\Gamma}\in\tb{CQO}_{\Gamma;\Gamma^S}$ is well defined.
\end{prop}
\begin{proof}
We prove the result by structural induction. For each statement $S$, we show that for all pair of environments $\Gamma,\Delta$, if $\Gamma\vdash S\triangleright \Delta$ is derivable, then $\sem{S}_{\Gamma}$ is well defined and in $\tb{CQO}_{\Gamma;\Delta}$.
This is clear for the $\tb{new qbit}$, $\tb{discard}$, unitary, and $\tb{skip}$ statements. For the sequential composition, $\tb{meas}$, and $\tb{qcase}$ statements, we use the fact that composition, $\overline{\tb{meas}}_\q$, and $\overline{\tb{qcase}}_\q$ preserve coherent quantum operations.

For the $\tb{while}$ statement, suppose that $\q,\Gamma\vdash \tb{while}\ \q \tb{ do }S\triangleright \q,\Gamma$ is derivable, for some environment $\Gamma$. The final rule of the derivation is necessarily the while rule:
\[
 \begin{prooftree}
 \hypo{\q,\Gamma\vdash S \triangleright \q,\Gamma}
  \infer1[
  ]{\q,\Gamma\vdash \tb{while }\q \tb{ do }S\triangleright \q,\Gamma}
 \end{prooftree}\raisebox{0.4ex}{ \!.}
\]
By induction hypothesis, $\sem{S}_{\q,\Gamma}$ is well defined in $\tb{CQO}_{\q,\Gamma;\q,\Gamma}$. We show that the least fixed point of $\mathscr{F}_{\q}^{S}$ is well defined. By Proposition~\ref{prop:cqoDCPO}, $(\tb{CQO}_{\q,\Gamma;\q,\Gamma},\sqsubseteq)$ is a pointed DCPO. Moreover, $\mathscr{F}_{\q}^{S}$ can be shown to be a Scott-continuous function from $\tb{CQO}_{\q,\Gamma;\q,\Gamma}$ to itself. First, $\mathscr{F}_{\q}^{S}$ preserves $\tb{CQO}_{\q,\Gamma;\q,\Gamma}$ because the composition and $\overline{\tb{meas}}_\q$ operations preserce $\tb{QCO}_{\q,\Gamma;\q,\Gamma}$.
Consider a directed subset $\{(\mathcal{C}_k,F_k)\}_k$ of $\tb{CQO}_{\q,\Gamma;\q,\Gamma}$. By Lemma~\ref{lem:scott_cont}, we have:
\begin{align*}
\small
\mathscr{F}_{\q}^{S}\left(\bigvee_k (\mathcal{C}_k,F_k)\right) &\small =\ \overline{\tb{meas}}_{\q}\bigl[(\mathcal{I}_{\q,\Gamma},I_{\q,\Gamma}), \left(\bigvee_k (\mathcal{C}_k,F_k)\right)\circ \sem{S}_{\q,\Gamma}\bigr] \\
 &\small =\ \overline{\tb{meas}}_{\q}\bigl[(\mathcal{I}_{\q,\Gamma},I_{\q,\Gamma}), \bigvee_k \bigl((\mathcal{C}_k,F_k)\circ \sem{S}_{\q,\Gamma}\bigr)\bigr]  \quad \text{ \small (by Scott-continuity of $\circ$)} \\
 &\small =\ \bigvee_k \overline{\tb{meas}}_{\q}\bigl[(\mathcal{I}_{\q,\Gamma},I_{\q,\Gamma}), (\mathcal{C}_k,F_k)\circ \sem{S}_{\q,\Gamma}\bigr]  \quad \text{(\small by Scott-continuity of }{\small \overline{\tb{meas}}_\q}) \\
 &\small =\ \bigvee_k \mathscr{F}_{\q}^{S}(\mathcal{C}_k,F_k),
\end{align*}
which shows that $\mathscr{F}_{\q}^{S}$ is Scott-continuous. Therefore its least fixed point $\text{lfp}(\mathscr{F}_{\q}^{S})$ is well defined, in $\tb{CQO}_{\q,\Gamma;\q,\Gamma}$, and equal to the supremum of the directed set $\{(\mathscr{F}_{\q}^{S})^n(0,0)\}_{n\in \mathbb{N}}$.
\end{proof}

Importantly, the denotational semantics of programs is compositional. This means that for each constructor, the interpretation of a program is a function of the interpretations of its sub-programs. For example $\sem{\tb{meas }\q \ (0\rightarrow S_0,1\rightarrow S_1)}_{\q,\Gamma}$ is a function of $\sem{S_0}_{\q,\Gamma}$ and $\sem{S_1}_{\q,\Gamma}$. This property is key to proving full abstraction in Section~\ref{section:observational_equivalence}.

\subsection{Examples}

We revisit some of the examples from Section~\ref{section:ex_op}, and determine their denotational semantics.

\begin{exa}[CNOT and SWAP]
 \label{ex:cnot_den}
 Consider the statement $S_{\text{CNOT}}^{\qc,\qt}\ = \ \tb{qcase}\ \qc \ (0 \rightarrow \tb{skip},$ $1 \rightarrow \qt\ {*}{=} X)$ from Example~\ref{ex:cnot_op}. We have for all environment $\Gamma \not\ni \qc,\qt$:
\begin{align*}
 \sem{S_{\text{CNOT}}^{\qc,\qt}}_{\qc,\qt,\Gamma} =& \ \overline{\tb{qcase}}_{\qc}[\sem{\tb{skip}}_{\qt,\Gamma},\sem{\qt\ {*}{=}X}_{\qt,\Gamma}] \\
 =&\ \overline{\tb{qcase}}_{\qc}[(\mathcal{I}_{\qt,\Gamma},I_{\qt,\Gamma}),(X_{\qt}(\cdot)X_{\qt},X_{\qt})] \\
 =&\ \bigl(\mathcal{P}_0^\qc + \mathcal{P}_1^\qc \otimes X_{\qt}(\cdot)X_{\qt} + \ket{0}\bra{0}_\qc (\cdot)(\ket{1}\bra{1}_\qc\otimes X_\qt^{\dag}) \\
 & +(\ket{1}\bra{1}_\qc \otimes X_{\qt})(\cdot)\ket{0}\bra{0}_\qc, \ket{0}\bra{0}_\qc + \ket{1}\bra{1}_\qc\otimes X_{\qt}\bigr)\\
 =&\ (\text{CNOT}_{\qc,\qt}(\cdot)\text{CNOT}_{\qc,\qt}^{\dag},\text{CNOT}_{\qc,\qt}).
\end{align*}
It follows for the statement $S_{\text{SWAP}}^{\p,\q}$ of Example~\ref{ex:cnot_op} that for $\Gamma \not\ni \p,\q$, $ \sem{S_{\text{SWAP}}^{\p,\q}}_{\p,\q,\Gamma} = (\text{SWAP}_{\p,\q}(\cdot)\text{SWAP}_{\p,\q}^{\dag},\text{SWAP}_{\p,\q})$.
Therefore, the interpretations of statements $S_{\text{CNOT}}^{\qc,\qt}$ and $S_{\text{SWAP}}^{\p,\q}$ are given by coherent quantum operations for the usual CNOT and SWAP channels.
\end{exa}

\begin{exa}[Repeated coin flip]
\label{ex:coin_den}
We consider the statement $\text{COIN}\ =\  \tb{while} \ \q \tb{ do } \q{*}{=} H$ from Example~\ref{ex:coin_op}.
We have $\sem{\text{COIN}}_{\q}=\text{lfp}(\mathscr{F}_{\q}^{\q{*}{=}H})$, which is given by the supremum of iterated applications of $\mathscr{F}_{\q}^{\q{*}{=}H}$ to $(0,0)$. For all $(\mathcal{C},F)\in\tb{CQO}_{\q;\q}$, we have:
\[
 \mathscr{F}_{\q}^{\q {*}{=}H}(\mathcal{C},F)=\overline{\tb{meas}}_{\q}\bigl[(\mathcal{I}_{\q},I_{\q}), (\mathcal{C},F)\circ (H(\cdot)H^{\dag},H)\bigr].
\]
We determine the first few iterations, where  $\ket{\pm}\triangleq\sfrac{1}{\sqrt{2}}(\ket{0}\pm\ket{1})$:
\begin{align*}
 \mathscr{F}_{\q}^{\q {*}{=}H}(0,0) &\ =\ \overline{\tb{meas}}_{\q}\bigl[(\mathcal{I}_{\q},I_{\q}), (0,0) \bigr] = \left(\mathcal{P}_{0}^{\q} ,\ket{0}\bra{0} \right)
 \displaybreak[0]\!,
 \\
 \left(\mathscr{F}_{\q}^{\q {*}{=}H}\right)^2(0,0) &\ =\ \mathscr{F}_{\q}^{\q {*}{=}H}\left(\mathcal{P}_{0}^{\q} ,\ket{0}\bra{0}\right) \\
 &\ =\ \overline{\tb{meas}}_{\q}\left[(\mathcal{I}_{\q},I_{\q}), \left(\mathcal{P}_{0}^{\q} ,\ket{0}\bra{0}\right) \circ (H(\cdot)H^{\dag},H) \right] \\
 &\ =\ \overline{\tb{meas}}_{\q}\left[(\mathcal{I}_{\q},I_{\q}), \left(\ket{0}\bra{+} (\cdot)\ket{+}\bra{0} ,\ket{0}\bra{+} \right) \right] \\
 &\ =\ \left(\mathcal{P}_{0}^{\q} + \frac{1}{2} \ket{0}\bra{1} (\cdot)\ket{1}\bra{0},\ket{0}\bra{0}\right)
 \displaybreak[0]\!, \\
 \left(\mathscr{F}_{\q}^{\q {*}{=}H}\right)^3(0,0) &= \ \ldots
\ =\  \left(\mathcal{P}_{0}^{\q} + \frac{1}{2}\ket{0}\bra{1}(\cdot)\ket{1}\bra{0} + \frac{1}{4} \ket{0}\bra{1} (\cdot)\ket{1}\bra{0},\ket{0}\bra{0}\right)\!.
\end{align*}
The limit is:
\begin{align*}
 \sem{\text{COIN}}_{\q}\ &\ = \left(\mathcal{P}_{0}^{\q} + \left(\sum_{k\geq 1}\frac{1}{2^k}\right)\ket{0}\bra{1}(\cdot)\ket{1}\bra{0}, \ket{0}\bra{0}\right) \\
 &\ =\ \left(Tr(\cdot)\ket{0}\bra{0},  \ket{0}\bra{0}\right)\!.
\end{align*}
\end{exa}

\begin{exa}[Infinite loop]
 \label{ex:loop_den}
 Consider the non-terminating statement from Example~\ref{ex:loop_op}:
 \begin{align*}
 \text{LOOP}\ \triangleq \  & \tb{new qbit }\qr\ ;\ \qr{*}{=}X\ ;\ \tb{while} \ \qr \tb{ do skip}\ ;\ \tb{discard} \ \qr.
 \end{align*}

For all environment $\Gamma$ such that $\qr \notin \Gamma$ we have:
\begin{align}
 \label{eq:loop}
\sem{\text{LOOP}}_{\Gamma} & = \ \sem{\tb{discard }\qr}_{\qr,\Gamma} \circ \sem{\tb{while }\qr\tb{ do skip}}_{\qr,\Gamma} \circ \sem{r{*}{=}X}_{\qr,\Gamma}\circ \sem{\tb{new qbit }\qr}_{\Gamma}.
\end{align}
The fixpoint $\sem{\tb{while }\qr\tb{ do skip}}_{\qr,\Gamma} = \text{lfp}(\mathscr{F}_{\qr}^{\tb{skip}})$ is given by the supremum of iterated applications of $\mathscr{F}_{\qr}^{\tb{skip}}$ to $(0,0)$. For all $(\mathcal{C},F)\in\tb{CQO}_{\qr,\Gamma;\qr,\Gamma}$ we have:
\[
 \mathscr{F}_{\qr}^{\tb{skip}}(\mathcal{C},F)\ =\ \overline{\tb{meas}}_{\qr}\bigl[(\mathcal{I}_{\qr,\Gamma},I_{\qr,\Gamma}), (\mathcal{C},F)\bigr]\!.
\]
Therefore:
\begin{align*}
 \mathscr{F}_{\qr}^{\tb{skip}}(0,0)&\ =\  \overline{\tb{meas}}_{\qr}\bigl[(\mathcal{I}_{\qr,\Gamma},I_{\qr,\Gamma}), (0,0) \bigr] = \left(\mathcal{P}_{0}^{\qr} ,\ket{0}\bra{0}_\qr \right)\!,
 \\
 (\mathscr{F}_{\qr}^{\tb{skip}})^2(0,0) &\ =\  \mathscr{F}_{\qr}^{\tb{skip}}\left(\mathcal{P}_{0}^{\qr} ,\ket{0}\bra{0}_\qr\right)  \\
 &\ =\  \overline{\tb{meas}}_{\qr}\left[(\mathcal{I}_{\qr,\Gamma},I_{\qr,\Gamma}), \left(\mathcal{P}_{0}^{\qr} ,\ket{0}\bra{0}_\qr \right) \right]  \\
 &\ =\  \left(\mathcal{P}_0^\qr , \ket{0}\bra{0}_\qr \right)\\
&\ =\  \mathscr{F}_{\qr}^{\tb{skip}}(0,0).
\end{align*}
Consequently, for all $k\geq 1$, $(\mathscr{F}_{\qr}^{\tb{skip}})^k(0,0)=\bigl(\mathcal{P}_0^\qr , \ket{0}\bra{0}_\qr \bigr)$, giving us the least fixed point \[\sem{\tb{while }\qr\tb{ do skip}}_{\qr,\Gamma} =\left(\mathcal{P}_0^\qr , \ket{0}\bra{0}_{\qr} \right)\!.\]
Combined with Equation~(\ref{eq:loop}), we obtain:
\begin{align*}
\sem{\text{LOOP}}_{\Gamma} \ =\ & (Tr_\qr,\bra{0}_\qr\otimes I_{\Gamma}) \circ \left(\mathcal{P}_0^\qr , \ket{0}\bra{0}_\qr \right) \circ (X_{\qr}(\cdot)X_\qr^{\dag},X_\qr) \circ (\ket{0}\bra{0}_\qr\otimes \mathcal{I}_{\Gamma},\ket{0}_\qr\otimes I_{\Gamma}) \\
\ =\ &(0,0).
\end{align*}

\end{exa}

\begin{exa}[Variable renaming]
\label{ex:renaming}
Consider the statement below renaming a variable $\p$ to $\q \neq \p$:
\begin{align*}
 \tb{rename }\p \rightarrow \q \ \triangleq\ & \tb{new qbit }\q \ ;\ S_{\text{SWAP}}^{\p,\q}\ ; \ \tb{discard }\p,
\end{align*}
where the statement $S_{\text{SWAP}}^{\p,\q}$ is defined in Example~\ref{ex:cnot_op}. For all $\Gamma$ containing neither $\p$ nor $\q$, we can derive $\p,\Gamma\vdash \tb{rename }\p \rightarrow \q\ \triangleright\q,\Gamma$. Then, using example~\ref{ex:cnot_den}:
\begin{align*}
 \sem{\tb{rename }\p \rightarrow \q}_\p \ =\ & \sem{ \tb{discard }\p}_{\p,\q} \circ \sem{S^{\p,\q}_{\text{SWAP}}}_{\p,\q} \circ \sem{\tb{new qbit }\q}_\p \\
\ =\ &(Tr_{\p},\bra{0}_\p\otimes I_\q) \circ (\text{SWAP}_{\p,\q}(\cdot)\text{SWAP}_{\p,\q}^{\dag},\text{SWAP}_{\p,\q}) \\
 & \circ (\mathcal{I}_\p\otimes \ket{0}\bra{0}_\q,I_\p\otimes \ket{0}_\q) \\
 \ =\ &(\mathcal{I}_{\p\rightarrow \q},I_{\p\rightarrow \q}),
\end{align*}
where $I_{\p\rightarrow \q}=\ket{0}_\q\bra{0}_\p + \ket{1}_\q\bra{1}_\p$ and $\mathcal{I}_{\p\rightarrow \q}=I_{\p\rightarrow \q}(\cdot)I_{\q\rightarrow \p}$.
\end{exa}

\section{Universality}
\label{section:universality_all}

\subsection{General Case}

The goal of this section is to show that the language is universal, meaning that every coherent quantum operation $(\mathcal{C},F)$ can be realized as the interpretation of a program. This is stated as the following theorem.
\begin{thm}[Universality]
 \label{thm:universality}
 Let $\Gamma$ and $\Delta$ be environments and $(\mathcal{C},F)\in\tb{CQO}_{\Gamma;\Delta}$ be a coherent quantum operation. Then, there exists a program $(\Gamma;S)$ such that $\Gamma^S=\Delta$ and $\sem{S}_{\Gamma}=(\mathcal{C},F)$.
\end{thm}

To prove this result, it is useful to introduce the notion of \emph{sub-unitary operations}.
\begin{defi}[Sub-unitary matrices and operations]
A matrix $U$ is said to be \emph{sub-unitary} if there exist matrices $U_1,U_2,U_3$ such that
\[
  \begin{bmatrix}
U & U_1\\
U_2 & U_3
\end{bmatrix}
\]
is unitary. A quantum operation $\mathcal{C}$ is said to be sub-unitary if there exists a sub-unitary matrix $U$ such that $\mathcal{C} = U(\cdot)U^{\dag}$.
\end{defi}

We then show that every vacuum-extended operation can be obtained by composing a vacuum extension of a sub-unitary operation with a vacuum extension of a traceout. This result is analogous to~\cite[Theorem 6.12]{QPL}.
\begin{lem}
 \label{lem:factorize_vac_extension}
 Let $\mathcal{D}\in \tb{QO}_{\tb{vac}}(\Hilb_A, \Hilb_B)$. Then $\mathcal{D}$ can be factorized as $\mathcal{D}=\widetilde{Tr_C}\circ \tilde{\mathcal{G}}$ for some system $C$, where:
 \begin{itemize}
  \item $\tilde{\mathcal{G}}\in\tb{QO}_\tb{vac}(\Hilb_A, \Hilb_C\otimes \Hilb_B)$ is a vacuum extension of a sub-unitary operation $\mathcal{G}\in\tb{QO}(\Hilb_A,$ $ \Hilb_C\otimes \Hilb_B)$, and
  \item $\widetilde{Tr_C}\in\tb{QO}_{\tb{vac}}(\Hilb_C\otimes \Hilb_B, \Hilb_B)$ is a vacuum extension of the traceout channel $Tr_C\in\tb{QC}(\Hilb_C\otimes \Hilb_B, \Hilb_B)$.
\end{itemize}
\end{lem}
\begin{proof}
Let $\{K_i \oplus \nu_i \ket{\text{vac}}\bra{\text{vac}}\}_{1\leq i\leq n}$ be a finite Kraus decomposition of $\mathcal{D}$, with $K_i\in\mathcal{L}(\Hilb_A, \Hilb_B)$ for all $i$ (Kraus decompositions of vacuum extensions necessarily have this form~\cite{kristjansson}).
Define the matrix $U$ as follows:
\begin{align}
\label{eq:matrix_U}
 U\triangleq\begin{bmatrix}
     K_1 \\
     \vdots \\
     K_n
    \end{bmatrix} \in\mathcal{L}(\Hilb_A, \Hilb_C\otimes \Hilb_B),
\end{align}
with $C$ an $n$-dimensional system. \cite[Lemma 6.11]{QPL} states that a matrix $V\in \mathbb{C}^{q\times p}$ is sub-unitary if and only if $VV^\dag \leq I_q$, if and only if $V^\dag V\leq I_p$. We have $U^\dag U = \sum_{i=1}^n K_i^{\dag}K_i \leq I_{d_A}$ ($d_A\triangleq\dim \Hilb_A$)
because $\{K_i\}_{1\leq i\leq n}$ also forms a Kraus decomposition. Therefore $U$ is sub-unitary. Let $\mathcal{G}\triangleq U (\cdot) U^{\dag}\in \tb{QO}(\Hilb_A, \Hilb_C\otimes \Hilb_B)$, and define $\tilde{\mathcal{G}} \triangleq (U \oplus\ket{\text{vac}} \bra{\text{vac}})(\cdot)(U^{\dag} \oplus\ket{\text{vac}} \bra{\text{vac}})$, which is clearly a vacuum extension of $\mathcal{G}$. Then, define $\widetilde{Tr_C}$ as the quantum operation with Kraus operators $\{(\bra{e_i}\otimes I_{B}) \oplus \nu_i\ket{\text{vac}}\bra{\text{vac}}\}_{1\leq i\leq n}$, where $(\ket{e_i})_{1\leq i\leq n}$ is the orthonormal basis of $\Hilb_C$ with respect to which $U$ is written in Equation~\ref{eq:matrix_U}, and $I_B$ is the identity on $\Hilb_B$. Then $\widetilde{Tr_C}$ is a vacuum extension of $Tr_C\in\tb{QC}(\Hilb_C\otimes \Hilb_B,\Hilb_B)$. Since for all $i$, $(\bra{e_i}\otimes I_B)U = K_i$, we have $\mathcal{D}=\widetilde{Tr_C}\circ \tilde{\mathcal{G}}$.
\end{proof}

Lemma~\ref{lem:factorize_vac_extension} can now be adapted from vacuum extension to coherent quantum operations (Section~\ref{section:vacuum}).
Given a coherent quantum operation $(\mathcal{C},F)\in\tb{CQO}_{\Gamma;\Delta}$, there exist a system $\Delta'$, a sub-unitary matrix $U\in\mathcal{L}(\Hilb_\Gamma,\Hilb_{\Delta',\Delta})$, and a normalized state $\ket{\psi}\in St(\Hilb_{\Delta'})$ such that:
\begin{equation}
\label{eq:decomp_vac_ext}
 (\mathcal{C},F)=(Tr_{\Delta'}, \bra{\psi}\otimes I_{\Delta}) \circ (U(\cdot) U^{\dag},U),
\end{equation}
where $Tr_{\Delta'}$ is defined as the traceout of each qubit of $\Delta'$.
Such a decomposition exists because the Hilbert space $\Hilb_C$ of Lemma~\ref{lem:factorize_vac_extension} can always be chosen with its dimension a power of 2.

To prove Theorem~\ref{thm:universality}, we show in turn that all unitary and sub-unitary operations can be realized as interpretations of programs (Lemmas~\ref{lem:unitary_prgm} and~\ref{lem:sub_unitary_prgm} respectively).
Since any coherent quantum operation can be decomposed into a sub-unitary operation followed by a traceout function, the proof of the main result follows.
\begin{lem}
 \label{lem:unitary_prgm}
 Let $\Gamma$ be an environment and $U\in\mathcal{L}(\Hilb_\Gamma)$ be a unitary map. Then there exists a program $(\Gamma;S_U)$ such that $\Gamma^{S_U}=\Gamma$ and $\sem{S_U}_{\Gamma}=(U(\cdot)U^{\dag},U)$.
\end{lem}
\begin{proof}
The unitary matrix $U$ (acting on an arbitrary number of qubits) can be decomposed into a product of CNOT and single-qubit gates~\cite{barenco}. The CNOT gates can be implemented with the $\tb{qcase}$ statement as was done in Example~\ref{ex:cnot_den}, and single-qubit unitaries are built into the language. Therefore a suitable statement $S_U$ can be constructed.
\end{proof}

We extend this result to sub-unitary maps as follows.\footnote{This lemma is analogous to~\cite[Lemma 6.13 (b)]{QPL}.}
\begin{lem}
 \label{lem:sub_unitary_prgm}
 Let $\Gamma$ and $\Delta$ be environments and $U\in\mathcal{L}(\Hilb_\Gamma,\Hilb_\Delta)$ be a sub-unitary matrix. Then there exists a program $(\Gamma;S_U)$ such that $\Gamma^{S_U}=\Delta$ and $\sem{S_U}_{\Gamma}=(U(\cdot)U^{\dag},U)$.
\end{lem}
\begin{proof}
Unlike unitary operations, sub-unitary operations may be trace decreasing. The proof captures this feature by combining measurement and the non-terminating LOOP program from Example~\ref{ex:loop_op}.

First, we extend the $\tb{rename}$ statement from Example~\ref{ex:renaming} to lists of variable names $\p_1,\ldots,\p_m$ and $\q_1,\ldots,\q_m$ as follows, where variable names are allowed to appear in both lists:
\begin{align*}
 \tb{rename }\p_1,\ldots,\p_m \rightarrow \q_1,\ldots,\q_m \ \triangleq\ & \tb{rename }\p_1 \rightarrow \qr_1 \ ;\ \cdots\ ;\ \tb{rename }\p_m \rightarrow \qr_m ; \\
 & \tb{rename }\qr_1 \rightarrow \q_1\ ;\ \cdots\ ;\  \tb{rename }\qr_m \rightarrow \q_m,
\end{align*}
with $\qr_1,\ldots,\qr_m$ fresh variables appearing neither in $\p_1,\ldots,\p_m$ nor in $\q_1,\ldots,\q_m$.
Then for all environment $\Sigma$ containing none of the variables $\p_1,\ldots,\p_m,\q_1,\ldots,\q_m,\qr_1,\ldots,\qr_m$, the judgment $\p_1,\ldots,\p_m,\Sigma \vdash \tb{rename} \ \p_1,\ldots,\p_m \rightarrow \q_1,\ldots,\q_m\triangleright \q_1,\ldots,\q_m,\Sigma$ is derivable.

Let $\Gamma = \p_1,\ldots,\p_k$ and $\Delta = \q_1,\ldots,\q_l$. There exist variables $\p_{k+1},\ldots,\p_n$ ($n\geq k,l$) and matrices $U_1,U_2,U_3$ such that
\[
 U'= \begin{bmatrix}
U & U_1\\
U_2 & U_3
\end{bmatrix} \in\mathcal{L}(\Hilb_{\bar{\Gamma}})
\]
is unitary, where $\bar{\Gamma}\triangleq\p_1,\ldots,\p_n$ is an extension of $\Gamma$.\footnote{Without loss of generality, $U'$ can be chosen in $\mathbb{C}^{2^n \times 2^n}$ for some integer $n$.} Then there exists a statement $S_{U'}$ that consists in applying the unitary map $U'$ to variables $\p_1,\ldots,\p_n$, as specified by Lemma~\ref{lem:unitary_prgm}.

We define the statement $S_U$, with notations generalizing the $\tb{new qbit}$, $\tb{meas}$, and $\tb{discard}$ statements to multiple variables, as follows:
\begin{align*}
 S_U\ \triangleq\ &\tb{new qbit }\p_{k+1},\ldots,\p_n ; \  S_{U'}; \\
 &\tb{meas}\ \p_{l+1},\ldots,\p_n \ (0^{n-l} \rightarrow \tb{skip}, \ \tb{else} \rightarrow \text{LOOP} ); \\
 & \tb{discard }\p_{l+1},\ldots,\p_n; \\
 & \tb{rename}\ \p_1,\ldots,\p_l \rightarrow \q_1,\ldots,\q_l,
\end{align*}
where LOOP is the non-terminating statement from Example~\ref{ex:loop_op}. Then, the judgment $\Gamma\vdash S_U\triangleright \Delta$ is derivable, and
\begin{align*}
 \sem{S_U}_{\Gamma} =&\  \sem{\tb{rename}\ \p_1,\ldots,\p_l \rightarrow \q_1,\ldots,\q_l}_{\p_1,\ldots,\p_l} \circ \sem{\tb{discard }\p_{l+1},\ldots,\p_n}_{\bar{\Gamma}} \\
 & \circ \sem{ \tb{meas}\ \p_{l+1},\ldots,\p_n \ ( 0^{n-l} \rightarrow \tb{skip}, \tb{else} \rightarrow \text{LOOP})}_{\bar{\Gamma}} \\
 & \circ \sem{S_{U'}}_{\bar{\Gamma}} \circ \sem{\tb{new qbit }\p_{k+1},\ldots,\p_n}_{\Gamma}.
\end{align*}
We have
\begin{align*}
  &\sem{ \tb{meas}\ \p_{l+1},\ldots,\p_n \ ( 0^{n-l} \rightarrow \tb{skip}, \tb{else} \rightarrow \text{LOOP})}_{\bar{\Gamma}} \\
 &\quad =\ \overline{\tb{meas}}_{\p_{l+1},\ldots,\p_n} \bigl[\sem{\tb{skip}}_{\bar{\Gamma}},\sem{\text{LOOP}}_{\bar{\Gamma}}\bigr] \\
 &\quad=\ \overline{\tb{meas}}_{\p_{l+1},\ldots,\p_n} \bigl[(\mathcal{I}_{\bar{\Gamma}},I_{\bar{\Gamma}}), (0,0)\bigr] \\
 &\quad =\ \Bigl((\ket{0}\bra{0})_{\p_{l+1},\ldots,\p_n}^{\otimes (n-l)} (\cdot)(\ket{0}\bra{0})_{\p_{l+1},\ldots,\p_n}^{\otimes (n-l)} , \ I_{\p_1,\ldots,\p_l} \otimes (\ket{0}\bra{0})_{\p_{l+1},\ldots,\p_n}^{\otimes (n-l)}\Bigr),
\end{align*}
where $\overline{\tb{meas}}_{\p_{l+1},\ldots,\p_n}$ generalizes $\overline{\tb{meas}}$ to several variables, the first input corresponding to all qubits in the $\ket{0}$ state. Therefore:
\begin{align*}
 & \sem{S_U}_{\Gamma}
  =(\mathcal{I}_{\p_1,\ldots,\p_l \rightarrow \q_1,\ldots,\q_l},I_{\p_1,\ldots,\p_l \rightarrow \q_1,\ldots,\q_l}) \circ \left(Tr_{\p_{l+1},\ldots,\p_n},I_{\p_1,\ldots,\p_l}\otimes \bra{0}_{\p_{l+1},\ldots,\p_n}^{\otimes(n-l)}\right)\\
 &\qquad \qquad \circ \Bigl((\ket{0}\bra{0})_{\p_{l+1},\ldots,\p_n}^{\otimes (n-l)} (\cdot)(\ket{0}\bra{0})_{\p_{l+1},\ldots,\p_n}^{\otimes (n-l)} , \ I_{\p_1,\ldots,\p_l} \otimes (\ket{0}\bra{0})_{\p_{l+1},\ldots,\p_n}^{\otimes (n-l)}\Bigr) \\
 & \qquad \qquad \circ \left(U'(\cdot)U'^{\dag},U'\right) \circ \left(\mathcal{I}_{\Gamma}\otimes (\ket{0}\bra{0})^{\otimes (n-k)}_{\p_{k+1},\ldots,\p_n},I_{\Gamma}\otimes \ket{0}^{\otimes (n-k)}_{\p_{k+1},\ldots,\p_n}\right)\\
 & \quad = \ (\mathcal{I}_{\p_1,\ldots,\p_l \rightarrow \q_1,\ldots,\q_l},I_{\p_1,\ldots,\p_l \rightarrow \q_1,\ldots,\q_l}) \\
  & \qquad \ \circ
 \Bigl(\bra{0}_{\p_{l+1},\ldots,\p_n}^{\otimes (n-l)}U'\ket{0}^{\otimes (n-k)}_{\p_{k+1},\ldots,\p_n} (\cdot) \bra{0}_{\p_{k+1},\ldots,\p_n}^{\otimes (n-k)}U'^{\dag}\ket{0}^{\otimes (n-k)}_{\p_{l+1},\ldots,\p_n}, \ \bra{0}_{\p_{l+1},\ldots,\p_n}^{\otimes (n-l)}U'\ket{0}^{\otimes (n-k)}_{\p_{k+1},\ldots,\p_n} \Bigr) \\
 & \quad =\ (U(\cdot)U^{\dag},U). \qedhere
\end{align*}
\end{proof}

We can now prove Theorem~\ref{thm:universality}.
\begin{proof}
 By Lemma~\ref{lem:factorize_vac_extension}, there exist an environment $\Delta'$, a sub-unitary matrix $U\in\mathcal{L}(\Hilb_\Gamma,\Hilb_{\Delta',\Delta})$, and a normalized state $\ket{\psi}\in St(\Hilb_{\Delta'})$ such that $ (\mathcal{C},F)=(Tr_{\Delta'}, \bra{\psi}\otimes I_{\Delta}) \circ (U(\cdot) U^{\dag},U)$. By Lemma~\ref{lem:sub_unitary_prgm}, there exists a program $(\Gamma;S_U)$ such that $\Gamma^{S_U}=\Delta',\Delta$ and $\sem{S_U}_{\Gamma}= (U(\cdot) U^{\dag},U)$. Take $\Delta' = \qr_1,\ldots,\qr_s$, and let $V\in \mathcal{L}(\Hilb_{\Delta'})$ denote a unitary matrix such that $V \ket{0}^{\otimes s} = \ket{\psi}$. By Lemma~\ref{lem:unitary_prgm}, there exists a program $((\Delta',\Delta);S_{V^{\dag}})$ such that $(\Delta',\Delta)^{S_{V^\dag}}= (\Delta',\Delta)$ and $\sem{S_{V^{\dag}}}_{\Gamma}=(V^{\dag}(\cdot)V\otimes \mathcal{I}_{\Delta},V^{\dag}\otimes I_{\Delta})$.
Then define:
\begin{align*}
 S\ \triangleq\ S_U ;\  S_{V^{\dag}}
 ;\ \tb{discard }\qr_1, \ldots, \qr_s\ \!.
\end{align*}
The judgment $\Gamma\vdash S\triangleright \Delta$ is derivable, and we have:
\begin{align*}
 \sem{S}_{\Gamma} &= \sem{\tb{discard}\ \qr_1,\ldots,\qr_s}_{\Delta',\Delta} \circ \sem{S_{V^{\dag}}}_{\Delta',\Delta}
 \circ \sem{S_U}_{\Gamma} \\
 &= (Tr_{\Delta'},\bra{0}^{\otimes s}_{\Delta'}\otimes I_{\Delta}) \circ (V^{\dag}(\cdot)V\otimes \mathcal{I}_{\Delta},V^{\dag}\otimes I_{\Delta})
 \circ (U (\cdot)U^{\dag},U) \\
 &= (Tr_{\Delta'},\bra{\psi}_{\Delta'}\otimes I_{\Delta}) \circ (U (\cdot)U^{\dag},U)\\
 &=(\mathcal{C},F).\qedhere
\end{align*}
\end{proof}

A key feature of our quantum programming language is that primitive unitary operations are restricted to local gates, with non-local operations like the CNOT constructed using quantum control mechanisms. The presence of quantum control further enables restricting the unitary primitives to just the Hadamard gate and global phases, defined as follows on computational basis states:
\[
H\triangleq\ \ket{x} \mapsto \sfrac{1}{\sqrt{2}}(\ket{0}+(-1)^x \ket{1}) \qquad \text{ and } \qquad
P_\theta\triangleq\ \ket{x} \mapsto e^{i\theta}\ket{x}.
\]
Here $P_\theta$ is presented as a one-qubit operation, but its action actually does not depend on the qubit on which $P_\theta$ is applied.
Quantum control enables construction of Z-rotations using global phases, thereby achieving full universality for the language:

\begin{prop}
 \label{prop:universality}
 For any coherent quantum operation $(\mathcal{C},F)\in\tb{CQO}_{\Gamma;\Delta}$, there exists a program $(\Gamma;S)$ such that $\Gamma^S=\Delta$ and $\sem{S}_{\Gamma}=(\mathcal{C},F)$, where $S$ uses only $H$ and global phases.
\end{prop}
To prove the above proposition, we first show the existence of programs implementing all $Z$-rotations using global phase gates and $\tb{qcase}$.
\begin{lem}
\label{lem:z-rotation}
 Let $\theta\in[0,2\pi[$ and $Z_\theta$ be the corresponding $Z$-rotation defined on computational basis states as $Z_\theta:\ket{x}\mapsto e^{i\theta x}\ket{x}$. Let $\q,\Gamma$ be an environment. Then there exists a program $(\q,\Gamma;S_\theta)$ such that $(\q,\Gamma)^{S_\theta}=\q,\Gamma$ and $\sem{S_{\theta}}_{\q,\Gamma}=((Z_\theta)_\q(\cdot) (Z_\theta)_\q^\dag, (Z_\theta)_\q)$, where $S_\theta$ uses only global phases.
\end{lem}
\begin{proof}
Let $\qr\notin \q,\Gamma$. Define
$S_\theta \triangleq \tb{new qbit}\ \qr ; \tb{qcase}\ \q\ (0\rightarrow \tb{skip},1\rightarrow \qr{*}{=}P_\theta) ; \tb{discard}\ \qr$. We have:
\begin{align*}
 \sem{S_\theta}_{\q,\Gamma} =\ &\sem{\tb{discard}\ \qr}_{\qr,\q,\Gamma} \circ \sem{\tb{qcase}\ \q\ (0\rightarrow \tb{skip},1\rightarrow \qr{*}{=}P_\theta)}_{\qr,\q,\Gamma} \circ \sem{\tb{new qbit}\ \qr}_{\q,\Gamma}\\
 =\ &(Tr_\qr,\bra{0}_\qr \otimes I_{\q,\Gamma}) \circ \overline{\tb{qcase}}_\q [\sem{\tb{skip}}_{\qr,\Gamma},\sem{\qr{*}{=}P_\theta}_{\qr,\Gamma}] \circ (\ket{0}\bra{0}_\qr\otimes\mathcal{I}_{\q,\Gamma},\ket{0}_\qr\otimes I_{\q,\Gamma})
 \\
 =\ &(Tr_\qr,\bra{0}_\qr \otimes I_{\q,\Gamma})
 \circ \overline{\tb{qcase}}_\q [(\mathcal{I}_{\qr,\Gamma},I_{\qr,\Gamma}),((P_{\theta})_\qr(\cdot)(P_{\theta})_\qr^\dag,(P_{\theta})_\qr)] \\
 &\circ (\ket{0}\bra{0}_\qr\otimes\mathcal{I}_{\q,\Gamma},\ket{0}_\qr\otimes I_{\q,\Gamma})
 \\
 =\ &\bigl(\mathcal{P}_0^\q \otimes \mathcal{I}_\Gamma + \mathcal{P}_1^\q\otimes\mathcal{I}_\Gamma
 + e^{-i\theta} (\ket{0}\bra{0}_\q\otimes I_\Gamma)(\cdot)(\ket{1}\bra{1}_\q\otimes I_\Gamma)
 \\
 &+ e^{i\theta} (\ket{1}\bra{1}_\q\otimes I_\Gamma)(\cdot)(\ket{0}\bra{0}_\q\otimes I_\Gamma), \ \ket{0}\bra{0}_\q \otimes I_\Gamma + e^{i\theta} \ket{1}\bra{1}\otimes I_\Gamma \bigr) \\
 =\ & ((Z_\theta)_\q(\cdot) (Z_\theta)_\q^\dag, (Z_\theta)_\q). \qedhere
\end{align*}
\end{proof}

It follows that all unitary gates can be implemented using only $H$ and global phases:
\begin{lem}
 Let $\Gamma$ be an environment and $U\in\mathcal{L}(\Hilb_\Gamma)$ be a unitary map. Then there exists a program $(\Gamma;S_U)$ such that $\Gamma^{S_U}=\Gamma$ and $\sem{S_U}_{\Gamma}=(U(\cdot)U^{\dag},U)$, where $S$ uses only $H$ and global phases.
\end{lem}
\begin{proof}
The unitary matrix $U$ can be decomposed into a sequence of CNOT and single-qubit gates~\cite{barenco}. CNOT gates can be implemented using Example~\ref{ex:cnot_den}, replacing $X$ by $HZH$ (with $Z=Z_\pi$). Moreover, every single-qubit unitary $V$ can be written as
\[
 V=P_{\theta_0}Z_{\theta_1}HZ_{\theta_2}HZ_{\theta_3},\]
for some $\theta_0,\theta_1,\theta_2,\theta_3$ (see for instance~\cite{nielsen}). Using Lemma~\ref{lem:z-rotation} to implement $Z$-rotations with $\tb{qcase}$ and global phases, we can construct a suitable statement $S_U$ using only $H$ and global phases.
\end{proof}
Using this result, a proof of Proposition~\ref{prop:universality} can be obtained by proceeding as in the proof of Theorem~\ref{thm:universality}.

Proposition~\ref{prop:universality} extends the result obtained for the unitary fragment in \cite{HLMR25,controlPROP}. In particular, the language~\cite{HLMR25} achieves universality by combining a global phase with a pattern-matching construct on the computational and Hadamard bases.

\subsection{Restriction to a Discrete Set of Qubit Gates}
\label{section:universality_approx}

So far, we have assumed that we have an infinite set of single-qubit unitaries at our disposal. To represent a more realistic setting, we now restrict the language to $HT$-\emph{programs}, programs $(\Gamma;S)$ whose statement $S$ uses only the single-qubit gates $H$ and $T\triangleq\ \ket{x} \mapsto e^{ix\frac{\pi}{4}}\ket{x}$.
The CNOT gate can be reconstructed using the $\tb{qcase}$ primitive (see Example~\ref{ex:cnot_op}), and using $H T^4 H$ for the Pauli $X$ gate. The CNOT, $H$, and $T$ gates are known to be a universal set of gates~\cite{boykin}.
Hence, the goal of this section is to show that any coherent quantum operation can be approximated up to arbitrary accuracy when only $H$ and $T$ gates are built into the language.

The notion of arbitrary accuracy is formalized using an appropriate distance function.
The distance function $d:\mathcal{L}(\mathcal{H}_A,\mathcal{H}_A)\times \mathcal{L}(\mathcal{H}_A,\mathcal{H}_A) \rightarrow \mathbb{R}$ for linear maps is defined as follows~\cite{nielsen}:
\[
 d(U,V)\triangleq\max_{\substack{
 \ket{\psi}\in \Hilb_A: \\ \|\ket{\psi}\!\|= 1}} \|(U-V)\ket{\psi}\|,
\]
with $\|\cdot\|$ the Euclidian norm. In order to define approximations of coherent quantum operations, we similarly define a distance function for $\tb{CQO}(\mathcal{H}_A,\mathcal{H}_B)$. The diamond norm~\cite{watrous} is defined as follows for all Hermitian-preserving linear map $\mathcal{C}\in\mathcal{L}(\mathcal{L}(\Hilb_A),\mathcal{L}(\Hilb_B))$:
\[
 \|\mathcal{C}\|_{\diamond} \triangleq \max_{\substack{ \ket{\psi}\in \Hilb_{A}\otimes \Hilb_A: \\ \|\ket{\psi}\!\|=1}} \|(\mathcal{I}_A\otimes\mathcal{C})(\ket{\psi}\bra{\psi})\|_1,
\]
where the  trace norm $\|\cdot\|_1$, defined by $\|X\|_{1}\triangleq Tr\bigl(\sqrt{X^{\dag}X}\bigr)$, differs from the Euclidian norm  $\|\cdot\|$.
Then, we define the distance $d:\tb{CQO}(\mathcal{H}_A,\mathcal{H}_B)\times \tb{CQO}(\mathcal{H}_A,\mathcal{H}_B) \rightarrow \mathbb{R}$:
\[
 d\bigl((\mathcal{C},F),(\mathcal{D},G)\bigr)\triangleq\bigl\|\tilde{\mathcal{C}}^F - \tilde{\mathcal{D}}^G\bigr\|_{\diamond}.
\]

We will now show that the two useful formulas below hold for all pair of sub-unitary gates $U_1,U_2\in \mathcal{L}(\Hilb_\Gamma)$.
\begin{align}
\label{eq:approx_1}
 &\|U_1(\cdot)U_1^{\dag}-U_2(\cdot)U_2^{\dag}\|_{\diamond}\leq 2d(U_1,U_2)
 \\
 \label{eq:approx_2}
 &d(U_1\oplus \ket{\text{vac}}\bra{\text{vac}},U_2\oplus \ket{\text{vac}}\bra{\text{vac}})=d(U_1,U_2)
\end{align}

To show Inequality~\ref{eq:approx_1}, we have:
\begin{align*}
 &\|U_1(\cdot)U_1^{\dag}-U_2(\cdot)U_2^{\dag}\|_{\diamond}
 \\
 &\quad = \max_{\ket{\psi}} \|(I\otimes U_1)\ket{\psi}\bra{\psi}(I\otimes U^{\dag}_1) -(I\otimes U_2)\ket{\psi}\bra{\psi}(I\otimes U^{\dag}_2) \|_1 \\
 &\quad \leq \max_{\ket{\psi}} \|(I\otimes U_1)\ket{\psi}\bra{\psi}(I\otimes (U^{\dag}_1-U_2^{\dag}))\|_1 +  \|(I\otimes (U_1-U_2))\ket{\psi}\bra{\psi}(I\otimes U_2^{\dag})\|_1 \\
 &\quad \leq 2 d(I\otimes U_1,I\otimes U_2),
\end{align*}
where, in the last line, we used the equality $\|\ket{\psi}\bra{\varphi}\|_1=\|\ket{\psi}\|\|\ket{\varphi}\|$ and the fact that both $I\otimes U_1$ and $I \otimes U_2$ are norm non-increasing.
We have $d(I\otimes U_1,I\otimes U_2)=\max_{\ket{\psi}}\|(I\otimes (U_1-U_2))\ket{\psi}\|$. For all $\ket{\psi}$, by the Schmidt decomposition, there exist orthonormal bases $(\ket{e_i})_i$ and $(\ket{\epsilon_i})_i$ of $\Hilb_{\Gamma}$ such that $\ket{\psi}=\sum_i \lambda_i\ket{e_i}\ket{\epsilon_i}$, where the $\lambda_i$ are non-negative real numbers such that $\sum_i\lambda_i^2=1$. Therefore:
\begin{align*}
 \|(I\otimes (U_1-U_2))\ket{\psi}\|^2 \ &=\ \bra{\psi}(I\otimes (U_1^{\dag}-U_2^{\dag}))(I\otimes (U_1-U_2))\ket{\psi}\\
 &=\ \sum_i \lambda_i^2 \bra{\epsilon_i}(U_1^{\dag}-U_2^{\dag})(U_1-U_2)\ket{\epsilon_i} \\
  &=\ \sum_i \lambda_i^2\|(U_1-U_2)\ket{\epsilon_i}\|^2 \\
 &\leq\ d(U_1,U_2)^2.
\end{align*}
Consequently, $d(I\otimes U_1,I\otimes U_2)\leq d(U_1,U_2)$ holds.\footnote{The reverse inequality is straightforward to prove.} From this, we obtain Inequality~\ref{eq:approx_1}. To show Equation~\ref{eq:approx_2}, we use the fact that $(U_1\oplus \ket{\text{vac}}\bra{\text{vac}})-(U_2\oplus \ket{\text{vac}}\bra{\text{vac}})=U_1-U_2$.

To prove the existence of $HT$-programs approximating all coherent quantum operations $(\mathcal{C},F)$, we will first prove the existence of suitable approximations for unitary and sub-unitary operations,
invoking the universality of the set of CNOT, $H$ and $T$ gates~\cite{boykin}.
Then, to approximate general coherent quantum operations, we will use the decomposition into a sub-unitary operation followed by a traceout (Equation~\ref{eq:decomp_vac_ext}), and show that it is sufficient to compose approximations of each component.

To begin, we show that coherent quantum operations for unitary channels can be approximated with arbitrary precision.
\begin{lem}
 \label{lem:approx_unitary_prgm}
Let $\Gamma$ be an environment and $U\in\mathcal{L}(\Hilb_\Gamma)$ a unitary map. Then for all $\varepsilon>0$, there exists a $HT$-program $(\Gamma;S)$ such that $\Gamma^S=\Gamma$ and $d\bigl((U(\cdot)U^{\dag},U),\sem{S}_{\Gamma}\bigr)\leq \varepsilon$.
\end{lem}
\begin{proof}
Fix $\varepsilon\geq 0$. As shown in~\cite{boykin}, there exists a unitary gate $V\in \mathcal{L}(\Hilb_\Gamma)$ consisting of a sequence of CNOT, $H$, and $T$ gates such that $d(U,V)\leq \varepsilon$. Then only using unitaries $H$ and $T$, we can construct a statement $S$ implementing $V$, that is $\sem{S}_{\Gamma}=(V(\cdot)V^{\dag},V)$. We have
\[
 \widetilde{(U(\cdot)U^{\dag})}^{U}=\bigl(U\oplus \ket{\text{vac}}\bra{\text{vac}}\bigr)(\cdot) \bigl(U^{\dag}\oplus \ket{\text{vac}}\bra{\text{vac}}\bigr)
\]
and similarly for $V$. Combining this with Equations~\ref{eq:approx_1} and~\ref{eq:approx_2}, we obtain:
\[
 d\bigl((U(\cdot)U^{\dag},U),\sem{S}_{\Gamma}\bigr)\leq 2d(U,V)\leq 2\varepsilon. \qedhere
\]
\end{proof}

\begin{lem}
 \label{lem:approx_sub_unitary_prgm}
 Let $\Gamma$ and $\Delta$ be environments and $U\in\mathcal{L}(\Hilb_\Gamma,\Hilb_\Delta)$ be a sub-unitary matrix. Then for all $\varepsilon\geq 0$, there exists an $HT$-program $(\Gamma;S)$ such that $\Gamma^S=\Delta$ and $d\bigl((U(\cdot)U^{\dag},U),\sem{S}_{\Gamma}\bigr)\leq \varepsilon$.
\end{lem}
\begin{proof}
Fix $\varepsilon\geq 0$. There exist matrices $U_1,U_2,U_3$ such that
\[
 U'= \begin{bmatrix}
U & U_1\\
U_2 & U_3
\end{bmatrix} \in\mathcal{L}(\Hilb_{\bar{\Gamma}})
\]
is unitary, where $\bar{\Gamma}$ is an extension of $\Gamma$. There exists a unitary gate $V'$ consisting of a sequence of CNOT, $H$, and $T$ gates such that $d(U',V')\leq \varepsilon$. We write:
\[
 V'= \begin{bmatrix}
V & V_1\\
V_2 & V_3
\end{bmatrix} \in\mathcal{L}(\Hilb_{\bar{\Gamma}})
\]
such that $U$ and $V$ have the same dimensions. By definition, $V$ is sub-unitary, and using the construction from the proof of Lemma~\ref{lem:sub_unitary_prgm}, there exists a statement $S$ such that $\sem{S}_{\Gamma}=(V(\cdot)V^{\dag},V)$. We have:
\[
 d(U,V)= \max_{\ket{\psi}} \|(U-V)\ket{\psi}\| \leq \max_{\ket{\psi}} \left\|(U'-V')\begin{pmatrix}
                                        \ket{\psi} \\
                                        0
                                       \end{pmatrix}
\right\| \leq d(U',V') \leq \varepsilon .
\]
Combining this with Equations~\ref{eq:approx_1} and~\ref{eq:approx_2}, we obtain:
\[
 d\bigl((U(\cdot)U^{\dag},U),\sem{S}_{\Gamma}\bigr)\leq 2d(U,V)\leq 2\varepsilon. \qedhere
\]
\end{proof}

We can now prove universality for $HT$-programs: we can reasonably restrict the language to only having a finite set of qubit gates and retain a universality property.
\begin{prop}[$HT$-Universality]
 \label{prop:approx_universality}
 Let $\Gamma$ and $\Delta$ be environments and $(\mathcal{C},F)\in\tb{CQO}_{\Gamma;\Delta}$ be a coherent quantum operation. Then, for all $\varepsilon\geq 0$, there exists an $HT$-program $(\Gamma;S)$ such that $\Gamma^S=\Delta$ and $d\bigl((\mathcal{C},F),\sem{S}_{\Gamma})\leq \varepsilon$.
\end{prop}
\begin{proof}
As shown in the previous subsection, there exist an environment $\Delta'$, a sub-unitary map $U\in\mathcal{L}(\Hilb_\Gamma, \Hilb_{\Delta',\Delta})$, and a unitary map $V\in\mathcal{L}(\Hilb_{\Delta'})$ such that the statement
\[
 S\ \triangleq\ S_U ;\  S_{V^{\dag}}
 ;\ \tb{discard }\qr_1, \ldots, \qr_s
\]
satisfies $\sem{S}_{\Gamma}=(\mathcal{C},F)$. Here $S_U$ and $S_{V^\dag}$ can contain any single-qubit unitary gates. In order to approximate $(\mathcal{C},F)$ using only $H$ and $T$, it is sufficient to find statements whose interpretations approximate those of $S_U$ and $S_{V^\dag}$. Indeed, we show the following inequality for all quantum operations $\mathcal{C}_1,\mathcal{C}_2,\mathcal{D}_1,\mathcal{D}_2$ of compatible dimensions ($\mathcal{C}_1$ and $\mathcal{C}_2$ can be though of as approximations of $\mathcal{D}_2$ and $\mathcal{D}_2$, respectively).
\begin{align*}
 \|\mathcal{C}_2 \circ \mathcal{C}_1 - \mathcal{D}_2 \circ\mathcal{D}_1 \|_{\diamond} &\leq \|\mathcal{C}_2 \circ (\mathcal{C}_1-\mathcal{D}_1)\|_{\diamond} + \|(\mathcal{C}_2-\mathcal{D}_2)\circ\mathcal{D}_1\|_{\diamond} \\
 &\leq \|\mathcal{C}_2 \|_{\diamond} \| \mathcal{C}_1-\mathcal{D}_1\|_{\diamond} + \|\mathcal{C}_2-\mathcal{D}_2\|_{\diamond}\|\mathcal{D}_1\|_{\diamond} \\
 &\leq \| \mathcal{C}_1-\mathcal{D}_1\|_{\diamond} + \|\mathcal{C}_2-\mathcal{D}_2\|_{\diamond},
\end{align*}
where we used the fact that the diamond norm is sub-multiplicative for sequential composition~\cite{aharonov_kitaev_nisan,watrous} and that the diamond norm of quantum operations is at most~1.

Lemmas~\ref{lem:approx_unitary_prgm} and~\ref{lem:approx_sub_unitary_prgm} ensure the existence of suitable approximations for $S_U$ and $S_{V^{\dag}}$.
\end{proof}

\section{Adequacy}
\label{section:adequacy}

\subsection{Main Result}

We have independently defined an operational semantics and a denotational semantics for our language. The operational semantics, on the one hand, defines the evolution of an input state induced by a program using a default transition. The denotational semantics, on the other hand, assigns to each program a coherent quantum operation describing its behavior. In this section, we show adequacy between the two semantics, which confirms that the two approaches describe the same program behavior.
\begin{thm}[Adequacy]
\label{thm:adequacy}
Let $[S,\ket{\psi}]_{\Gamma} \in \tb{Conf}$ be a configuration such that $\sem{S}_{\Gamma}=(\mathcal{C},F)$ and $\mathscr{M}([S,\ket{\psi}]_\Gamma)=\{(\ket{\psi'_i},\nu_i)\}_i$. Then:
\[
\mathcal{C}(\ket{\psi}\bra{\psi}) = \sum_i \ket{\psi'_i}\bra{\psi_i'} \quad \text{ and }\quad
F\ket{\psi} = \sum_i \nu_i\ket{\psi'_i}.
\]
\end{thm}

This result is shown by structural induction on $S$. For the $\tb{while}$ statement, we will use the fact that the denotational semantics is given by the supremum of iterated applications of $\mathscr{F}_{\q}^{S}$ to $(0,0)$, and proceed by induction on the number of iterations of the $\tb{while}$ loop.
\begin{proof}
First, following from Lemma~\ref{lem:prob_distr} and absolute summability~\cite{topology}, the sums on the right hand sides of each equation are always well defined.

We show by induction on the structure of $S$ that for all environment $\Gamma$ such that some judgment $\Gamma\vdash S\triangleright \Delta$ is derivable, and for all subnormalized state $\ket{\psi}\in St(\Hilb_\Gamma)$, both equalities hold. The result is clear when $S$ is one of $\tb{new qbit}$, $\tb{discard}$, unitary application, or $\tb{skip}$. We treat the remaining cases in more detail.

Suppose that the statement has the form $S_0;S_1$. Let $\mathscr{M}([S_0,\ket{\psi}]_\Gamma)=\{(\ket{\psi_i'},\nu_i)\}_i$, and for all $i$, $\mathscr{M}([S_1,\ket{\psi_i'}]_{\Gamma^{S_0}})=\{(\ket{\psi_{ij}''},\mu_{ij})\}_{j}$.
Since the derivation of any transition from $[S_0;S_1,\ket{\psi}]_\Gamma$ must end with Rule (S), we have $\mathscr{M}([S_0;S_1,\ket{\psi}]_\Gamma) = \{(\ket{\psi_{ij}''},\nu_i\mu_{ij})\}_{ij}$.
On the other hand, we have $\sem{S_0;S_1}_{\Gamma} = \sem{S_1}_{\Gamma^{S_0}}\circ \sem{S_0}_{\Gamma}$. Writing $\sem{S_0}_{\Gamma}=(\mathcal{C},F)$ and $\sem{S_1}_{\Gamma^{S_0}} = (\mathcal{D},G)$, we have $\sem{S_0;S_1}  = (\mathcal{D}\circ\mathcal{C},GF)$.  Therefore, the following equalities hold by induction hypothesis on $S_0$ and $S_1$.
 \begin{align*}
   \mathcal{D}\circ\mathcal{C}(\ket{\psi}\bra{\psi}) &= \mathcal{D}\left(\sum_i\ket{\psi_i'}\bra{\psi_i'}\right)
  = \sum_{i,j} \ket{\psi_{ij}''}\bra{\psi_{ij}''} \\
   GF\ket{\psi} & = G\sum_i \nu_i\ket{\psi_i'} = \sum_{i,j}\nu_i\mu_j\ket{\psi_{ij}''}
 \end{align*}

Next, suppose that the statement has the form $\tb{meas}\ \q\ (0 \rightarrow S_0,1\rightarrow S_1)$. Let $\mathscr{M}([S_0, \ket{0}\bra{0}_{\q}\ket{\psi}]_{\q,\Gamma}) = \{(\ket{\varphi_i},\nu_i)\}_i$ and $\mathscr{M}([S_1, \ket{1}\bra{1}_{\q}\ket{\psi}]_{\q,\Gamma}) = \{(\ket{\phi_j},\mu_j)\}_j$. Since the derivation of a transition from $[\tb{meas}\ \q\ (0 \rightarrow S_0,1\rightarrow S_1),\ket{\psi}]_{\q,\Gamma}$ must end with either Rule (M$_0$) or Rule (M$_1$), we have
\[
\mathscr{M}( [\tb{meas}\ \q\ (0 \rightarrow S_0,1\rightarrow S_1),\ket{\psi}]_{\q,\Gamma}) = \left\{\left(\ket{\varphi_i},\nu_i\right)\right\}_i \uplus \left\{\left(\ket{\phi_j},0\right)\right\}_j.
\]
On the other hand, $\sem{\tb{meas }\q \ (0\rightarrow S_0,1\rightarrow S_1)}_{\q,\Gamma}=\overline{\tb{meas}}_{\q}[\sem{S_0}_{\q,\Gamma},\sem{S_1}_{\q,\Gamma}]$. With $\sem{S_0}_{\q,\Gamma}=(\mathcal{C}_0,F_0)$ and $\sem{S_1}_{\q,\Gamma}=(\mathcal{C}_1,F_1)$, we have:
\[\sem{\tb{meas }\q \ (0\rightarrow S_0,1\rightarrow S_1)}_{\q,\Gamma} = \left(\mathcal{C}_0\circ \mathcal{P}_{0}^{\q} + \mathcal{C}_1\circ \mathcal{P}_{1}^{\q},F_0\ket{0}\bra{0}_\q\right)\!.
\]
By induction hypothesis on $S_0$ and $S_1$, the following can be derived:
\begin{align*}
 \bigl(\mathcal{C}_0\circ \mathcal{P}_{0}^{\q} + \mathcal{C}_1\circ \mathcal{P}_{1}^{\q}\bigr) (\ket{\psi}\bra{\psi}) &= \mathcal{C}_0(\ket{0}\bra{0}_\q (\ket{\psi}\bra{\psi})\ket{0}\bra{0}_\q) +\mathcal{C}_1(\ket{1}\bra{1}_\q (\ket{\psi}\bra{\psi})\ket{1}\bra{1}_\q) \\
 &= \sum_i \ket{\varphi_i}\bra{\varphi_i}+\sum_j\ket{\phi_j}\bra{\phi_j}, \\
  F_0\ket{0}\bra{0}_\q \ket{\psi} &  = \sum_i \nu_i\ket{\varphi_i}.
\end{align*}

Next, suppose that the statement has the form $\tb{qcase}\ \q\ (0 \rightarrow S_0,1\rightarrow S_1)$. Let $\mathscr{M}([S_0,\bra{0}_\q\ket{\psi}]_\Gamma) = \{(\ket{\varphi_i},\nu_i)\}_i$ and $\mathscr{M}([S_1,\bra{1}_\q\ket{\psi}]_{\Gamma}) = \{(\ket{\phi_j},\mu_j)\}_j$. Since the derivation of a transition from $[\tb{qcase}\ \q\ (0 \rightarrow S_0,1\rightarrow S_1),\ket{\psi}]_{\q,\Gamma}$ must end with Rule (Q),
\[
 \mathscr{M}([\tb{qcase}\ \q\ (0 \rightarrow S_0,1\rightarrow S_1),\ket{\psi}]_{\q,\Gamma}) = \{(\mu_j \ket{0}_\q\otimes \ket{\varphi_i}_{\Gamma^{S_0}} + \nu_i\ket{1}_\q \otimes \ket{\phi_j}_{\Gamma^{S_1}},\nu_i\mu_j)\}_{i,j},
\]
where $\Gamma^{S_0}=\Gamma^{S_1}$.
Moreover, $\sem{\tb{qcase }\q \ (0\rightarrow S_0,1\rightarrow S_1)}_{\q,\Gamma}=\overline{\tb{qcase}}_{\q}[\sem{S_0}_{\Gamma},\sem{S_1}_{\Gamma}]$. With $\sem{S_0}_{\Gamma}=(\mathcal{C}_0,F_0)$ and $\sem{S_1}_{\Gamma}=(\mathcal{C}_1,F_1)$, we obtain:
\begin{align*}
  & \sem{\tb{qcase }\q \ (0\rightarrow S_0,1\rightarrow S_1)}_{\q,\Gamma}  \\
  &\quad = \bigl(\mathcal{P}_0^\q \otimes \mathcal{C}_0  + \mathcal{P}_1^\q \otimes \mathcal{C}_1 + (\ket{0}\bra{0}_\q \otimes F_0)(\cdot)(\ket{1}\bra{1}_\q\otimes F_1^{\dag}) \\
 & \qquad +(\ket{1}\bra{1}_\q \otimes F_1)(\cdot)(\ket{0}\bra{0}_\q\otimes F_0^{\dag}),\ket{0}\bra{0}_\q\otimes F_0 + \ket{1}\bra{1}_\q\otimes F_1\bigr).
\end{align*}
First, we have:
\begin{align*}
 \sum\nolimits_{i,j} & (\mu_j \ket{0}_\q\otimes \ket{\varphi_i}  + \nu_i\ket{1}_\q \otimes \ket{\phi_j})(\mu_j \bra{0}_\q\otimes \bra{\varphi_i} + \nu_i\bra{1}_\q \otimes \bra{\phi_j})  \\
&= \sum\nolimits_{i,j}\mu_j \ket{0}\bra{0}_\q \otimes \ket{\varphi_i}\bra{\varphi_j} + \nu_i\ket{1}\bra{1}_\q\otimes \ket{\phi_j}\bra{\phi_j}\\
&\qquad + \mu_j\nu_i \ket{0}\bra{1}_\q \otimes \ket{\varphi_i}\bra{\phi_i} + \nu_i\mu_j\ket{1}\bra{0}_\q\otimes \ket{\phi_j}\bra{\varphi_i} \\
 &= \ket{0}\bra{0}_\q \otimes \sum\nolimits_i \ket{\varphi_i}\bra{\varphi_i} + \ket{1}\bra{1}_\q\otimes \sum\nolimits_j \ket{\phi_j}\bra{\phi_j} \\
 & \quad + \ket{0}\bra{1}_\q \otimes \Bigl(\sum\nolimits_i \nu_i \ket{\varphi_i}\Bigr)\Bigl(\sum\nolimits_j \mu_j \bra{\phi_j} \Bigr) + \ket{1}\bra{0}_\q \otimes \Bigl(\sum\nolimits_j \mu_i\ket{\phi_j}\Bigr)\Bigl(\sum\nolimits_i \nu_i \bra{\varphi_i} \Bigr) \\
 &= \ket{0}\bra{0}_\q \otimes \sum\nolimits_i \ket{\varphi_i}\bra{\varphi_i} + \ket{1}\bra{1}_\q\otimes \sum\nolimits_j \ket{\phi_j}\bra{\phi_j} \\
 & \quad + \ket{0}\bra{1}_\q \otimes F_0\bra{0}_\q \ket{\psi}\braket{\psi}{1}_\q F_1^{\dag} + \ket{1}\bra{0}_\q \otimes F_1\bra{1}_\q \ket{\psi}\braket{\psi}{0}_\q F_0^{\dag},
\end{align*}
where, to obtain the second equality, we used Lemma~\ref{lem:prob_distr} and, to obtain the third equality, we used the induction hypothesis. Then, by induction hypothesis we also have:
\begin{align*}
 &\bigl[\mathcal{P}_0^\q \otimes \mathcal{C}_0+ \mathcal{P}_1^\q \otimes \mathcal{C}_1 + (\ket{0}\bra{0}_\q \otimes F_0)(\cdot)(\ket{1}\bra{1}_\q\otimes F_1^{\dag}) \\
 &+(\ket{1}\bra{1}_\q \otimes F_1)(\cdot)(\ket{0}\bra{0}_\q\otimes F_0^{\dag})\bigr] (\ket{\psi}\bra{\psi})  \\
 & \quad = \ket{0}\bra{0}_\q \otimes \mathcal{C}_0(\bra{0}_\q (\ket{\psi}\bra{\psi})\ket{0}_\q)
 + \ket{1}\bra{1}_\q \otimes \mathcal{C}_1(\bra{1}_\q (\ket{\psi}\bra{\psi})\ket{1}_\q) \\
 &\qquad + \ket{0}\bra{1}_\q \otimes F_0 \bra{0}_\q (\ket{\psi}\bra{\psi}) \ket{1}_\q F_1^{\dag}
 + \ket{1}\bra{0}_\q \otimes F_1 \bra{1}_\q (\ket{\psi}\bra{\psi}) \ket{0}_\q F_0^{\dag} \\
 & \quad =  \ket{0}\bra{0}_\q \otimes \sum\nolimits_i \ket{\varphi_i}\bra{\varphi_i} + \ket{1}\bra{1}_\q\otimes \sum\nolimits_j \ket{\phi_j}\bra{\phi_j}  \\
 & \qquad + \ket{0}\bra{1}_\q \otimes F_0 \bra{0}_\q (\ket{\psi}\bra{\psi}) \ket{1}_\q F_1^{\dag}
 + \ket{1}\bra{0}_\q \otimes F_1 \bra{1}_\q (\ket{\psi}\bra{\psi}) \ket{0}_\q F_0^{\dag},
\end{align*}
which matches the previous expression, and
\begin{align*}
\bigl( \ket{0}\bra{0}_\q\otimes F_0 + \ket{1}\bra{1}_\q\otimes F_1\bigr)\ket{\psi} & = \ket{0}_q\otimes F_0 \bra{0}_\q\ket{\psi} + \ket{1}_q\otimes F_1 \bra{1}_\q\ket{\psi} \\
& =\ket{0}_\q \otimes \sum\nolimits_i \nu_i\ket{\varphi_i} + \ket{1}_\q \otimes \sum\nolimits_j \mu_j\ket{\phi_j} \\
& =\sum\nolimits_{ij}\nu_i\mu_j\bigl(\mu_j \ket{0}_\q\otimes \ket{\varphi_i} + \nu_i\ket{1}_\q \otimes \ket{\phi_j}\bigr),
\end{align*}
where to obtain the last equality we used Lemma~\ref{lem:prob_distr}.

Finally, suppose that the statement has the form $\tb{while} \ \q \ \tb{do} \ S$. On the one hand, the denotational semantics is given by the supremum of iterated applications of $\mathscr{F}_{\q}^{S}$ to $(0,0)$:
\begin{align}
\label{eq:while_semantics_2}
 \sem{\tb{while} \ \q\ \tb{do}\ S}_{\q,\Gamma} = \bigvee_{n\in \mathbb{N}}(\mathscr{F}_{\q}^{S})^n(0,0).
\end{align}
On the other hand, for all $\ket{\psi}$, let us consider $\mathscr{M}([\tb{while}\ \q \ \tb{do}\ S,\ket{\psi}]_{\q,\Gamma})$. Derivations of transitions from this configuration can have arbitrary depth regardless of $S$ because the Rule (W$_1$) can be applied arbitrarily many times. In order to bound the number of iterations of the $\tb{while}$ loop, we modify our language by adding a $\tb{while}$ statement indexed by a maximum number of iterations $n\geq 1$ (counting the step during which the loop is exited). New rules are added to the operational semantics:
\[
 \begin{prooftree}
 \infer0[(W$_{n,0}$)]{[\tb{while}_n\ \q \tb{ do}\ S,\ket{\psi}]_{\q,\Gamma}\stackrel{1}{\to} \ket{0}\bra{0}_{\q}\ket{\psi}}
 \end{prooftree}\raisebox{1.6ex}{,\  \text{for all $n\geq 1$,}}
\]
and
\[
 \begin{prooftree}
\hypo{[S;\tb{while}_{n-1} \ \q \ \tb{do}\ S, \ket{1}\bra{1}_{\q}\ket{\psi}]_{\q,\Gamma}\stackrel{\nu}{\to} \ket{\psi'}}
\infer1[(W$_{n,1}$)]{[\tb{while}_n\ \q\ \tb{do}\ S,\ket{\psi}]_{\q,\Gamma}\stackrel{0}{\to} \ket{\psi'}}
 \end{prooftree},\  \text{for all $n\geq 2$.}
\]
If for some $n$, $[\tb{while}_n \ \q\ \tb{do}\ S,\ket{\psi}]_{\q,\Gamma}\stackrel{\nu}{\to}\ket{\psi'}$ then $[\tb{while} \ \q\ \tb{do}\ S,\ket{\psi}]_{\q,\Gamma}\stackrel{\nu}{\to}\ket{\psi'}$. This can be shown by induction on $n$. And, conversely, if $[\tb{while} \ \q\ \tb{do}\ S,\ket{\psi}]_{\q,\Gamma}\stackrel{\nu}{\to}\ket{\psi'}$ then there exists $n$ such that $[\tb{while}_n \ \q\ \tb{do}\ S,\ket{\psi}]_{\q,\Gamma}\stackrel{\nu}{\to}\ket{\psi'}$. Moreover, there is a one-to-one correspondence between derivations of these transitions.
Therefore:
\begin{align}
\label{eq:while_outcomes_2}
\mathscr{M}([\tb{while}\ \q \ \tb{do}\ S,\ket{\psi}]_{\q,\Gamma}) = \bigcup_{n\geq 1}\mathscr{M}([\tb{while}_n\ \q \ \tb{do}\ S,\ket{\psi}]_{\q,\Gamma}),
\end{align}
where $\cup$ denotes multiset union.
In order to show a correspondence between Equation~\ref{eq:while_semantics_2} and Equation~\ref{eq:while_outcomes_2}, we prove the following by induction on $n\geq 1$: for all $\ket{\psi}\in St(\Hilb_{\q,\Gamma})$, writing $\mathscr{M}([\tb{while}_n\ \q \ \tb{do}\ S,\ket{\psi}]_{\q,\Gamma})=\{(\ket{\psi_i'},\nu_i)\}_i$ and $ (\mathscr{F}_{\q}^{S})^n(0,0)=(\mathcal{C}_n,F_n)$, we have:
\begin{align*}
\mathcal{C}_n(\ket{\psi}\bra{\psi}) = \sum\nolimits_i \ket{\psi'_i}\bra{\psi_i'} \quad \text{ and }\quad
F_n\ket{\psi} = \sum\nolimits_i \nu_i\ket{\psi'_i}.
\end{align*}
For $n=1$, $\mathscr{M}([\tb{while}_1\ \q \ \tb{do}\ S,\ket{\psi}]_{\q,\Gamma}) = \left\{\left(\ket{0}\bra{0}_\q\ket{\psi},1\right)\right\}$
and $\mathscr{F}_{\q}^{S}(0,0)=\bigl(\mathcal{P}_0^\q,\ket{0}\bra{0}_\q\bigr)$,
and we can check that both equalities hold.
Suppose the result holds for a fixed $n\geq 1$. For all multiset $\{(\ket{\phi_i},\mu_i)\}_i$ and $\lambda\in\mathbb{C}$, define $\{(\ket{\phi_i},\mu_i)\}_i\cdot \lambda \triangleq \{(\ket{\phi_i},\lambda \mu_i)\}_i$, where all occurrences of the element $(0,0)$ are removed.
We have:
\begin{align}
\label{eq:while_indexed_outcomes_2}
\begin{split}
 & \mathscr{M}([\tb{while}_{n+1}\ \q\ \tb{do}\ S,\ket{\psi}]_{\q,\Gamma}) \\
 &\quad = \left\{\left(\ket{0}\bra{0}_{\q}\ket{\psi},1\right)\right\} \uplus \mathscr{M}\bigl(\left[S;\tb{while}_n\ \q\ \tb{do}\ S,\ket{1}\bra{1}_\q\ket{\psi}\right]_{\q,\Gamma}\bigr)\cdot 0.
 \end{split}
\end{align}

Let $\mathscr{M}([S, \ket{1}\bra{1}_\q\ket{\psi}]_{\q,\Gamma})=\{(\ket{\psi_i'},\nu_i)\}_i$ and for all $i$, $\mathscr{M}([\tb{while}_n\ \q\ \tb{do}\ S, \ket{\psi'_i}]_{\q,\Gamma})=\{(\ket{\psi''_{ij}},\mu_{ij})\}_j$. Then, $ \mathscr{M}\bigl(\left[S;\tb{while}_n\ \q\ \tb{do}\ S,\ket{1}\bra{1}_\q\ket{\psi}\right]_{\q,\Gamma}\bigr) = \{(\ket{\psi''_{ij}},\mu_{ij})\}_{ij}$. Moreover,
\begin{align*}
 (\mathscr{F}_{\q}^{S})^{n+1}(0,0)&=\mathscr{F}_{\q}^{S}(\mathcal{C}_n,F_n) \\
 &=\overline{\tb{meas}}_\q \left[(\mathcal{I}_{\q,\Gamma},I_{\q,\Gamma}),(\mathcal{C}_n,F_n)\circ \sem{S}_{\q,\Gamma}\right] \\
 &= \left(\mathcal{P}_0^\q + \mathcal{C}_n \circ \mathcal{D} \circ \mathcal{P}_1^\q,\ket{0}\bra{0}_\q \right)\!,
\end{align*}
where we write $\sem{S}_{\q,\Gamma}=(\mathcal{D},G)$. Therefore:
\begin{align*}
 \mathcal{C}_{n+1}(\ket{\psi}\bra{\psi})
 &=\mathcal{P}_0^\q(\ket{\psi}\bra{\psi})+\mathcal{C}_n\circ\mathcal{D}\circ\mathcal{P}_1^\q(\ket{\psi}\bra{\psi}) \\
 &= \ket{0}\bra{0}_\q (\ket{\psi}\bra{\psi})\ket{0}\bra{0}_\q +\mathcal{C}_n\circ\mathcal{D}(\ket{1}\bra{1}_\q(\ket{\psi}\bra{\psi})\ket{1}\bra{1}_\q) \\
 &= \ket{0}\bra{0}_\q (\ket{\psi}\bra{\psi})\ket{0}\bra{0}_\q +\mathcal{C}_n \left(\sum\nolimits_i \ket{\psi_i'}\bra{\psi_i'}\right) && \text{\small(by induction hyp.)} \\
 &= \ket{0}\bra{0}_\q (\ket{\psi}\bra{\psi})\ket{0}\bra{0}_\q +\sum\nolimits_{i,j} \ket{\psi_{ij}''}\bra{\psi_{ij}''}.  && \text{\small(by induction hyp.)}
\end{align*}
Combining this with Equation~\ref{eq:while_indexed_outcomes_2}, we obtain the first result. Then, we have $F_{n+1}\ket{\psi}=\ket{0}\bra{0}_\q\ket{\psi}$, which gives us the second result.

We conclude the proof by taking the limit in $n\rightarrow +\infty$ and using the fact that if $\bigvee_n (\mathscr{F}_{\q}^{S})^n(0,0)=(\mathcal{C},F)$, then $\mathcal{C}(\ket{\psi}\bra{\psi})=\bigvee_n \mathcal{C}_n(\ket{\psi}\bra{\psi})$.
\end{proof}

\begin{exa}
 By Examples~\ref{ex:coin_op} and~\ref{ex:coin_den}, we have for all $\smx{\alpha \\ \beta} \triangleq \alpha\ket{0}+\beta\ket{1}\in St(\Hilb_{\q})$:
\begin{align*}
\mathscr{M}([\text{COIN},\smx{\alpha \\ \beta}]_\q) &=\ \left\{\left(\alpha \ket{0},1\right)\right\}\uplus \left\{\left(-\beta\left(-\frac{1}{\sqrt{2}}\right)^k\ket{0},0\right)\ \middle|\ k\geq 1\right\}\!,\\[1.5ex]
\sem{\text{COIN}}_{\q} &=\ \left(Tr(\cdot)\ket{0}\bra{0}, \ket{0}\bra{0}\right).
\end{align*}
To demonstrate adequacy, we can check that
\begin{align*}
Tr\bigl(\smx{\alpha \\ \beta}\smx{\alpha \\ \beta}^\dagger \bigr)\ket{0}\bra{0}&= (|\alpha|^2+|\beta|^2)\ket{0}\bra{0} =  \ |\alpha|^2 \ket{0}\bra{0} + \sum_{k\geq 1}|\beta|^2 \frac{1}{2^k}\ket{0}\bra{0}, \\
 \ket{0}\bra{0} \smx{\alpha \\ \beta}&=\alpha\ket{0}.
\end{align*}
\end{exa}

As a consequence, we also obtain a simple formula for the probability of termination of a program on a given input.

\begin{cor}
\label{cor:proba}
 Given a configuration $[S,\ket{\psi}]_{\Gamma} \in \tb{Conf}$ such that $\sem{S}_{\Gamma}=(\mathcal{C},F)$, the probability of termination of $S$ on input state $\ket{\psi}$ is given by
 \[
  p(S,\ket{\psi})_\Gamma = \frac{Tr(\mathcal{C}(\ket{\psi}\bra{\psi}))}{Tr(\ket{\psi}\bra{\psi})}\raisebox{0.4ex}{.}
 \]
\end{cor}

\subsection{Choice of Coefficients}
\label{section:choice}

In this section, we explain our choice to define the operational semantics by fixing default transitions (indicated by the extra bit $\nu$ in Section~\ref{section:operational_semantics}, Figure~\ref{fig:semantics_operational}) and how this relates to the adequacy theorem. In principle, the binary coefficient $\nu$ could more generally be defined as any complex amplitude with norm at most 1. In such a case, the most general formulation of the rules for the probabilistic statements $\tb{discard}$, $\tb{meas}$ and $\tb{while}$ would be:
\[
  \begin{prooftree}
 \infer0[(D$_{0}$)]{[\tb{discard }\q,\ket{\psi}]\stackrel{\alpha}{\to} \bra{0}_{\q} \ket{\psi}}
 \end{prooftree}
 \qquad
 \begin{prooftree}
 \infer0[(D$_{1}$)]{[\tb{discard }\q,\ket{\psi}]\stackrel{\beta}{\to} \bra{1}_{\q} \ket{\psi}}
 \end{prooftree}
 \]
 \[
 \begin{prooftree}
 \hypo{[S_0,\ket{0}\bra{0}_{\q}\ket{\psi}]\stackrel{\nu}{\to} \ket{\psi'}}
 \infer1[(M$_0$)]{[\tb{meas }\q\ (0\rightarrow S_0,1 \rightarrow S_1),\ket{\psi}]\stackrel{\nu\gamma}{\to} \ket{\psi'}}
 \end{prooftree}
 \quad
 \begin{prooftree}
 \hypo{[S_1,\ket{1}\bra{1}_{\q}\ket{\psi}]\stackrel{\nu}{\to} \ket{\psi'}}
 \infer1[(M$_1$)]{[\tb{meas }\q\ (0\rightarrow S_0,1 \rightarrow S_1),\ket{\psi}]\stackrel{\nu\delta}{\to} \ket{\psi'}}
 \end{prooftree}
 \]
\[
 \begin{prooftree}
 \infer0[(W$_0$)]{[\tb{while }\q \tb{ do}\ S,\ket{\psi}]\stackrel{\epsilon}{\to} \ket{0}\bra{0}_{\q}\ket{\psi}}
 \end{prooftree}
 \qquad
 \begin{prooftree}
\hypo{[S;\tb{while} \ \q \ \tb{do}\ S,\ket{1}\bra{1}_{\q}\ket{\psi}]\stackrel{\nu}{\to} \ket{\psi'}}
\infer1[(W$_1$)]{[\tb{while}\ \q\ \tb{do}\ S,\ket{\psi}]\stackrel{\nu\eta}{\to} \ket{\psi'}}
 \end{prooftree}
\]
with $\alpha,\beta,\gamma,\delta,\epsilon,\eta\in\mathbb{C}$.

For our results to hold, these coefficients must satisfy a few extra conditions:
\begin{itemize}
 \item With binary coefficients, we have the property that each configuration has exactly one default transition. That is, assuming that $\mathscr{M}([S,\ket{\psi}])=\{(\ket{\psi'_i},\nu_i)\}_{i\in I}$ ($\ket{\psi}\neq 0$) there is exactly one index $i$ such that $\nu_i=1$.\footnote{See Lemma~\ref{lem:prob_distr} for details.} With more general complex coefficients, this would instead become $\sum_i |\nu_i|^2 =1$. For this property to hold, the coefficients must satisfy:
 \[
  |\alpha|^2+|\beta|^2=|\gamma|^2+|\delta|^2=|\epsilon|^2+|\eta|^2=1,\qquad
  \epsilon\neq 0.
\]
The last condition $\epsilon\neq 0$ is related to the $\tb{while}$ rule: with $\epsilon =0$, all transitions from a $\tb{while}$ statement would have a coefficient 0, which is forbidden.

\item In order for adequacy (Theorem~\ref{thm:adequacy}) to hold, the denotational semantics must also be modified to match the operational semantics. This requires the following changes. First, in the adequacy statement itself, $F\ket{\psi} = \sum_i \nu_i\ket{\psi'_i}$ must be rewritten as $F\ket{\psi} = \sum_i \bar{\nu}_i\ket{\psi'_i}$. Then we must redefine the following operations in the denotational semantics:
\begin{align*}
\sem{\tb{discard }\q}_{\q,\Gamma}\ & \triangleq\ \bigl(Tr_{\q},(\bar{\alpha}\bra{0}_\q+\bar{\beta}\bra{1}_\q)\otimes I_{\Gamma}\bigr),\\
 \overline{\tb{meas}}_{\q}[(\mathcal{C},F),(\mathcal{D},G)]\ & \triangleq \ \left(\mathcal{C}\circ \mathcal{P}_{0}^{\q} + \mathcal{D}\circ \mathcal{P}_{1}^{\q}, \gamma F\ket{0}\bra{0}_\q + \delta G\ket{1}\bra{1}_\q\right).
\end{align*}
And, lastly, it is necessary that $\gamma=\epsilon$ and $\delta=\eta$ because the semantics of $\tb{while}$ is based on the same $\overline{\tb{meas}}_\q$ operation as the measurement.
\end{itemize}

Once these conditions are satisfied, our results hold for any choice of coefficients. In particular, universality (Theorem~\ref{thm:universality}) still holds.\footnote{To prove universality for generalized coefficients, one must make some small adjustments to the proof of Theorem~\ref{thm:universality}: occurrences of $\tb{meas}$ must be encoded using a $\tb{qcase}$ statement as is done in Example~\ref{ex:meas}; and $\tb{discard}$ statements must be preceded by appropriate unitaries applied to the qubit to be discarded. The proof otherwise follows the same structure.}

The rules in Figure~\ref{fig:semantics_operational} were obtained by choosing the binary coefficients $\alpha=\gamma=\epsilon=1$ and $\beta=\delta=\eta=0$. While there is no canonical choice of coefficients, binary coefficients offer an interpretation as ``default transitions'', and enable some simplifications. For instance, they create a symmetry between the initialization and discard statements:
\begin{align*}
 \sem{\tb{new qbit }\q}_{\Gamma}\ & \triangleq\ \bigl(\ket{0}\bra{0}_{\q}\otimes \mathcal{I}_{\Gamma},\ket{0}_\q\otimes I_{\Gamma}\bigr), \\
 \sem{\tb{discard }\q}_{\q,\Gamma}\ &\triangleq\ \bigl(Tr_{\q},\bra{0}_\q\otimes I_{\Gamma}\bigr).
\end{align*}
That is, the $\tb{new qbit}$ statement always initializes qubits in the state $\ket{0}$, which also appears in the transformation matrix for $\tb{discard}$. In particular, this ensures that the statement $ \tb{new qbit}\ \q\ ;\ \tb{discard}\ \q$ has the same denotational semantics as $\tb{skip}$, which is convenient and would not necessarily be the case for other choices of coefficients.

\begin{rem}
Notice that one could also consider a generalization of the language where the choice of $\nu$ is left to the programmer, e.g., $\tb{meas }\q\ (0\rightarrow^{\nu_0}  S_0,1 \rightarrow^{\nu_1} S_1)$ for arbitrary $\nu_0, \nu_1\in \mathbb C$ as long as $|\nu_0|^2+|\nu_1|^2=1$. This would not increase the expressive power of the language. For this reason, we opt for the simpler syntax that leaves these parameters implicit.
\end{rem}

\section{Full Abstraction}
\label{section:observational_equivalence}

The final result on the language concerns the notion of observational equivalence. Two programs are said to be observationally equivalent if their observable behavior, in any context, is identical.\footnote{We call this property ``observational equivalence'' rather than ``contextual equivalence'' to avoid confusion with the notion of contextuality in quantum mechanics.} We show that programs are observationally equivalent if and only if they have the same denotational semantics.

\subsection{Observational Equivalence}
We start by introducing a notion of context $C^{\Gamma\rightarrow \Delta}$ as  a statement with exactly one hole~$[\cdot]^{\Gamma\rightarrow \Delta}$.
\begin{defi}[Contexts]
 For all environments $\Gamma$ and $\Delta$, contexts are defined by the following grammar:
 \begin{align*}
  C^{\Gamma\rightarrow \Delta}  ::=  & [\cdot]^{\Gamma\rightarrow \Delta} \ |\  S;C^{\Gamma\rightarrow \Delta}\ |\ C^{\Gamma\rightarrow \Delta};S \\
  | \ & \tb{meas}\ \q \ (0\rightarrow S,1\rightarrow C^{\Gamma\rightarrow\Delta})\ | \ \tb{meas}\ \q \ (0\rightarrow C^{\Gamma\rightarrow \Delta},1\rightarrow S)  \\
  |\ & \tb{while}\ \q\ \tb{do}\ C^{\Gamma\rightarrow \Delta} \ |\ \tb{qcase}\ \q \ (0\rightarrow S,1\rightarrow C^{\Gamma\rightarrow\Delta}) \\
  | \ & \tb{qcase}\ \q \ (0\rightarrow C^{\Gamma\rightarrow\Delta},1\rightarrow S)
 \end{align*}
where $S$ is a statement.
\end{defi}

We introduce the inference rule below in order to extend well-formedness to contexts:
\[
\begin{prooftree}
\infer0[(C)]{\Gamma,\Sigma\vdash [\cdot]^{\Gamma\rightarrow\Delta}\triangleright \Delta,\Sigma}
\end{prooftree}
\]
 Setting $\text{Var}([\cdot]^{\Gamma\rightarrow \Delta})\triangleq\emptyset$, we recursively define the set of variables $\text{Var}(C^{\Gamma\rightarrow \Delta})$ appearing in $C^{\Gamma\rightarrow \Delta}$. For any statement $S$ and context $C^{\Gamma\rightarrow\Delta}$, let $C[S]$ be the statement obtained by substituting $S$ for $[\cdot]^{\Gamma\rightarrow\Delta}$ in $C^{\Gamma\rightarrow\Delta}$.
\begin{defi}[Compatibility]
A context $C^{\Gamma\rightarrow\Delta}$ is \emph{compatible} with a program $(\Gamma;S)$ if
\begin{enumerate}
\item $\Delta = \Gamma^S$,
\item $ \emptyset\vdash C^{\Gamma\rightarrow \Delta}\triangleright \emptyset$ is derivable, and
 \item BV$(S)\cap \text{Var}(C^{\Gamma\rightarrow \Delta})=\emptyset$.
\end{enumerate}
\end{defi}
If $C^{\Gamma\rightarrow\Delta}$ is compatible with $(\Gamma;S)$, then $C[S]$ is well formed; specifically $\emptyset\vdash C[S]\triangleright\emptyset$ is derivable.
The definition of compatibility ensures that i) the context $C^{\Gamma\rightarrow \Delta}$ introduces all the variables of $\Gamma$ and deletes all the variables of $\Delta$, and ii) the context can introduce some extra variables that are not in $S$.

\begin{defi}[Observational equivalence]
\label{def:observational_equiv}
Two programs $(\Gamma;S)$ and $(\Gamma;S')$ such that $\Gamma^S=\Gamma^{S'}$ are \emph{observationally equivalent}, written as $(\Gamma;S)\approx(\Gamma;S')$, if for all context $C^{\Gamma\rightarrow \Delta}$ compatible with both $(\Gamma;S)$ and $(\Gamma;S')$, we have $p(C[S],\ket{})_\emptyset=p(C[S'],\ket{})_\emptyset$.
\end{defi}
Here $\ket{}$ is a one-dimensional input state, distinct from the vacuum state $\ket{\text{vac}}$. Observational equivalence formalizes the idea that there is no difference between the observable behavior of two programs.

\subsection{Main Result and Illustrating Examples}
\begin{thm}[Full abstraction]
\label{thm:fullabstraction}
\[(\Gamma;S)\approx (\Gamma;S')\text{ if and only if  }\sem{S}_\Gamma = \sem{S'}_{\Gamma}.\]
\end{thm}

In order to prove this theorem, we will need the following lemmas:
\begin{lem}
\label{lem:unused_vars_ct}
Suppose that $((\Gamma,\Sigma_1);S)$ and $((\Gamma,\Sigma_2);S)$ are well-formed programs with $\Sigma_1\cap\Sigma_2=\emptyset$. Then:
\begin{itemize}
 \item $S$ does not contain any of the variables in $\Sigma_1$ or $\Sigma_2$, i.e., $\text{\textnormal{Var}}(S)\cap (\Sigma_1,\Sigma_2)=\emptyset$.
\item With $\Delta=(\Gamma,\Sigma_1)^S\cap(\Gamma,\Sigma_2)^S$, we have $(\Gamma,\Sigma_1)^S=\Delta,\Sigma_1$ and $(\Gamma,\Sigma_2)^S=\Delta,\Sigma_2$.
\end{itemize}
\end{lem}
\begin{proof}
This is shown by structural induction on $S$.
\end{proof}

\begin{lem}
\label{lem:unused_vars_id}
Let $((\Gamma,\Sigma);S)$ be a well-formed program such that none of the variables of $\Sigma$ appear in $S$, that is $\Sigma\cap \text{\textnormal{Var}}(S)=\emptyset$. Then $\sem{S}_{\Gamma,\Sigma}$ has the form $(\mathcal{C}\otimes \mathcal{I}_{\Sigma},F\otimes I_{\Sigma})$, where $\mathcal{C}\in \tb{QO}(\Hilb_\Gamma,\Hilb_{\Gamma^S})$ and $F\in\mathcal{L}(\Hilb_\Gamma,\Hilb_{\Gamma^S})$. Moreover, $\mathcal{C}$ and $F$ are independent of $\Sigma$.
\end{lem}
\begin{proof}
We show by structural induction on $S$ that for all environments $\Gamma,\Delta,\Sigma$ such that  $(\Gamma,\Sigma)\vdash S\triangleright (\Delta,\Sigma)$ is derivable and $\text{Var}(S)\cap \Sigma=\emptyset$, the result holds. All the cases are clear except the sequence and $\tb{while}$ statements.

For the sequence case, suppose that $S=S_1;S_2$. The derivation of the judgment $\Gamma,\Sigma\vdash S \triangleright \Gamma^S,\Sigma$ must have as its last step
\[
\begin{prooftree}
 \hypo{\Gamma,\Sigma\vdash S_1 \triangleright \Lambda}
 \hypo{\Lambda\vdash S_2\triangleright \Gamma^S,\Sigma}
 \infer2[
 ]{\Gamma,\Sigma \vdash S_1;S_2\triangleright \Gamma^S,\Sigma}
\end{prooftree},
\]
for some environment $\Lambda$. Since $S_1$ does not contain any of the variables of $\Sigma$, we can show by structural induction on $S_1$ that $\Lambda$ has the form $\Lambda',\Sigma$.
Therefore, by induction hypothesis, $\sem{S_1}_{\Gamma,\Sigma}=(\mathcal{C}_1\otimes \mathcal{I}_{\Sigma},F_1\otimes I_{\Sigma})$ (with  $\mathcal{C}_1\in \tb{QO}(\Hilb_\Gamma,\Hilb_{\Lambda'})$ and $F_1\in\mathcal{L}(\Hilb_\Gamma,\Hilb_{\Lambda'})$) and  $\sem{S_2}_{\Lambda',\Sigma}=(\mathcal{C}_2\otimes \mathcal{I}_{\Sigma},F_2\otimes I_{\Sigma})$ (with  $\mathcal{C}_2\in \tb{QO}(\Hilb_{\Lambda'},\Hilb_{\Gamma^S})$ and $F_2\in\mathcal{L}(\Hilb_{\Lambda'},\Hilb_{\Gamma^S})$); and $\mathcal{C}_1,\mathcal{C}_2,F_1,F_2$ are independent of $\Sigma$. Therefore, $\sem{S}_{\Gamma,\Sigma}=((\mathcal{C}_2\circ\mathcal{C}_1)\otimes \mathcal{I}_{\Sigma},F_2F_1\otimes I_{\Sigma})$.

Then, suppose that $S=\tb{while}\ \q \ \tb{do}\ S'$ and $(\q,\Gamma,\Sigma)\vdash S\triangleright(\q,\Gamma,\Sigma)$ is derivable with $\Sigma\cap \text{Var}(S')=\emptyset$.
We have $\sem{S}_{\q,\Gamma,\Sigma}=\bigvee_{n}F^n_{\q,S'}(0,0)$, where for all $(\mathcal{C},F)\in\tb{CQO}_{\q,\Gamma,\Sigma;\q,\Gamma,\Sigma}$:
\[
 \mathscr{F}_{\q}^{S'}(\mathcal{C},F)=\overline{\tb{meas}}_{\q}\bigl[(\mathcal{I}_{\q,\Gamma,\Sigma},I_{\q,\Gamma,\Sigma}), (\mathcal{C},F)\circ \sem{S'}_{\q,\Gamma,\Sigma}\bigr]
\]
By induction hypothesis, $\sem{S'}_{\q,\Gamma,\Sigma}=(\mathcal{D}\otimes \mathcal{I}_{\Sigma},G\otimes I_{\Sigma})$ (with  $\mathcal{D}\in \tb{QO}(\Hilb_{\q,\Gamma}, \Hilb_{\q,\Gamma})$ and $G\in\mathcal{L}(\Hilb_{\q,\Gamma},\Hilb_{\q,\Gamma})$); and $\mathcal{D}$ and $G$ are independent of $\Sigma$. Using this fact, we show by induction that for all $n$, $(\mathscr{F}_{\q}^{S'})^n(0,0)$ has the form $(\mathscr{F}_{\q}^{S'})^n(0,0)=(\mathcal{C}_n\otimes \mathcal{I}_{\Sigma},F_n\otimes I_{\Sigma})$, for some $\mathcal{C}_n,F_n$ which are also independent of $\Sigma$.
We write $\sem{S}_{\q,\Gamma, \Sigma}=(\mathcal{C},F)$. As a consequence of Lemma~\ref{lem:topological_limit_channel}, we have $\mathcal{C}=\lim_{n\rightarrow +\infty} \mathcal{C}_n\otimes \mathcal{I}_{\Sigma}$ and $F=\lim_{n\rightarrow +\infty} F_n\otimes I_{\Sigma}$. By continuity of $\otimes$, $\mathcal{C}$ and $F$ have the forms $\mathcal{C}=\mathcal{C}'\otimes \mathcal{I}_\Sigma$ and $F=F'\otimes I_\Sigma$, where $\mathcal{C}'$ and $F'$ are independent of $\Sigma$.
\end{proof}

Under the assumptions of Lemma~\ref{lem:unused_vars_id}, we can show by induction that $(\Gamma;S)$ is also a well-formed program, and therefore $(\mathcal{C},F)=\sem{S}_{\Gamma}$.

\begin{lem}
 \label{lem:interpretation_context}
 Suppose $(\Gamma_1;S)$, $(\Gamma_1;S')$, $(\Gamma_2;S)$, $(\Gamma_2;S')$ are well-formed programs with $\Gamma_1^S=\Gamma_1^{S'}$ and $\Gamma_2^S=\Gamma_2^{S'}$. Then $\sem{S}_{\Gamma_1}=\sem{S'}_{\Gamma_1}$ if and only if $\sem{S}_{\Gamma_2}=\sem{S'}_{\Gamma_2}$.
\end{lem}
\begin{proof}
We write $\Delta_1=\Gamma_1^S=\Gamma_1^{S'}$ and $\Delta_2=\Gamma_2^S=\Gamma_2^{S'}$. Let $\Gamma=\Gamma_1\cap\Gamma_2$ and $\Delta=\Delta_1 \cap \Delta_2$. We write $\Gamma_1=\Gamma,\Sigma_1$ and $\Gamma_2=\Gamma,\Sigma_2$. By Lemma~\ref{lem:unused_vars_ct}, $S$ and $S'$ do not contain any of the variables of $\Sigma_1$ and $\Sigma_2$, and $\Delta_1=\Delta,\Sigma_1$ and $\Delta_2=\Delta,\Sigma_2$. Therefore by Lemma~\ref{lem:unused_vars_id}, we have:
 \begin{align*}
\sem{S}_{\Gamma_1}&=(\mathcal{C}\otimes \mathcal{I}_{\Sigma_1},F\otimes I_{\Sigma_1}) & \sem{S'}_{\Gamma_1}&=(\mathcal{D}\otimes \mathcal{I}_{\Sigma_1},G\otimes I_{\Sigma_1})\\
\sem{S}_{\Gamma_2}&=(\mathcal{C}\otimes \mathcal{I}_{\Sigma_2},F\otimes I_{\Sigma_2}) & \sem{S'}_{\Gamma_2}&=(\mathcal{D}\otimes \mathcal{I}_{\Sigma_2},G\otimes I_{\Sigma_2})
 \end{align*}
for some $\mathcal{C},\mathcal{D}\in\tb{QO}(\Hilb_\Gamma, \Hilb_\Delta)$ and $F,G\in\mathcal{L}(\Hilb_\Gamma,\Hilb_\Delta)$. Consequently, $\sem{S}_{\Gamma_1}=\sem{S'}_{\Gamma_1}$ if and only if $\sem{S}_{\Gamma_2}=\sem{S'}_{\Gamma_2}$ if and only if $\mathcal{C}=\mathcal{D}$ and $F=G$.
\end{proof}

We are now able to prove the main theorem.
The proof can be summarized as follows:
to show that two programs with the same denotation are observationally equivalent, we use the compositionality of the denotational semantics. To show the reverse implication, we show its contrapositive: for all pair of programs whose denotation is different, we construct a context for which the probability of termination is different. Each program is interpreted as a coherent quantum operation (a quantum operation and a transformation matrix). If their quantum operations are different, we construct a context that prepares an input state for which the outputs are different, and performs a measurement that distinguishes them. If their quantum operations are the same, but their transformation matrices are different, we first perform a $\tb{qcase}$ operation, resulting in new programs with different quantum operations, and proceed as in the first case.
The full proof is detailed below:

\begin{proof}
Suppose that $\sem{S}_\Gamma = \sem{S'}_{\Gamma}$, and let $C^{\Gamma\rightarrow \Delta}$ be a context compatible with $(\Gamma;S)$ and $(\Gamma;S')$ (where $\Delta=\Gamma^S=\Gamma^{S'}$).
The derivation of $\emptyset \vdash C^{\Gamma\rightarrow\Delta} \triangleright \emptyset$ must contain a derivation of $\Gamma,\Sigma\vdash [\cdot]^{\Gamma\rightarrow \Delta}\,\triangleright\, \Delta,\Sigma$ for some environment $\Sigma$. The compatibility condition ensures that $\sem{S}_{\Gamma,\Sigma}$ and $\sem{S'}_{\Gamma,\Sigma}$ are both well defined, and by Lemma~\ref{lem:interpretation_context}, they are equal.
Then, by compositionality of the denotational semantics, we have $\sem{C[S]}_{\emptyset}=\sem{C[S']}_{\emptyset}$. Combined with Corollary~\ref{cor:proba}, this implies that $p(C[S],\ket{})=p(C[S'],\ket{})$. Since this holds for all suitable $C^{\Gamma\rightarrow \Delta}$, we have $(\Gamma;S)\approx (\Gamma;S')$.

Conversely, suppose that $\sem{S}_{\Gamma}\neq\sem{S'}_{\Gamma}$. Let $\sem{S}_{\Gamma}=(\mathcal{C},F)$ and $\sem{S'}_{\Gamma}=(\mathcal{D},G)$. The aim is to construct a context $C^{\Gamma\rightarrow \Delta}$ (where $\Delta=\Gamma^S=\Gamma^{S'}$) that results in different probabilities of termination. First suppose that $\mathcal{C}\neq\mathcal{D}$. Let $\Gamma=\p_1,\ldots,\p_n$ and $\Delta=\q_1,\ldots,\q_m$, and let $\ket{\psi}\in St(\Hilb_\Gamma)$ be a input state such that $\mathcal{C}(\ket{\psi}\bra{\psi})\neq\mathcal{D}(\ket{\psi}\bra{\psi})$.
We write $\mathcal{C}(\ket{\psi}\bra{\psi})-\mathcal{D}(\ket{\psi}\bra{\psi})=\sum_{i=0}^{2^m-1} \lambda_i \ket{\epsilon_i}\bra{\epsilon_i}$, where $(\ket{\epsilon_i})_{0\leq i\leq 2^m-1}$ is an orthonormal basis of $\Hilb_{\Delta}$. Without loss of generality, we assume that $\lambda_0\neq 0$.
Let $U$ be a unitary matrix such that $U \ket{0}^{\otimes n}=\ket{\psi}$, and let $V$ be the unitary matrix defined by $V\ket{\epsilon_i}=\ket{i}$ for $0\leq i\leq 2^m-1$.
Then, we define $C^{\Gamma\rightarrow \Delta}$ as follows:
\begin{align*}
 C^{\Gamma\rightarrow \Delta} \ \triangleq \ & \tb{new qbit }\p_1,\ldots,\p_n ; \ S_U ;\ [\cdot]^{\Gamma\rightarrow \Delta} ;\ S_V ; \\
 & \tb{meas}\  \q_1,\ldots,\q_m\ (0^m\rightarrow \tb{skip} , \ \tb{else}\rightarrow \text{LOOP}); \\
 & \tb{discard} \ \q_1,\ldots,\q_m,
\end{align*}
where $S_U$ and $S_V$ are statements implementing $U$ and $V$ (as constructed in Lemma~\ref{lem:unitary_prgm}), LOOP is the non-terminating statement from Example~\ref{ex:loop_op}, and notations generalizing the $\tb{new qbit}$, $\tb{meas}$, and $\tb{discard}$ statements to multiple variables are used. Then $C^{\Gamma \rightarrow \Delta}$ is compatible with $(\Gamma;S)$ and $(\Gamma;S')$ and we have:
\begin{align*}
 \sem{C[S]}_{\emptyset}=\ & \sem{\tb{discard}\ \q_1,\ldots,\q_m}_{\Delta} \circ \sem{\tb{meas}\  \q_1,\ldots,\q_m\ ( 0^m\rightarrow \tb{skip} , 1\rightarrow \text{LOOP})}_{\Delta}\\
 & \circ \sem{S_V}_{\Delta}\circ \sem{S}_{\Gamma}\circ\sem{S_U}_{\Gamma}\circ\sem{\tb{new qbit }\p_1,\ldots,\p_n}_{\emptyset}.
\end{align*}
We have:
\begin{align*}
\sem{ \tb{meas}\ \q_{1},\ldots,\q_m \ ( 0^{m} \rightarrow \tb{skip}, \tb{else} \rightarrow \text{LOOP})}_{\Delta}
&\ =\ \overline{\tb{meas}}_{\q_{1},\ldots,\q_m} \bigl[\sem{\tb{skip}}_{\Delta},\sem{\text{LOOP}}_{\Delta}\bigr] \\
 &\ =\ \overline{\tb{meas}}_{\q_{1},\ldots,\q_m} \bigl[(\mathcal{I}_{\Delta},I_{\Delta}), (0,0)\bigr] \\
 &\ =\ \bigl((\ket{0}\bra{0})_{\Delta}^{\otimes m} (\cdot)(\ket{0}\bra{0})_{\Delta}^{\otimes m} , (\ket{0}\bra{0})_{\Delta}^{\otimes m}\bigr).
\end{align*}
where $\overline{\tb{meas}}_{\q_{1},\ldots,\q_m}$ generalizes $\overline{\tb{meas}}_{\q}$ to several variables, the first input corresponding to all qubits in the $\ket{0}$ state. Therefore:
\begin{align*}
 \sem{C[S]}_{\emptyset}=&\bigl(Tr_{\Delta},\bra{0}^{\otimes m}_{\Delta}\bigr)\circ \bigl((\ket{0}\bra{0})_{\Delta}^{\otimes m} (\cdot)(\ket{0}\bra{0})_{\Delta}^{\otimes m} , (\ket{0}\bra{0})_{\Delta}^{\otimes m}\bigr)\circ (V(\cdot)V^{\dag},V)\\
 & \circ (\mathcal{C},F)\circ(U(\cdot)U^{\dag},U) \circ ((\ket{0}\bra{0})^{\otimes n}_{\Gamma},\ket{0}^{\otimes n}_{\Gamma}) \\
 =&\left(\bra{\epsilon_0}\mathcal{C}(\ket{\psi}\bra{\psi})\ket{\epsilon_0},\bra{\epsilon_0}F\ket{\psi}\right).
\end{align*}
And similarly, $\sem{C[S']}_{\emptyset}=\left(\bra{\epsilon_0}\mathcal{D}(\ket{\psi}\bra{\psi})\ket{\epsilon_0},\bra{\epsilon_0}G\ket{\psi}\right) $. Therefore by Corollary~\ref{cor:proba},
\[
 p(C[S],\ket{})-p(C[S'],\ket{}) = \bra{\epsilon_0}\mathcal{C}(\ket{\psi}\bra{\psi})\ket{\epsilon_0} - \bra{\epsilon_0}\mathcal{D}(\ket{\psi}\bra{\psi})\ket{\epsilon_0} = \lambda_0\neq 0.
\]
Therefore $(\Gamma;S)\not\approx (\Gamma;S')$.

Now suppose that $\mathcal{C}=\mathcal{D}$ and $F\neq G$. Let $R$ be the statement that performs the necessary qubit initializations and discards to achieve the judgment $\Gamma \vdash R \triangleright \Delta$:
\begin{align*}
 R\ \triangleq\ & \left. \tb{new qbit }\q ; \quad \ \, \right\} \text{for all }\q\in \Delta\setminus \Gamma \\
  & \left. \tb{discard }\p ; \qquad \right\} \text{for all }\p\in \Gamma\setminus \Delta
\end{align*}
Then we have:
\begin{align*}
 \sem{R}_\Gamma =& \left(Tr_{\Gamma\setminus \Delta},\bra{0}^{\otimes |\Gamma\setminus \Delta|}_{\Gamma\setminus\Delta}\otimes I_\Delta \right) \circ \left((\ket{0}\bra{0})^{\otimes|\Delta\setminus\Gamma|}_{\Delta\setminus\Gamma}\otimes \mathcal{I}_{\Gamma}, \ket{0}^{\otimes|\Delta\setminus\Gamma|}_{\Delta\setminus\Gamma}\otimes I_{\Gamma}\right) \\
 =&\left(Tr_{\Gamma\setminus \Delta}\left( \cdot \right)\otimes (\ket{0}\bra{0})^{\otimes|\Delta\setminus\Gamma|}_{\Delta\setminus\Gamma},\bra{0}^{\otimes |\Gamma\setminus \Delta|}_{\Gamma\setminus\Delta}\otimes\ket{0}^{\otimes|\Delta\setminus\Gamma|}_{\Delta\setminus\Gamma}\otimes I_{\Gamma\cap\Delta}\right).
\end{align*}
We write this pair as $\sem{R}_{\Gamma}=(\mathcal{E},H)$. Let $\qr$ be a qubit variable that does not appear in $\Gamma,\Delta$, $S$ or $S'$. Consider the context:
\begin{align*}
C_0^{\Gamma\rightarrow \Delta}\ \triangleq \ \tb{qcase} \ \qr \ (0\rightarrow [\cdot], 1 \rightarrow R ).
\end{align*}
Then,
\begin{align*}
 \sem{C_0[S]}_{\qr,\Gamma} & = \overline{\tb{qcase}}_{\qr}[(\mathcal{C},F),(\mathcal{E},J)] = (\mathcal{C}',F'), \\
 \sem{C_0[S']}_{\qr,\Gamma} & = \overline{\tb{qcase}}_{\qr}[(\mathcal{D},G),(\mathcal{E},J)] = (\mathcal{D}',G'),
\end{align*}
where
\begin{align*}
 \mathcal{C}' =& \mathcal{P}_0^\qr \otimes \mathcal{C} + \mathcal{P}_1^\qr \otimes \mathcal{E} + (\ket{0}\bra{0}_\qr \otimes F)(\cdot)(\ket{1}\bra{1}_\qr\otimes J^{\dag})
 +(\ket{1}\bra{1}_\qr \otimes J)(\cdot)(\ket{0}\bra{0}_\qr\otimes F^{\dag}), \\
 \mathcal{D}' =& \mathcal{P}_0^\qr \otimes \mathcal{D} + \mathcal{P}_1^\qr \otimes \mathcal{E} + (\ket{0}\bra{0}_\qr \otimes G)(\cdot)(\ket{1}\bra{1}_\qr\otimes J^{\dag})
 +(\ket{1}\bra{1}_\qr \otimes J)(\cdot)(\ket{0}\bra{0}_\qr\otimes G^{\dag}).
\end{align*}
Then $\mathcal{C}'\neq \mathcal{D'}$ as, for instance, with $\ket{\varphi}\in \Hilb_{\Gamma}$ such that $F\ket{\varphi}\neq G\ket{\varphi}$, we have
\begin{align*}
\mathcal{C}'\left((\ket{0}_\qr\otimes \ket{\varphi})(\bra{1}_\qr \otimes \bra{0}_{\Gamma}^{\otimes |\Gamma|})\right)
&= (\ket{0}_\qr\otimes F\ket{\varphi})(\bra{1}_\qr \otimes \bra{0}_{\Delta}^{\otimes |\Delta|}) \\
& \neq (\ket{0}_\qr\otimes G\ket{\varphi})(\bra{1}_\qr \otimes \bra{0}_{\Delta}^{\otimes |\Delta|}) \\
& = \mathcal{D}'\left((\ket{0}_\qr\otimes \ket{\varphi})(\bra{1}_\qr \otimes \bra{0}_{\Gamma}^{\otimes |\Gamma|})\right)\!.
\end{align*}
Therefore we can proceed as in the first case. The resulting context will be compatible with $(\Gamma;S)$ and $(\Gamma;S')$, and therefore $(\Gamma;S)\not\approx (\Gamma;S')$.
\end{proof}

Therefore, the denotational semantics exactly captures the observable behavior of programs. This further justifies the use of coherent quantum operations for the denotational semantics.
We end this section with a few examples.

\begin{exa}
Consider the statements
\begin{align*}
 S_1 \ &\triangleq \ \tb{skip}, \\
S_2 \ &\triangleq\ \tb{new qbit }\q;\ \tb{discard}\ \q ,\\
 S_3\ & \triangleq \ \q {*}{=} I.
\end{align*}
First, the judgments $\emptyset \vdash S_1 \triangleright \emptyset$ and $\emptyset\vdash S_2 \triangleright \emptyset$ are derivable, and we have $\sem{S_2}_\emptyset =  (Tr_\q,\bra{0}_\q)\circ (\ket{0}\bra{0}_\q,\ket{0}_\q)= (\mathcal{I}_{\emptyset},I_{\emptyset}) = \sem{S_1}_{\emptyset}$.
Therefore by Theorem~\ref{thm:fullabstraction}, $(\emptyset;S_1)\approx (\emptyset;S_2)$: initializing and deleting a qubit is equivalent to doing nothing.

Second, the judgments $\q \vdash S_1 \triangleright \q$ and $\q\vdash S_3 \triangleright \q$ are derivable, and we have $\sem{S_1}_{\q} = (\mathcal{I}_{\q},I_{\q}) = \sem{S_3}_{\q}$. Therefore, $(\q;S_1)\approx (\q;S_3)$. This example illustrates why it is necessary to specify the input environment when stating that two programs are observationally equivalent. In particular, there exists no environment $\Gamma$  such that $(\Gamma;S_2)\approx (\Gamma;S_3)$—at least one of these programs is ill formed.

\end{exa}

\begin{exa}[Discarding with respect to different bases]
We show that discard statements with respect to different bases are not observationally equivalent. Consider the following statement indexed by a unitary $U$, which discards the qubit $\q$:
 \[
  \text{D}_U \ \triangleq \ \q{*}{=}U ; \tb{discard}\ \q.
 \]
The judgment $\q\vdash D_U \triangleright \emptyset$ is derivable. We have:
\begin{align*}
 \sem{\text{D}_U}_{\q} &= (Tr_\q, \bra{0}_\q)\circ (U_\q(\cdot)U^{\dag}_\q,U_\q) =(Tr_\q, \bra{0}_\q U_\q).
\end{align*}
Therefore $(\q;\text{D}_U)\approx (\q;\text{D}_V)$ if and only if $U^{\dag}\ket{0}=V^{\dag}\ket{0}$, meaning that discards with respect to different bases are generally not equivalent. This is different than the usual framework of quantum operations, where unitary operations performed just before a discard can be ignored. Indeed, when $\sem{\text{D}_U}_{\q}\neq \sem{\text{D}_V}_{\q}$ only the transformation matrices differ, the quantum operations both being the usual traceout operation.
\end{exa}

\begin{exa}
We consider the coin flip statement (see Examples~\ref{ex:coin_op} and~\ref{ex:coin_den}), $\text{COIN} \triangleq  \tb{while} \ \q \tb{ do } \q{*}{=} H$. We have shown in Example~\ref{ex:coin_den} that $\sem{\text{COIN}}_{\q} = \left(Tr(\cdot)\ket{0}\bra{0}, \ket{0}\bra{0}\right)$. In fact, COIN is observationally equivalent to a discard/prepare program DP defined by $\text{DP} \ \triangleq \ \tb{discard} \ \q; \tb{new qbit }\q$. Indeed,
\begin{align*}
\sem{\text{DP}}_\q \!=\! (\ket{0}\bra{0},\ket{0})\!\circ\! (Tr_\q,\bra{0}_\q) \!=\!\left(Tr(\cdot)\ket{0}\bra{0},\ket{0}\bra{0}\right) \!=\! \sem{\text{COIN}}_{\q}.
\end{align*}
Therefore $(\q;\text{COIN})\approx (\q;DP)$. This is because repeatedly measuring a qubit until we obtain the state $\ket{0}$ has the same effect as discarding the qubit and preparing it again in the state~$\ket{0}$.
\end{exa}

\begin{exa}
\label{ex:interchange}
Consider the following statements:
\begin{align*}
S\ \triangleq \ &\tb{qcase }\q\ (0\rightarrow S_0,\ 1\rightarrow S_1); \ \tb{qcase }\q\ (0\rightarrow S_0',\ 1\rightarrow S_1'), \\
S'\ \triangleq \ &\tb{qcase }\q\ (0\rightarrow S_0;S_{0}',\ 1\rightarrow S_1;S_1'),
 \end{align*}
with judgments $\Gamma\vdash S_i \triangleright \Gamma'$ and $\Gamma' \vdash S_i' \triangleright \Gamma''$ being derivable, for $i \in \{0,1\}$, and
such that $\q$ does not appear in $S_0,S_0',S_1,S_1',\Gamma,\Gamma',\Gamma''$. Suppose that $ \forall i \in \{0,1\},$ $\sem{S_i}_\Gamma=(\mathcal{C}_i,F_i)$ and $\sem{S_i'}_{\Gamma'}=(\mathcal{C}_i',F_i')$.
Then $\q,\Gamma \vdash S \triangleright\q,\Gamma''$ and $\q,\Gamma \vdash S' \triangleright \q,\Gamma''$ are derivable, and:
\begin{align*}
 \sem{S}_{\q,\Gamma}&=\overline{\tb{qcase}}_{\q}\left[(\mathcal{C}_0',F_0'),(\mathcal{C}_1',F_1')\right] \circ \overline{\tb{qcase}}_\q \left[(\mathcal{C}_0,F_0),(\mathcal{C}_1,F_1)\right]\!, \\
 \sem{S'}_{\q,\Gamma}&=\overline{\tb{qcase}}_\q\left[(\mathcal{C}_0',F_0')\circ (\mathcal{C}_0,F_0),(\mathcal{C}_1',F_1')\circ (\mathcal{C}_1,F_1)\right]\!.
\end{align*}
By Proposition~\ref{prop:factorize_qcase}, $\sem{S}_{\q,\Gamma}=\sem{S'}_{\q,\Gamma}$, and therefore the two programs are observationally equivalent: $((\q,\Gamma);S)\approx ((\q,\Gamma);S')$.
\end{exa}

\begin{exa}
\label{ex:qcase_of_coin}
We aim to study the behavior of the $\tb{qcase}$ primitive when its branches contain probabilistic statements. Consider the following statement implementing a single coin flip:
\begin{align*}
 \text{COIN}_1 \triangleq \tb{new qbit }\q\ ;\ \q{*}{=} H \ ;\ \tb{meas }\q \ (0\rightarrow \tb{skip},1\rightarrow \tb{skip}).
\end{align*}
The qubit $\q$ will finish in state $\ket{0}$ or $\ket{1}$, each with probability $\frac{1}{2}$.
Then, consider the statement consisting of a $\tb{qcase}$ where each branch contains one coin flip:
\[
 \text{QCOIN}_1 \triangleq \tb{qcase}\ \p\ (0\rightarrow\text{COIN}_1 , 1 \rightarrow \text{COIN}_1).
\]
Let us compare the semantics of COIN$_1$ and QCOIN$_1$.
On the one hand, we can derive $\p \vdash \text{COIN}_1 \triangleright \p,\q$, and show that:
\begin{align*}
 \sem{\text{COIN}_1}_{\p}&=\left(\mathcal{I}_\p\otimes \frac{1}{2}\left(\ket{0}\bra{0}+\ket{1}\bra{1}\right),\frac{1}{\sqrt{2}}\ket{0}\right) \\
 & = \left(\begin{bmatrix}
          a&b\\
          c&d
\end{bmatrix}
\mapsto
\begin{bmatrix}
 \sfrac{a}{2}&0&\sfrac{b}{2}&0 \\
 0&\sfrac{a}{2}&0&\sfrac{b}{2}\\
 \sfrac{c}{2}&0&\sfrac{d}{2}&0\\
 0&\sfrac{c}{2}&0&\sfrac{d}{2}
\end{bmatrix}\ ,\
\begin{bmatrix}
 \sfrac{1}{\sqrt{2}} & 0 \\
 0&0\\
 0&\sfrac{1}{\sqrt{2}}\\
 0&0
\end{bmatrix}
\right)\!.
\end{align*}
And on the other hand, we can derive $\p\vdash \text{QCOIN}_1 \triangleright \p,\q$ and show that:
\begin{align*}
 \sem{\text{QCOIN}_1}_{\p} & =\overline{\tb{qcase}}_\p\left[\sem{\text{COIN}_1}_\p,\sem{\text{COIN}_1}_\p\right] \\
&= \left(\begin{bmatrix}
          a&b\\
          c&d
\end{bmatrix}
\mapsto
\begin{bmatrix}
 \sfrac{a}{2}&0&\sfrac{b}{2}&0 \\
 0&\sfrac{a}{2}&0&0\\
 \sfrac{c}{2}&0&\sfrac{d}{2}&0\\
 0&0&0&\sfrac{d}{2}
\end{bmatrix}\ ,\
\begin{bmatrix}
 \sfrac{1}{\sqrt{2}} & 0 \\
 0&0\\
 0&\sfrac{1}{\sqrt{2}}\\
 0&0
\end{bmatrix}
\right)\!.
\end{align*}

Therefore $\sem{\text{QCOIN}_1}_\p \neq\sem{\text{COIN}_1}_\p$, and by Theorem~\ref{thm:fullabstraction}, $(\p;\text{QCOIN}_1)\not\approx (\p;\text{COIN}_1)$.
This might seem surprising, because QCOIN$_1$ is nothing but two instances of COIN$_1$ in superposition. However, this is a consequence of the probabilistic behavior of COIN$_1$. Indeed, the statement QCOIN$_1$ allows for superposition of both outcomes of the coin toss—one occurring in each branch, leading to a total of 4 possible combinations. By contrast, a single copy of COIN$_1$ has 2 possible outcomes. Therefore, programs whose statements are $S$ on the one hand, and $\tb{qcase} \ \q \ (0\rightarrow S, \ 1\rightarrow S)$ on the other hand are generally not observationally equivalent.
\end{exa}

\section{Physical Interpretation and Design Choices}
\label{section:interpretation_design}

The introduction of a default transition is inspired by a series of works in quantum foundations \cite{oi2003interference,aaberg2004subspace,dong2021controlled,universal_control,grenoble_coherent_control} on quantum control of general quantum transformations. Typically, such transformations are represented as sets of possible evolutions—commonly referred to as Kraus operators.
While the controlled version of a quantum operation can be defined for a given Kraus representation, different choices of Kraus representations will yield different results~\cite{badescu}. Various extensions of quantum operations have been introduced to this end.
In particular, \emph{pinned Kraus operators}, introduced in \cite{universal_control}, consist in pinpointing a specific Kraus operator in order to unambiguously define quantum control. Similarly we identify, in the operational semantics,  a default transition that allows the definition of the semantics of $\tb{qcase}$.

In \cite{dong2021controlled}, the authors show that quantum control can be achieved by associating an additional complex number $\nu_i$ to each possible evolution (i.e., each Kraus operator), such that $\sum_i |\nu_i|^2 = 1$. In our setting, we consider a particular case where these amplitudes are restricted to $\{0,1\}$, which amounts to selecting a default evolution—the only one for which the amplitude is $1$.

The use of bits for the amplitudes is a deliberate design choice for our language. It ensures deterministic transitions when no inputs are provided to the programs. Note that one could have opted for alternative default transitions—for example, by swapping the roles of (D$_0$) and (D$_1$). More generally, other admissible choices exist that may not yield determinism but still preserve essential properties of the language, such as universality and adequacy (see Sections~\ref{section:universality_all} and~\ref{section:adequacy}).
However in the presence of recursion, some choices should be avoided; for example, were we to select (W$_1$) as the default transition, termination on the empty input would fail.

 In quantum foundations, other approaches to the quantum control of general quantum evolutions have been considered, in particular by extending the domain and co-domain with an extra state $\ket{\text{vac}}$ representing the empty input \cite{kristjansson_1,kristjansson}. We follow this alternative interpretation for the denotational semantics of the language (see Section \ref{section:denotational_semantics}) and show the adequacy of the two semantics, underlining the equivalence of these various approaches to quantum control of general quantum evolutions.

Implementation of quantum control on arbitrary quantum operations has been discussed in \cite{grenoble_coherent_control}. In particular quantum computers with ``flying qubits'' (photon-based or electron-based for instance) offers natural implementations of quantum control of arbitrary quantum operation. Intuitively, given a device $d_0$ implementing $S_0$ and a device $d_1$ implementing $S_1$, $\tb{qcase}\ \q\ (0\to S_0, 1\to S_1)$ can be implemented via a polarizing beam splitter, that depending on the polarization of the particle (i.e., the state of $\q$) routes the particle to $d_0$ or to $d_1$. Moreover, in such systems, a while loop can be implemented by combining the implementation of the body of the loop with a measurement device (to decided whether the condition is satisfied) and a feedback loop made for instance by means of optical fibers.

\section{Conclusion}
\label{section:ccl}
We have developed a high-level programming language combining classical and quantum control in the presence of while loops, and succeeded in developing well-defined semantics, thus resolving a problem that has been open for a decade.
The language has strong properties: universality, adequacy, and full abstraction, which demonstrate the relevance of the semantic approaches used (default transitions and coherent quantum operations).

Future work concerns the extension of this study to a higher-order language and the compilation of programs to lower-level quantum models.
There also remain some open questions regarding the speed-up offered by coherent control in the non-unitary setting.
For instance, quantum switch~\cite{chiribella} was shown to enable an exponential separation in comparison to models with classically ordered gates; and importantly, unitary gates alone are insufficient to demonstrate this advantage~\cite{kristjansson_exponential}.
An interesting question would be whether the \tb{qcase} construct enables a similar speedup.

\bibliographystyle{alphaurl}
\bibliography{refs.bib}

\end{document}